\documentclass[11pt]{article}
\usepackage[legalpaper]{geometry}

\usepackage{cancel}
\usepackage{amsmath}    
\usepackage{amssymb}    
\usepackage{amsthm}
\usepackage{mathtools}
\usepackage{amsfonts}
\usepackage{stmaryrd}
\usepackage{enumerate}
\usepackage{bbm}        
\usepackage{xparse}
\usepackage{xcolor}
\usepackage{mathrsfs} 
\usepackage{graphicx}
\usepackage{tcolorbox}
\usepackage{enumitem} 
\usepackage{makecell}
\usepackage{xstring} 

\usepackage{glossaries}

\usepackage{caption}
\usepackage{algorithm}
\usepackage{algpseudocode}
\floatname{algorithm}{Protocol}

\usepackage{tikz}
\usetikzlibrary{arrows.meta,positioning,fit,calc,shapes.multipart}

\usepackage{subcaption}

\usepackage{ifthen}
\newboolean{darkMode}
\setboolean{darkMode}{false}
\newboolean{comments}
\setboolean{comments}{true}

\definecolor{magenta}{RGB}{255, 0, 255}
\definecolor{rotemcolor}{RGB}{255, 0, 0}
\definecolor{thomsycolor}{RGB}{200, 30, 230}

\numberwithin{equation}{section}

\newtheorem{theorem}{Theorem}
\newtheorem*{mainrestate}{Theorem~\ref{thm:main}}
\newtheorem*{virtualinputrestate}{Theorem~\ref{thm:intro_virtual_input_hiding}}

\newtheorem{lemma}{Lemma}[section]

\newtheorem{proposition}[lemma]{Proposition}
\newtheorem{corollary}[lemma]{Corollary}

\theoremstyle{definition}
\newtheorem{definition}[lemma]{Definition}
\newtheorem{assumption}{Assumption}

\theoremstyle{remark}
\newtheorem*{remark}{Remark}

\ifthenelse{\boolean{darkMode}}
    {
        \pagecolor[rgb]{0,0,0}\color[rgb]{1,1,1}
        \usepackage{hyperref}

    }
    {
        \usepackage[
        colorlinks=true,
        linkcolor=black,       
        citecolor=black,       
        urlcolor=blue,         
        pdfborder={0 0 0}      
        ]{hyperref}
    }

\begin{document}


\newcommand{\preimage}{x}
\newcommand{\Preimage}{\mathcal{X}}
\newcommand{\image}{y}
\newcommand{\Image}{\mathcal{Y}}
\newcommand{\key}{\kappa}
\newcommand{\Key}{\mathcal{K}}
\newcommand{\trapdoor}{\tau}
\newcommand{\tcfunc}{f}
\newcommand{\Tcfunc}{\mathcal{F}}

\newcommand{\preimagerv}{X}

\newcommand{\security}{\lambda}
\newcommand{\eval}{\mathsf{Eval}}
\newcommand{\gen}{\mathsf{Gen}}
\newcommand{\inv}{\mathsf{Inv}}


\newcommand{\ket}[1]{\left|#1\right\rangle}
\newcommand{\bra}[1]{\left\langle#1\right|}
\newcommand{\dyad}[1]{\ket{#1}\!\!\bra{#1}}
\newcommand{\id}{\mathbb{I}}

\newcommand{\tr}{\mathrm{tr}}

\newcommand{\state}{\psi}
\newcommand{\transState}{{\state_{\trans}}}
\newcommand{\transStaterv}{{\state_{_\transRV}}}
\newcommand{\sigmax}{{\mathtt{X}}}
\newcommand{\sigmay}{{\mathtt{Y}}}
\newcommand{\sigmaz}{\mathtt{Z}}

\newcommand{\measurementAngle}[1]{\theta^{(#1)}}
\newcommand{\anticommSquare}{\alpha_{\mathrm{ac}}}

\newcommand{\bellOutcomeRV}{B}
\newcommand{\bellOutcomeRVAt}[1]{\bellOutcomeRV_{#1}}

\newcommand{\blockIndex}{i}
\newcommand{\blockIndexRV}{I}
\newcommand{\blockProjector}[1]{\Pi_{#1}}
\newcommand{\angleBlockSet}[1]{\mathsf{S}_{#1}}
\newcommand{\biasPairSet}[1]{\mathsf{G}_{#1}}
\newcommand{\goodPairSet}{\Gamma_{*}}
\newcommand{\blockProbability}[2]{p_{#1}^{(#2)}}

\newcommand{\pacc}{p_{\mathrm{acc}}}

\newcommand{\postPhaseState}[1]{\rho_{\proverReg\eveReg}^{#1}}
\newcommand{\postPhaseStateP}[1]{\rho_{\proverReg}^{#1}}
\newcommand{\blockPostPhaseState}[2]{\rho_{\proverReg\eveReg}^{#1,#2}}
\newcommand{\blockPostPhaseStateP}[2]{\rho_{\proverReg}^{#1,#2}}

\newcommand{\continuingPrefixes}{S_\mathsf{Cont}}
\newcommand{\goodPrefixes}{S_\mathsf{Good}}
\newcommand{\concDecomm}{\chi}
\newcommand{\concRand}{\varrho}
\newcommand{\predictorGuess}{G}
\newcommand{\goodBlockEntropy}{h_*}

\newcommand{\entropyConst}{c_0}
\newcommand{\entropySqrtConst}{c_{1/2}}


\newcommand{\vne}{H}
\renewcommand\Pr{\mathrm{Pr}}
\newcommand{\indicator}[1]{\mathbbm{1}_{#1}}


\newcommand{\Y}{\mathcal{Y}}
\newcommand{\X}{\mathcal{X}}
\newcommand{\F}{\mathcal{F}}
\newcommand{\N}{\mathbb{N}}
\newcommand{\A}{\mathcal{A}}
\newcommand{\B}{\mathcal{B}}
\newcommand{\E}{\mathbb{E}}
\renewcommand{\O}{\mathcal{O}}


\newcommand{\moya}{\cite{MYowf23}}

\newcommand{\adversary}{\mathscr{A}}
\newcommand{\qadversary}{\mathscr{Q}}
\newcommand{\guesser}{\mathscr{G}}
\newcommand{\verifier}{\mathcal{V}}
\newcommand{\prover}{\mathcal{P}}

\newcommand{\trans}{t}
\newcommand{\transRV}{{T}}
\newcommand{\Trans}{\mathcal{T}}
\newcommand{\rand}{r}
\newcommand{\hadx}{d}

\newcommand{\bellChallenge}{y}
\newcommand{\bellChallengeRV}{Y}
\newcommand{\BellChallenge}{\mathcal{Y}}
\newcommand{\bellMeasurement}{M}
\newcommand{\bellOutcome}{\bit}
\newcommand{\BellOutcome}{\B}
\newcommand{\virtualInput}{x}
\newcommand{\VirtualInput}{\X}
\newcommand{\virtualInputRV}{X}
\newcommand{\virtualOutput}{a}
\newcommand{\VirtualOutput}{\A}
\newcommand{\virtualOutputRV}{A}

\newcommand{\acceptPredicate}{\mathsf{Acc}}

\newcommand{\verifierReg}{V}
\newcommand{\proverReg}{P}
\newcommand{\eveReg}{{E}}

\newcommand{\hashfam}{\mathcal{H}}
\newcommand{\hashfunc}{h}
\newcommand{\Hashval}{\mathcal{Z}}
\newcommand{\hashval}{z}
\newcommand{\oracle}{\mathcal{G}}

\newcommand{\negl}{\mathrm{negl}}
\newcommand{\poly}{\mathrm{poly}}
\newcommand{\imap}{\xi}
\newcommand{\omap}{\alpha}

\newcommand{\chshopt}{\omega^*}

\newcommand{\hashIndex}{j}
\newcommand{\hashNum}{m}
\newcommand{\hashEpsilon}{\varepsilon}
\newcommand{\hashRangeSize}{k}


\newcommand{\bgkmProtocol}{\hyperref[protocol:bgkm_protocol]{Protocol~\ref*{protocol:bgkm_protocol}}}
\newcommand{\seqProtocol}{\hyperref[protocol:sequential_rand_certification]{Protocol~\ref*{protocol:sequential_rand_certification}}}
\newcommand{\iidProtocol}{\hyperref[protocol:single_round]{Protocol~\ref*{protocol:single_round}}}

\newcommand{\Ber}{\mathrm{Ber}}
\newcommand{\generationRoundBernoulli}{\beta}
\newcommand{\winRV}{W}
\newcommand{\generationRoundRV}{G}

\newcommand{\contProbWidth}{\Delta_\pce}
\newcommand{\expectedContProb}{\pce_{\mathrm{exp}}}

\newcommand{\winProbWidth}{\Delta_\omega}
\newcommand{\expectedWinProb}{\omega_{\mathrm{exp}}}

\newcommand{\expectedEpsilon}{\varepsilon_{\mathrm{exp}}}
\newcommand{\epsilonWidth}{\Delta_\varepsilon}

\newcommand{\smoothingParam}{\varepsilon_{\mathrm{s}}}

\newcommand{\testMeasurement}{\Delta}
\newcommand{\testProjector}[1]{\testMeasurement_{#1}}
\newcommand{\testOutcome}{D}
\newcommand{\simplifiedOutcome}{\Pi}


\newcommand{\bc}{\mathcal{BC}}
\newcommand{\shbcscheme}{\mathcal{BC}_{\text{SH}}}

\newcommand{\sender}{\mathcal{S}}
\newcommand{\receiver}{\mathcal{R}}
\newcommand{\open}{\mathsf{Open}}
\newcommand{\commit}{\mathsf{Commit}}

\newcommand{\Decomm}{\X}
\newcommand{\decomm}{x}
\newcommand{\tape}{s}

\newcommand{\bindExp}{\mathsf{CompBindExp}}
\newcommand{\qbindExp}{\mathsf{SumBindExp}}
\newcommand{\hideExp}{\mathsf{HideExp}}

\newcommand{\bit}{b}
\newcommand{\bchallenge}{{\bit_{\mathrm{challenge}}}}

\newcommand{\tcfbc}{\bc_{\mathsf{tcf}}}

\newcommand{\rotag}{\mathsf{RO}}
\newcommand{\robc}{\bc_{\rotag{}}}


\newcommand{\flag}{\mathrm{flag}}
\newcommand{\abort}{\mathrm{abort}}
\newcommand{\continue}{\mathrm{continue}}
\newcommand{\contEvent}{\mathtt{CE}}
\newcommand{\pce}{\gamma}
\newcommand{\accept}{\mathrm{accept}}

\newcommand{\pwin}{p_{\mathrm{win}}}
\newcommand{\bctrans}{{}\trans_{\mathrm{bc}}}
\newcommand{\protocolState}{\psi}

\newcommand{\angleThreshold}{\delta}
\newcommand{\projThreshold}{\eta}

\newcommand{\ymeps}{{\gamma_\textrm{my}}}


\newacronym{ppt}{PPT}{probabilistic polynomial-time}

\newcommand{\ppt}{\gls{ppt}}
\newcommand{\Ppt}{\Gls{ppt}}
\newcommand{\pptshort}{\glsentryshort{ppt}}

\newacronym{qpt}{QPT}{quantum polynomial-time}

\newcommand{\qpt}{\gls{qpt}}
\newcommand{\Qpt}{\Gls{qpt}}
\newcommand{\qptshort}{\glsentryshort{qpt}}

\newacronym{qrom}{QROM}{quantum random-oracle model}

\newcommand{\qrom}{\gls{qrom}}
\newcommand{\Qrom}{\Gls{qrom}}
\newcommand{\qromshort}{\glsentryshort{qrom}}

\newacronym{tcf}{TCF}{trapdoor claw-free function}

\newcommand{\tcf}{\gls{tcf}}
\newcommand{\Tcf}{\Gls{tcf}}
\newcommand{\tcfs}{\glspl{tcf}}
\newcommand{\Tcfs}{\Glspl{tcf}}
\newcommand{\tcfshort}{\glsentryshort{tcf}}
\newcommand{\tcfsshort}{\glsentryshortpl{tcf}}

\newacronym{shbc}{SHBC}{statistically-hiding bit commitment}

\newcommand{\shbc}{\gls{shbc}}
\newcommand{\Shbc}{\Gls{shbc}}
\newcommand{\shbcs}{\glspl{shbc}}
\newcommand{\Shbcs}{\Glspl{shbc}}
\newcommand{\shbcshort}{\glsentryshort{shbc}}
\newcommand{\shbcsshort}{\glsentryshortpl{shbc}}

\newacronym{pq}{PQ}{post-quantum}

\newcommand{\pq}{post-quantum}
\newcommand{\Pq}{Post-quantum}
\newcommand{\pqshort}{post-quantum}

\newacronym{owf}{OWF}{one-way function}

\newcommand{\owf}{\gls{owf}}
\newcommand{\Owf}{\Gls{owf}}
\newcommand{\owfs}{\glspl{owf}}
\newcommand{\Owfs}{\Glspl{owf}}
\newcommand{\owfshort}{\glsentryshort{owf}}
\newcommand{\owfsshort}{\glsentryshortpl{owf}}

\newacronym{owp}{OWP}{one-way permutation}

\newcommand{\owp}{\gls{owp}}
\newcommand{\Owp}{\Gls{owp}}
\newcommand{\owps}{\glspl{owp}}
\newcommand{\Owps}{\Glspl{owp}}
\newcommand{\owpshort}{\glsentryshort{owp}}
\newcommand{\owpsshort}{\glsentryshortpl{owp}}

\newacronym{poq}{PoQ}{proof of quantumness}

\newcommand{\poq}{\gls{poq}}
\newcommand{\Poq}{\Gls{poq}}
\newcommand{\poqs}{\glspl{poq}}
\newcommand{\Poqs}{\Glspl{poq}}
\newcommand{\poqshort}{\glsentryshort{poq}}
\newcommand{\poqsshort}{\glsentryshortpl{poq}}

\newacronym{eat}{EAT}{entropy accumulation theorem}

\newcommand{\eat}{\gls{eat}}
\newcommand{\Eat}{\Gls{eat}}
\newcommand{\eatshort}{\glsentryshort{eat}}

\newacronym{qccc}{QCCC}{quantum computation with classical communication}

\newcommand{\qccc}{\gls{qccc}}
\newcommand{\Qccc}{\Gls{qccc}}
\newcommand{\qcccshort}{\glsentryshort{qccc}}

{
    \title{Randomness Certification from\\Statistically Hiding Commitments}
    \author{
        Ilya Merkulov$^{1}$\thanks{Email: \texttt{ilya.merkulov@weizmann.ac.il}}
        \quad and \quad
        Rotem Arnon$^{2}$
        \\
        \\
        \normalsize $^{1}$The Center for Quantum Science and Technology, \\
        \normalsize Weizmann Institute of Science, Rehovot, Israel
        \\
        \normalsize $^{2}$School of Computer and Communication Sciences, \\
        \normalsize EPFL, Lausanne, Switzerland
    }
    \date{\today}
    \maketitle
	\begin{abstract}
		Randomness certification aims to guarantee that the output of an untrusted quantum device contains genuine randomness that remains unpredictable even given quantum side information held by an adversary.
		For a classical verifier interacting with a single quantum device, such a guarantee cannot rely on the physical separation underlying standard Bell tests and must instead follow from computational restrictions on the prover.
		Building on the framework of computational Bell inequalities~\cite{cbi}, we distinguish classical and quantum virtual-input hiding, showing that the former suffices for proofs of quantumness, whereas the latter enables randomness certification against quantum side information.

		Our main technical contribution is to establish quantum virtual-input hiding for a commitment-based protocol derived from the proof of quantumness of Morimae and Yamakawa~\cite{MYowf23}.
		We show that statistical hiding together with ordinary computational binding against quantum polynomial-time senders suffices, without requiring stronger quantum binding notions such as sum-binding or collapse-binding.
		This yields a randomness-certification protocol from post-quantum statistically hiding bit commitments.
		In particular, the NOVY commitment instantiated from any post-quantum one-way permutation gives randomness certification in the standard model, albeit with an inefficient classical verifier in our framework.

		We make three instantiations explicit.
		Post-quantum one-way permutations yield a perfectly hiding commitment with binding against quantum polynomial-time senders and hence randomness certification with an inefficient verifier.
		Post-quantum trapdoor claw-free functions additionally provide trapdoor inversion and therefore an efficient classical verifier.
		In the quantum random-oracle model, a hash-based commitment yields randomness certification without additional computational assumptions, again with an inefficient verifier.
		These instantiations separate the cryptographic properties sufficient for randomness certification from the additional structure required for efficient verification within our commitment-based framework.
		\glsresetall
	\end{abstract}
}

\newpage
\footnotesize
\tableofcontents
\normalsize

\newpage

\section{Introduction}
\label{sec:introduction}

Randomness is a fundamental resource in modern cryptography.
Encryption schemes, digital signatures, authentication protocols, and many other cryptographic primitives rely on fresh secret randomness for their security~\cite{rolland2015randomness}.
The ability to generate trustworthy randomness is therefore a fundamental requirement for cryptographic security~\cite{quantum_randomness}.

At first glance, randomness generation may appear to be a purely statistical problem: produce a sequence of bits that ``looks random.''
This is not sufficient for cryptographic applications.
A public sequence may pass standard statistical tests~\cite{statistical_test} while being perfectly predictable to anyone who possesses a copy.
The relevant requirement is instead that the generated output remain unpredictable even when conditioned on all information available to a potentially adversarial observer.

Device-independent randomness certification addresses precisely this stronger requirement.
It aims to certify entropy in the output of an untrusted device, conditioned on an adversary's side information, without relying on a detailed model of the device's internal implementation~\cite{pironio2010random,vidick_quantumdice,yuen2014expansion,acin2016optimal,acin2016certified,randomness_expansion,arnonfriedman2020device,hahn2026randomness}.
Bell inequalities~\cite{bell64,scarani2019bell_nonlocality,brunner2014bell} provide the canonical mechanism for obtaining such guarantees.
In a standard Bell experiment, a user interacts classically with spatially separated quantum devices, while an adversary may have prepared the devices and may possess quantum side information correlated with their internal state.
The essential physical assumption is that the devices cannot communicate after receiving their respective inputs, enforcing the no-signaling condition on their correlations.
Under this assumption, a Bell violation, serving as a \poq{}, rules out a local classical model~\cite{bell64,scarani2019bell_nonlocality}.
Building on the characterization of local hidden-variable models underlying Fine's theorem~\cite{fine1982hidden}, local classical correlations can more generally be expressed as probabilistic mixtures of deterministic local assignments, in which the outcomes of all possible measurements are predetermined~\cite{brunner2014bell}.
Thus, a Bell violation rules out such a local classical explanation of the observed correlations and enables one to infer genuine randomness in the device outputs, even conditioned on quantum side information held by an adversary~\cite{scarani2019bell_nonlocality,security_proofs}.

The cryptographic single-device setting was introduced by Brakerski \emph{et al.}~\cite{brakerski2021cryptographic}, who showed that a classical verifier interacting with a single computationally bounded quantum device can certify quantum behavior and, in suitable protocols, certify randomness under \pq{} cryptographic assumptions.
This setting, considered also in the present work, operates in the \qccc{} model, where the device may perform quantum computations internally, but all communication with the classical verifier remains classical.
It differs fundamentally from the standard device-independent setting described above.
A single device receives the entire interaction transcript, so there is no physical separation that prevents it from learning all verifier messages.
Instead, \pq{} cryptography imposes computational restrictions on the information available to the device, playing a role analogous to the information constraints enforced by spatial separation in a Bell experiment.
Mahadev's classical-verification protocol~\cite{mahadev2018classical} provided a closely related foundational application of \pq{} cryptography to the classical verification of quantum computation.\footnote{
    Randomness certification has been demonstrated experimentally both from nonlocal correlations~\cite{experimental_randomness0,experimental_randomness1,experimental_randomness3} and from computational hardness~\cite{experimental_randomness2}.
    The computational Bell test of Kahanamoku-Meyer \emph{et al.}~\cite{kcvy} has also been implemented experimentally~\cite{zhu2023interactive}.
    Although originally introduced and experimentally demonstrated as a \poq{}, our previous work~\cite{cbi} showed that this protocol also certifies randomness against quantum side information.
    Consequently, the existing experiment can also be interpreted as an experimental demonstration of computational randomness certification.
}
Subsequent works developed a range of computational analogues of Bell tests and related verification procedures~\cite{kcvy,mahadev2022efficient,brakerski2023simple,kalai2022compiled,compiledtrapdoor2024,cbi,tilted_mehta,baroni2026seqnpa,cui2025seqnpa}.

The framework of computational Bell inequalities~\cite{cbi} provides a connection between the standard Bell setting and these cryptographic single-device protocols.
In a Bell experiment, spatial separation prevents one device from learning the measurement input given to the other.
In the single-device setting, an analogous role can be played by a \emph{virtual input}: a classical variable determined by the interaction transcript that, when cunningly chosen, is computationally hard for the device to recover.
Together with a corresponding virtual output, the virtual input defines a mapping from the single-device interaction to a Bell scenario, allowing the resulting correlations to be analyzed using Bell inequalities.

\begin{figure}[t]
    \centering
    \includegraphics[width=0.8\textwidth]{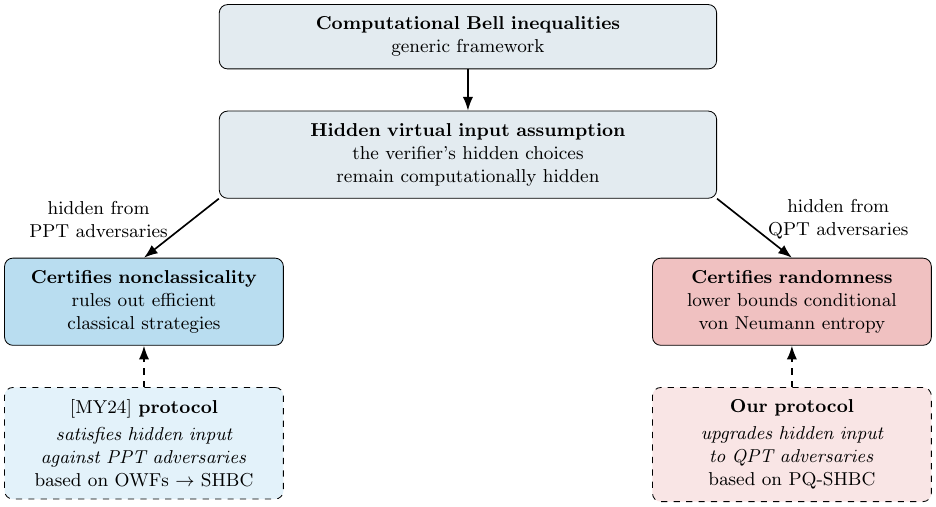}
    \caption{
        \footnotesize
        The hidden virtual input is the computational analogue of a Bell measurement setting.
        Classical hiding excludes efficient classical strategies and yields a \poq{}, whereas quantum hiding also protects against the prover's residual quantum state and enables randomness certification.
        Our main technical contribution is to establish this stronger guarantee for a commitment-based proof-of-qubit protocol.
    }
    \label{fig:intro_diagram}
\end{figure}

\paragraph{Our contribution.}

In this work, we distinguish two natural notions of virtual-input hiding according to the computational model of the adversary.
The weaker notion requires the virtual input to remain hidden from \ppt{} classical algorithms, whereas the stronger notion requires it to remain hidden from \qpt{} algorithms.
In both cases, the adversary may use the interaction transcript together with any residual state retained from the protocol execution.
We show that this distinction is operationally meaningful within the computational Bell-mapping framework used here: the weaker classical notion suffices for \emph{\poqs{}}, whereas the stronger quantum notion is the additional condition that enables randomness certification against quantum side information.

Motivated by this distinction, we take a closer look at the commitment-based protocol for \poq{} of Morimae and Yamakawa~\cite{MYowf23}.
The original analysis establishes the classical hiding needed for a \poq{}, but does not by itself imply randomness certification against quantum side information.
We strengthen the protocol by requiring quantum virtual-input hiding, thereby obtaining a new randomness-certification protocol.
Specifically, we establish quantum virtual-input hiding assuming only a \shbc{} satisfying the standard computational binding property but against \qpt{} senders (rather than merely \ppt{} senders), without requiring stronger \pq{} notions such as sum-binding or collapse-binding that prevent a quantum sender from deciding only at the opening stage which bit to unveil.

At a high level, we use the simplified~\cite{kcvy} protocol from~\cite{brakerski2023simple} to map the structure induced by the commitment to a BB84 state whose basis is determined by the virtual input.
Crucially, we prove that this virtual input remains computationally hidden even from a \qpt{} prover with access to its residual quantum state.
Combined with the \emph{computational Bell inequalities} framework of~\cite{cbi}, this upgrades the underlying~\moya{} \poq{} construction to a randomness-certification protocol.

We further provide a direct analytical proof of randomness certification, combining approximate anticommutation of the prover's measurements with entropic uncertainty relations~\cite{eur} to lower bound the conditional von Neumann entropy of the prover's output given the protocol transcript and any quantum side information held by an adversary.

We then instantiate our construction in three settings.
First, the NOVY commitment from a \pq{} \owp{} yields randomness certification in the standard model, with an inefficient verifier that recovers the unique openings by exhaustive inversion.
Second, exact quantum-secure \tcfs{} with efficiently recognizable valid public keys yield an efficient classical verifier, since the trapdoor allows the verifier to recover and certify the relevant unique commitment openings.
Third, in the \qrom{}, we obtain randomness certification without additional computational assumptions, again with an inefficient verifier.
Together, these instantiations separate the cryptographic properties required for security from the additional algorithmic structure required for efficient verification, and identify efficient recovery and certification of unique commitment openings as the additional ingredient supplied by the trapdoor.

We now explain and discuss our contributions in more depth.

\subsection{The Required Commitment Primitive}
\label{subsec:post_quantum_binding}

The cryptographic primitive underlying our construction is a \shbc{} that satisfies computational binding against \qpt{} senders.

Informally, a bit commitment can be thought of as the digital analogue of placing a bit in a sealed envelope.
A sender first commits to a bit $\bit\in\{0,1\}$ without revealing its value to the receiver, and later opens the commitment to reveal the bit.
The two basic security requirements capture the properties expected from such a sealed envelope: \emph{hiding} prevents the receiver from learning the committed bit before it is opened, while \emph{binding} prevents the sender from changing its mind after committing.

Formally, a classical bit commitment scheme~$\bc=(\commit,\open)$ consists of a commit phase and an open phase between a sender and a receiver~\cite{Naor1991,shbc_from_owf,shbc_existence}.
During the commit phase, the sender chooses a bit~$\bit\in\{0,1\}$ and interacts with the receiver, producing a classical transcript~$\trans$.
At this point, the bit is committed but has not yet been revealed to the receiver.
During the open phase, the sender reveals a bit~$\bit$ together with a decommitment string~$\decomm$, which serves as evidence that this is a valid opening of the earlier commitment.
The receiver verifies the opening by evaluating
\begin{equation}
	\open(\trans,\bit,\decomm)\in\{\top,\bot\} \; ,
\end{equation}
where $\top$ indicates that the opening is accepted and $\bot$ that it is rejected.

Our protocol requires the underlying bit commitment to satisfy the following two security properties.

\paragraph{Statistical hiding.}

The statistical hiding property means that the commit phase reveals essentially no information about the committed bit: the receiver's view when an honest sender commits to~$0$ is statistically indistinguishable from its view when the sender commits to~$1$, even for a computationally \emph{unbounded} receiver.

Although statistical hiding is conventionally a security property of a commitment scheme, in our construction it is used to establish the completeness of the randomness-certification protocol, \iidProtocol{}.
Specifically, it ensures that the two commitment branches have comparable weight, so that the honest prover continues the protocol with non-negligible probability.

\paragraph{Post-quantum computational binding.}
\label{paragraph:double_opening_binding}

Throughout this work, we consider classical bit commitment schemes: the commit protocol, all communication, the commitment transcript and decommitment strings, and the opening verification algorithm are all classical.
The adversarial sender, however, may be quantum.

For classical commitment schemes, the basic notion of computational binding requires that no \ppt{} adversarial sender can produce valid openings of the same commitment to two distinct messages~\cite{goldreich2001foundations,katz_lindell}.
In the bit-commitment setting, this means that no \ppt{} sender can output a pair of classical decommitment strings~$(\decomm_0,\decomm_1)$ such that
\begin{equation}
	\label{eq:double_opening_condition}
	\open(\trans,0,\decomm_0)=\top
	\quad\wedge\quad
	\open(\trans,1,\decomm_1)=\top
	\;,
\end{equation}
except with negligible probability.
We require the same condition to hold against \qpt{} adversarial senders, and refer to this notion as \emph{\pq{} computational binding}.

\Pq{} computational binding provides a strictly weaker security guarantee than other \pq{} binding notions considered in the literature, such as sum-binding and collapse-binding~\cite{unruh16,spooner2023collapsing}.
In particular, it does not imply that the sender is committed to a definite classical value at the end of the commit phase.
Indeed~\cite{aru14} shows that, relative to an oracle, a classical commitment scheme may be computationally binding against quantum senders while nevertheless allowing a quantum sender to choose which message to open only after the commitment has been produced.

This weakness is not problematic for our application.
Our proof never requires the committed bit to be fixed at the end of the commit phase.

Instead, we show that any \qpt{} algorithm that predicts the virtual input with non-negligible advantage can be transformed into a \qpt{} sender that outputs, in a single execution, a pair of classical decommitment strings~$(\decomm_0,\decomm_1)$ satisfying the computational binding condition in Equation~\eqref{eq:double_opening_condition}, which is exactly the event ruled out by \pq{} computational binding.
Therefore, stronger \pq{} binding notions, such as sum-binding or collapse-binding,\footnote{Unlike for general commitments, these notions are equivalent when restricted to \emph{bit} commitments~\cite{spooner2023collapsing}.} which better capture the spirit of bit commitment against quantum senders, are not required for our proof.
Our cryptographic assumption is thus simply the existence of a \shbc{} satisfying \pq{} computational binding.

This should be contrasted with the assumption underlying the original \moya{} \poq{}.
There, the existence of \owfs{} suffices, since \shbc{} schemes satisfying computational binding against \ppt{} senders can be constructed generically from \owfs{}~\cite{shbc_from_owf}.
In our setting, however, the corresponding binding guarantee must hold against \qpt{} senders, and no analogous generic construction from arbitrary \pq{} \owfs{} is known.
Recent independent works establish \pq{} security of the NOVY commitment protocol~\cite{novy} from \pq{} \owp{}.
Chen, Wang, and Xie~\cite{chen2026minicrypt} prove quantum sum-binding security and extend their construction to efficiently verifiable one-way puzzles with nearly uniform puzzle distributions.
Dupuis, Ligez, Medhioub, Salvail, and Thibodeau~\cite{dupuis2026postquantum} prove \pq{} open-once binding and, using the unique-opening structure of NOVY, also derive sum-binding and collapse-binding.
Since the NOVY commitment is perfectly hiding and both works establish binding guarantees stronger than ordinary \pq{} computational binding, their results provide commitment schemes satisfying the assumptions of our theorem.
Thus, \pq{} \owp{} suffice for our randomness-certification construction, although the verifier in this instantiation is inefficient because recovering the valid openings requires inverting the permutation.
Whether arbitrary \pq{} \owfs{} suffice remains open.
We therefore state our assumption directly at the level of the commitment primitive and discuss concrete instantiations in Section~\ref{sec:instantiations}.

\subsection{Classical versus Quantum Hidden Virtual Inputs}
\label{subsec:intro_hidden_virtual_inputs}

A broad class of single-device verification protocols can be expressed in the canonical two-phase form studied in simple tests of quantumness and computational Bell inequalities~\cite{brakerski2023simple,cbi}.
In such protocols, a prover interacts with a classical verifier in two phases and attempts to convince the verifier that it possesses quantum computational capabilities.

In Phase~A, the verifier and prover interact through classical communication.
At the end of this phase, they hold a classical transcript~$\trans\in\Trans$, represented by the random variable~$\transRV$, and the prover retains a residual register~$\proverReg$, which may contain classical or quantum information.

We associate with the protocol two deterministic functions
\begin{equation}
	\imap,\omap:\Trans\to\{0,1\}
	\; ,
\end{equation}
which map each Phase~A transcript to a \emph{virtual input} and \emph{virtual output}, respectively.
Accordingly, we define the corresponding random variables by
\begin{equation}
	\virtualInputRV=\imap(\transRV)
	\; ,
	\qquad
	\virtualOutputRV=\omap(\transRV)
	\; .
\end{equation}
The virtual input-output pair~$(\virtualInput,\virtualOutput)$ determined from Phase~A plays the role of the input and output of one of the two parties participating in a CHSH experiment.

In Phase~B, the verifier samples a fresh challenge~$\bellChallenge\in\{0,1\}$, sends it to the prover, and the prover responds with~$\bellOutcome\in\{0,1\}$.
The fresh challenge-response pair~$(\bellChallenge,\bellOutcome)$ plays the role of the input and output of the second party.
The verifier can therefore associate each execution of the single-prover protocol with a complete CHSH tuple~$(\virtualInput,\virtualOutput,\bellChallenge,\bellOutcome)$ and accepts according to the usual CHSH winning condition
\begin{equation}
\virtualOutput \oplus \bellOutcome
=
\virtualInput \cdot \bellChallenge \; .
\end{equation}
Importantly,~$\virtualInput$ or~$\virtualOutput$ are never explicitly sent to the prover; they are determined only virtually from the Phase~A transcript.

Building on the framework of computational Bell inequalities~\cite{cbi}, we distinguish two notions of virtual-input hiding according to the computational model of the adversary, which may either be \ppt{} or \qpt{}.
Both notions require the virtual input to be unpredictable from the Phase~A transcript and the prover's residual information, but differ in whether this information and the corresponding predictor are classical or quantum.

Let~$\transStaterv$ denote the prover's residual state in register~$\proverReg$ at the end of Phase~A.
For both notions, the relevant hiding condition is that every adversarial predictor~$\adversary$, given the transcript and the prover's residual state, satisfies
\begin{equation}
	\Pr\!\left[
		\adversary(\transRV,\transStaterv)
		=
		\virtualInputRV
	\right]
	\leq
	\frac12+\negl(\security)
	\; .
\end{equation}

We say that the virtual input is \emph{classically hidden} when this condition holds for every \ppt{} predictor whose residual information~$\transStaterv$ is classical.
Classical virtual-input hiding is sufficient for a \poq{}.
Indeed, under the computational Bell mapping, an efficient classical strategy achieving correlations beyond the computational classical bound could be transformed into an efficient classical predictor for~$\virtualInputRV$, contradicting classical virtual-input hiding.
Thus, a quantum prover achieving correlations beyond this bound certifies a separation from all efficient classical strategies and hence constitutes a \poq{}.

Randomness certification, however, requires the stronger notion of \emph{quantum virtual-input hiding}, in which the same condition must hold for every \qpt{} predictor whose post-Phase-A residual state~$\transStaterv$ may contain quantum information.
Quantum virtual-input hiding implies classical virtual-input hiding, since classical computation is a special case of quantum computation, but the converse need not hold.
The virtual input may be hidden from every efficient classical predictor while remaining efficiently recoverable by a quantum predictor from its residual quantum state.

The distinction between classical and quantum hiding has a direct operational consequence.
As we show in Section~\ref{sec:randomness_certification}, if a quantum prover can recover~$\virtualInputRV$ from the transcript~$\transRV$ and the residual state~$\proverReg$, it can condition its Phase~B operation on the virtual input and reproduce correlations that cannot be achieved classically without implementing measurements whose incompatibility forces uncertainty.
Such a protocol may still separate efficient quantum and classical computation, and therefore certify quantumness, while failing to certify randomness against the quantum prover's side information.
Quantum hiding prevents such a strategy by ensuring that the virtual input remains unpredictable even given the prover's residual quantum state.

\subsection{Upgrading \poqshort{} to Randomness Certification}
\label{subsec:intro_moya_randomness}

The protocol we study originates from the \poq{} of Morimae and Yamakawa~\cite{MYowf23}.
Their construction uses a \shbc{} together with the KCVY protocol~\cite{kcvy} to prepare a hidden BB84 state on the prover's side.
Neither the original \moya{} analysis nor the KCVY protocol is formulated in terms of virtual inputs or a Bell mapping.

Our first observation is that the same commitment-based construction nevertheless admits such an interpretation.
In particular, we identify a virtual input directly from the transcript of the protocol and show that its hiding property captures the distinction between proofs of quantumness and randomness certification.
The classical hiding guarantee already implicit in the original construction suffices to recover the \poq{} result, whereas strengthening it to quantum hiding yields randomness certification.

An honest quantum prover can execute the sender's commitment procedure coherently and, after the hashing and postselection steps, retain one valid opening~$\decomm_\bit$ for each commitment value~$\bit\in\{0,1\}$.
At an informal level, the resulting state has the form
\begin{equation}
	\ket{0}\ket{\decomm_0}
	+
	\ket{1}\ket{\decomm_1}
	\; .
\end{equation}

The \moya{} construction uses the KCVY protocol~\cite{kcvy} to map the coherent pair of openings to a hidden BB84 state.
In our construction, we replace this component with the simplified hidden-qubit procedure presented in~\cite{brakerski2023simple}, which is particularly convenient for identifying the corresponding Bell mapping.
As part of this procedure, the verifier samples random strings~$\rand_0$ and~$\rand_1$, which are included in the transcript~$\trans$.
Motivated by the basis of the resulting BB84 state, we choose the virtual-input function~$\imap$ such that
\begin{equation}
	\label{eq:intro_virtual_input}
	\imap(\trans)
	=
	(\decomm_0\cdot\rand_0)
	\oplus
	(\decomm_1\cdot\rand_1)
	\; .
\end{equation}

The security of the resulting protocol therefore depends on whether the virtual input remains hidden from the prover.
Computational binding against \ppt{} senders suffices for hiding against efficient classical provers, while \pq{} computational binding allows us to establish the corresponding guarantee against quantum provers, as summarized in the following theorem.

\begin{theorem}[Classical and quantum virtual-input hiding, informal]
	\label{thm:intro_virtual_input_hiding}
	Assume that the underlying \shbc{} satisfies computational binding against \ppt{} (\qpt{}) senders.
	Then, whenever \iidProtocol{} reaches Phase~B with non-negligible probability, its virtual input is computationally hidden from any \ppt{} (\qpt{}) prover, even given its residual post-Phase-A classical (quantum) state.
\end{theorem}

The classical statement recovers the \poq{} guarantee of the original \moya{} construction by excluding efficient classical strategies.
The quantum statement is the new ingredient needed for randomness certification.

At a high level, the quantum statement is proved by showing that any \qpt{} algorithm that predicts~$\imap(\trans)$ from the transcript and its residual quantum state with non-negligible advantage can be converted, using a quantum Goldreich--Levin procedure~\cite{qgl}, into a \qpt{} sender that recovers both~$\decomm_0$ and~$\decomm_1$ from a single commit-phase execution.
These values yield valid openings of the same commitment transcript to both~$0$ and~$1$, contradicting \pq{} computational binding.

The quantum virtual-input hiding guarantee in Theorem~\ref{thm:intro_virtual_input_hiding} allows us to establish our main entropy-certification result.

\begin{theorem}[Entropy certification (informal)]
    \label{thm:main}
    Assume the existence of a \shbc{} satisfying \pq{} computational binding.
    Then, there exists a randomness-certification protocol with a classical verifier such that the following holds.

    Let~$\prover^*$ be a \qpt{} prover, and let~$\gamma(\security)>0$ denote the probability that the protocol reaches its CHSH test.
    Suppose that, conditioned on reaching this test, the prover achieves a CHSH winning probability of at least~$\chshopt-\varepsilon$, where~$\chshopt$ is the optimal quantum CHSH winning probability and~$\varepsilon\geq0$.
    Then,
    \begin{equation}
        \label{eq:thm_main}
        \vne(B\mid T,Y,E)_{\rho\mid\mathrm{CE}}
        \geq
        \entropyConst
        -
        \O\!\left(\sqrt{\varepsilon}\right)
        -
        \frac{\negl(\security)}{\gamma(\security)}
        \; ,
    \end{equation}
    where~$\entropyConst>0$ is a universal constant,~$\mathrm{CE}$ denotes the event that the protocol reaches its CHSH test,~$B$ is the prover's output,~$T$ is the preceding classical interaction transcript,~$Y$ is the public CHSH challenge, and~$E$ is arbitrary quantum side information held by an adversary.
\end{theorem}

We prove this entropy bound in Appendix~\ref{app:entropy_proof} by combining approximate anticommutation of the prover's measurements, Jordan's lemma, and an entropic uncertainty relation~\cite{eur}.
The proof also establishes the fixed-challenge entropy bound used in the entropy-accumulation analysis.
Alternatively, when the protocol reaches the CHSH test with inverse-polynomial probability, the entropy guarantee follows from the computational Bell-inequality framework of~\cite{cbi}.

\paragraph{Entropy accumulation.}
Entropy-accumulation techniques and their subsequent refinements provide general tools for converting suitable single-round entropy bounds into finite-size guarantees for sequential protocols~\cite{dupuis_2019,metger2022generalised,arqand2024mutualinformationchainrules,eat_marginal,reyni_eat}, and have been applied to a variety of quantum-cryptographic tasks, including quantum key distribution, randomness certification and expansion, and other device-independent cryptographic protocols~\cite{Arnon-Friedman2018,qkd_eat,randomness_eat}.

The conditional von Neumann entropy bound of Theorem~\ref{thm:main} provides the single-round entropy guarantee needed for such an analysis.
In Section~\ref{sec:eat}, we use this bound to construct a min-tradeoff function and apply the entropy accumulation theorem to a sequential repetition of \iidProtocol{}.
The resulting protocol uses parameter estimation to test the continuation probability and CHSH winning probability, and certifies smooth min-entropy in the generation-round outputs conditioned on the adversary's quantum side information and the public round-type, challenge, and test-data registers.
The sequential theorem uses the convention that the verifier's copies of the complete Phase-A transcripts are not released as public side information and may be erased after the corresponding scores are computed.
The prover may of course retain information from its own interaction; the limitation concerns subsequent release of the verifier's transcript register to the external adversary.
Releasing that register requires an additional leakage argument or an enlarged EAT output register and is not covered by the theorem as presently stated.
Subject to this explicit convention, the construction gives a full sequential randomness-certification protocol.

\subsection{Instantiations and the Complexity of Verification}
\label{subsec:into_instantiations}

Theorem~\ref{thm:main} assumes the existence of a \shbc{} satisfying \pq{} computational binding.
This raises two distinct questions.
First, can such a commitment primitive be instantiated from concrete cryptographic assumptions?
Second, for a given instantiation, can the verifier efficiently identify and recover the valid openings required by the hidden-qubit preparation procedure?
The first question concerns the cryptographic assumptions required for the security guarantee of Theorem~\ref{thm:main}, whereas the second concerns the computational complexity of verification and is not required for the theorem itself.

We address these questions through three instantiations: the NOVY commitment from \pq{} \owps{}, a commitment from \pq{} \tcfs{}, and a hash-based commitment in the \qrom{}.
Here, we summarize only the properties relevant to security and verifier efficiency; the constructions are given in Section~\ref{sec:instantiations}.
All three satisfy the cryptographic requirements of Theorem~\ref{thm:main}, but they differ in the additional structure available to the verifier for recovering the relevant openings.

\paragraph{\Pq{} \owpshort{}.}

The NOVY commitment~\cite{novy} is perfectly hiding and computationally binding, conditioned on the existence of \owps{}.
Recent concurrent works extend its security to the quantum setting, establishing \pq{} computational binding conditioned on the existence of \pq{} \owps{}~\cite{chen2026minicrypt,dupuis2026postquantum,chevalier2026interactive}.
For every valid commitment transcript and bit~$b$, there is at most one valid opening, namely the preimage of one of the two outputs produced by interactive hashing.
This makes the uniqueness part of our continuation check automatic, but recovering the openings requires inverting the permutation.
A computationally unbounded verifier can do this by exhaustive search, so this route yields randomness certification in the standard model with an inefficient classical verifier.

\paragraph{\Pq{} \tcfshort{}.}
Our third instantiation uses an exact \pq{} \tcf{} family.
The underlying commitment is the natural \pq{} analogue of the standard bit commitment from a claw-free pair of functions~\cite{tcf_commitment,Chase2013}.
For every key accepted by the public key-validation algorithm, the common-image and injectivity properties make the commitment statistically hiding, while quantum claw-freeness prevents a \qpt{} sender from producing valid openings of the same commitment to both bits.
The trapdoor additionally allows the verifier to efficiently recover the relevant openings and certify their uniqueness, yielding an efficient classical verifier.
This claim concerns the exact, noiseless abstraction of Definition~\ref{def:tcf}; extending it to standard noisy TCF families requires a separate approximation analysis.

This instantiation is not intended to improve on existing \tcf{}-based randomness-certification protocols~\cite{brakerski2021cryptographic,merkulov2023entropy}.
Rather, expressing the construction through the commitment abstraction separates the cryptographic properties needed for security from the additional structure needed for efficient verification.
In the \tcf{} instantiation, statistical hiding and \pq{} binding provide the security properties required by our framework, while the trapdoor is used only to efficiently recover and certify the relevant openings.
This separation becomes meaningful in the \qrom{} instantiation below, where the same security properties can be obtained without a trapdoor, but efficient verification is lost.
Thus, the \tcf{} instantiation serves both as a consistency check for our abstraction and as a point of comparison that isolates the role of the trapdoor in efficient verification.

\paragraph{Quantum random oracle.}
Our third instantiation works in the \qrom{}~\cite{qrom_boneh}.
Let~$\oracle$ denote the random oracle, and consider the standard hash-based commitment~\cite{qrom_commitment,wonderful_qrom,unruh16}
\begin{equation}
    c
    =
    \oracle(\bit\Vert\rand)
    \; ,
\end{equation}
where~$\bit\in\{0,1\}$ is the committed bit,~$\rand$ is a uniformly random string held by the sender until the opening phase, and~$c$ is the resulting commitment string.
When the randomness space is sufficiently larger than the output space of~$\oracle$, the commitment is statistically hiding.
Two valid openings to different bits yield a collision in the random oracle, which gives \pq{} computational binding.
Thus, the \qrom{} provides the cryptographic properties required by Theorem~\ref{thm:main} without an additional computational assumption.

The verifier in this instantiation is, however, inefficient.
For a commitment~$c$ and bit~$\bit$, let
\begin{equation}
    \Decomm_{\bit,c}
    \coloneqq
    \left\{
        \rand :
        \oracle(\bit\Vert\rand)=c
    \right\}
    \; .
\end{equation}
Recall that the hidden-qubit preparation procedure uses pairwise-independent hashing to isolate a unique valid opening for each commitment value.
In particular, it samples hash functions~$\hashfunc_0,\hashfunc_1$ and obtains a hash value~$\hashval$, as described in \iidProtocol{}; the corresponding \qrom{} instantiation is presented in detail in Subsection~\ref{sec:qrom_commitment}.
The verifier must then determine whether
\begin{equation}
    \left|
        \Decomm_{\bit,c}
        \cap
        \hashfunc_{\bit}^{-1}(\hashval)
    \right|
    =
    1
\end{equation}
for both~$\bit\in\{0,1\}$.
In other words, for each commitment branch~$\bit$, the verifier checks that exactly one valid opening~$\rand$ is consistent with the observed hash value~$\hashval$, and recovers this unique opening.
The random oracle supplies no trapdoor or algebraic structure for this exact search-and-counting task, so our verifier performs it by exhaustive search.
The \qrom{} construction therefore gives an inefficiently verifiable, commitment-based route to randomness certification and isolates unique-opening decoding as the source of the inefficiency.

\subsection{Related Work}

A substantial line of work studies verification of quantum behavior and randomness certification using a single quantum device and a classical verifier~\cite{brakerski2021cryptographic,merkulov2023entropy,fefferman_randomness,aaronson2023certified,cbi}.
Unlike standard Bell tests with spatially separated devices, these protocols replace the physical no-signaling assumption with computational restrictions on what a single efficient prover can learn or compute.
\tcfs{} have played a central role in efficient classical verification and randomness certification~\cite{mahadev2018classical,kcvy,brakerski2021cryptographic,mahadev2022efficient,brakerski2023simple}.

\paragraph{Single-device verification and \poqs{}.}
The protocol we present in this work originates in the Morimae--Yamakawa \poq{} from \owfs{}~\cite{MYowf23}.
The~\moya{} protocol uses the \poq{} of Kahanamoku-Meyer \emph{et al.}~\cite{kcvy} as a subroutine.
We replace this component with the simplified hidden-qubit construction of Brakerski \emph{et al.}~\cite{brakerski2023simple}, which makes the virtual-input analysis and the reduction to commitment binding more transparent.
The resulting construction preserves the commitment-based structure while exposing the precise strengthening needed for entropy certification.

\paragraph{Remote state preparation and qubit testing.}
Closely related to our hidden-qubit preparation step is the literature on remote state preparation with verifiability, which formalizes the task of preparing and certifying a quantum state held by a remote quantum server.
Zhang~\cite{zhang2025formulations} develops a systematic framework for this primitive, studies its relation to tests of a qubit and to verifiable preparation of BB84 states, and gives constructions for broader state families, including coherent two-branch states of the form
\begin{equation}
        \ket{0}\ket{x_0}
        +
        \ket{1}\ket{x_1}
    \; .
\end{equation}
This is conceptually related to the coherent pair of valid openings produced in our construction.
The objectives and security guarantees are nevertheless different: rather than establishing a full remote-state-preparation functionality, we use the resulting hidden qubit to define a computationally hidden virtual input and derive randomness against quantum side information.

\paragraph{Applications to position verification.}
Kaleoglu \emph{et al.}~\cite{kaleoglu2026equivalence} establish generic connections between single-device certified randomness and classically verifiable position verification.
They give compilers from suitable certified-randomness protocols to position-verification protocols, extend the connection to multi-round protocols, and prove a converse for a relaxed certified-randomness notion.
These results are complementary to the present work: our contribution supplies a new commitment-based route to the entropy guarantee that such downstream transformations may use.

\paragraph{Randomness certification in the \qrom{}.}

Although our construction is not specific to the \qrom{}, one of its instantiations yields randomness certification in this model.
We therefore compare this instantiation with previous single-device randomness-certification protocols in the \qrom{}.

The Fourier-sampling construction of~\cite{fefferman_randomness} obtains certified randomness without additional computational assumptions, but does not establish security against an adversary holding quantum side information entangled with the device.
Extending their randomness guarantee to this adversarial setting is left as an open question in~\cite{fefferman_randomness}.
Aaronson and Hung~\cite{aaronson2023certified}, on the other hand, establish randomness certification in the random-oracle model against adversaries holding entangled quantum side information, albeit with inefficient classical verification.

The construction of~\cite{no_structure_quantum_advantage} achieves efficient verification and security against quantum side information, conditional on the Aaronson--Ambainis conjecture~\cite{aaronsonambainis2009}.
More recently, Khurana \emph{et al.}~\cite{randomness_no_aa} removed this conjecture for adversaries of quantum-query depth~$o(\log\security)$.
Subsequently, Coladangelo \emph{et al.}~\cite{unconditional_randomness} established efficient, publicly verifiable randomness certification against computationally unbounded adversaries making a subexponential number of adaptive quantum random-oracle queries.
Their result substantially strengthens the unconditional security guarantees available in the \qrom{}, although it does not explicitly establish conditional min-entropy against an external adversary holding quantum side information entangled with the device.

Jain, Mutreja, and Roberts~\cite{jain2026optimal} subsequently obtain optimal-rate certified randomness in the \qrom{} with a deterministic polynomial-time classical verifier.
They also define and construct a proof of maximal average conditional min-entropy conditioned on an adversarially chosen classical label.
Their result is therefore directly relevant to the QROM comparison, but it does not state the same entropy guarantee against an arbitrary external quantum side-information register~$\eveReg$ considered here.

Our construction follows a different approach.
Rather than relying on the structure of the Yamakawa--Zhandry protocol, we establish quantum virtual-input hiding from the computational binding of statistically hiding commitments.
This yields randomness certification against quantum side information without additional computational assumptions in the \qrom{}, at the cost of inefficient classical verification.
Table~\ref{tab:randomness_comparison} summarizes the different assumptions, security guarantees, and verification complexities.

\begin{table}[ht]
    \centering
    \small
    \setlength{\tabcolsep}{2pt}
    \renewcommand{\arraystretch}{2.6}
    \newcommand{\lineSpacing}{-2.5pt}

    \begin{tabular}{lccccccc}
        \hline
        \textbf{Protocol}
        & \cite{no_structure_quantum_advantage}
        & \cite{randomness_no_aa}
        & \cite{unconditional_randomness}
        & \cite{jain2026optimal}
        & \cite{fefferman_randomness}
        & \cite{aaronson2023certified}
        & \textbf{This work}
        \\
        \hline

        \makecell[l]{
            \textbf{Efficient}\\[\lineSpacing]
            \textbf{verification}
        }
        & Yes
        & Yes
        & Yes
        & Yes
        & No
        & No
        & No
        \\

        \makecell[l]{
            \textbf{Additional}\\[\lineSpacing]
            \textbf{assumption}
        }
        & \makecell{AA conj.}
        & None
        & None
        & None
        & None
        & None
        & None
        \\

        \makecell[l]{
            \textbf{Quantum}\\[\lineSpacing]
            \textbf{side}\\[\lineSpacing]
            \textbf{information}
        }
        & Yes
        & Yes
        & N.E.
        & N.E.
        & N.E.
        & Yes
        & Yes
        \\

        \makecell[l]{
            \textbf{Quantum-oracle}\\[\lineSpacing]
            \textbf{query restriction}
        }
        & None
        & \makecell{$o(\log\security)$\\[\lineSpacing]query depth}
        & None
        & None
        & None
        & None
        & None
        \\

        \hline
    \end{tabular}

	\caption{
		\footnotesize
		Comparison of single-device randomness-certification protocols in the \qrom{}.
		The comparison is adapted from Table~1 of~\cite{fefferman_randomness}, with additional works and comparison criteria relevant to the present setting.
		Here, ``AA conj.'' denotes the Aaronson--Ambainis conjecture~\cite{aaronsonambainis2009}, while ``N.E.'' (not established) indicates that the cited work does not explicitly prove the corresponding quantum-side-information guarantee.
		The quantum-oracle query restriction row records restrictions beyond the polynomial-time quantum adversary model used as the baseline for comparison.
		Here, ``None'' indicates that the result applies to this baseline without an additional query-depth or query-count restriction.
		In particular,~\cite{unconditional_randomness} establishes security against a larger class of computationally unbounded, subexponential-query adversaries, whereas~\cite{randomness_no_aa} requires quantum-query depth~$o(\log\security)$.
	}

    \label{tab:randomness_comparison}
\end{table}

\paragraph{Computational Bell inequalities.}
Our work builds heavily on the framework of computational Bell inequalities introduced in~\cite{cbi}; see Subsection~\ref{subsec:intro_hidden_virtual_inputs} for an overview.
Although~\cite{cbi} primarily applies this framework to analyze existing cryptographic protocols through their associated Bell mappings, here we use it as a tool for constructing and proving the security of a new randomness-certification protocol.
In particular, once quantum hiding of the virtual input is established from the underlying commitment scheme, the computational Bell inequality framework allows us to translate this cryptographic property into an entropy guarantee against quantum side information.
Thus, our construction illustrates how the framework can be used modularly: the protocol-specific cryptographic analysis establishes the required virtual-input hiding property, while the Bell analysis then yields randomness certification.

\subsection{Organization}
\label{subsec:organization}

Section~\ref{sec:preliminaries} introduces the required notation and background on conditional entropy, entropic uncertainty relations, bit commitments, computational Bell inequalities, hidden virtual inputs, and \tcfs{}.
Section~\ref{sec:randomness_certification} presents our commitment-based single-round protocol and its Bell mapping.
We establish classical and quantum virtual-input hiding, proving the latter through a reduction to \pq{} computational binding using the quantum Goldreich--Levin theorem.
We then derive the corresponding \poq{} and single-round randomness-certification guarantees.
Section~\ref{sec:eat} extends the construction to a sequential protocol.
Using the single-round entropy bound, we construct a min-tradeoff function and apply the entropy accumulation theorem to obtain a lower bound on the conditional smooth min-entropy of the generation string against quantum side information.
Section~\ref{sec:instantiations} instantiates the commitment primitive using the NOVY commitment from \pq{} \owp{}, exact \pq{} \tcfs{}, and a hash-based commitment in the \qrom{}.
We examine the distinction between efficient and inefficient verification and identify the unique-opening check as the main obstacle to efficient verification in the plain \qrom{}.

\section{Preliminaries}
\label{sec:preliminaries}

\subsection{Notation}

We let~$\N$ denote the set of natural numbers and use~$\security\in\N$ for the security parameter.
We use the terminology of negligible functions.
A function~$\mu:\N\to\mathbb{R}_{\geq 0}$ is \emph{negligible} if, for every polynomial~$\poly$,
\begin{equation}
	\poly(\security)\mu(\security)
	\xrightarrow{\security\to\infty}
	0
	\; .
\end{equation}
We write~$\negl(\security)$ to denote an arbitrary negligible function of~$\security$.

Unless stated otherwise, all \ppt{} and \qpt{} adversaries in this work may be nonuniform.
A \ppt{} adversary may receive polynomial-size classical advice, and a \qpt{} adversary may receive polynomial-size quantum advice.
The advice may depend on~$\security$ and need not be efficiently preparable, but it is independent of the honest parties' fresh random coins.

We denote the Pauli operators by~$\sigmax$, $\sigmay$, and~$\sigmaz$, and the identity operator by~$\id$.
For an operator~$A$, we denote its operator norm by
\begin{equation}
	\|A\|_\infty
	\coloneqq
	\sup_{\ket{\psi}:\,\|\ket{\psi}\|=1}
	\|A\ket{\psi}\|
	\; .
\end{equation}
For any statement~$S$, we denote its indicator by~$\indicator{S}$, which equals~$1$ if~$S$ is true and~$0$ otherwise.
For random variables~$X$ and~$Y$ over a finite sample space~$\Omega$, we denote their statistical distance by
\begin{equation}
	\Delta(X,Y)
	\coloneqq
	\frac{1}{2}
	\sum_{\omega\in\Omega}
	\left|
		\Pr[X=\omega]
		-
		\Pr[Y=\omega]
	\right|
	\; .
\end{equation}

Throughout this work, we omit normalization factors whenever they are immaterial to the argument and explicitly indicate normalization whenever it matters.

\subsection{Entropic Uncertainty Relations and Jordan's Lemma}
\label{sec:entropic_uncertainty_relations}

\begin{lemma}[Jordan's lemma]\label{lemma:jordan}
	Let~$U_1,U_2$ be two self-adjoint unitary operators acting on a Hilbert space~$\mathcal{H}$ of~a countable dimension.
	Let~$L$ be a normal operator acting on the same space such that~
	$[L,U_{1}]=[L,U_{2}]=0$.
	Then, there exists a decomposition of the Hilbert space into a direct sum of orthogonal subspaces~$\mathcal{H} = \oplus_{\alpha} \mathcal{H}_{\alpha}$
	such that for all~$\alpha$, $\dim{(\mathcal{H}_{\alpha})} \leq 2$ and given~$\ket{\psi} \in
	\mathcal{H}_\alpha$, all three operators satisfy~
	$U_{1}\ket{\psi},U_{2}\ket{\psi},L\ket{\psi}\in\mathcal{H}_{\alpha}$.
\end{lemma}

\begin{definition}[Von Neumann entropy]
	Let~$\rho\in\mathcal{D}(\mathcal{H})$ be a density operator on a Hilbert space~$\mathcal{H}$.
	The \emph{von Neumann entropy} of $\rho$ is defined as
	\begin{equation}
		\vne(\rho)
		\coloneqq
		-\tr\!\left[
			\rho \log \rho
		\right].
	\end{equation}
\end{definition}

\begin{definition}[Conditional von Neumann entropy]
	Let $\rho_{AE}$ be a density operator on~$\mathcal{H}_A \otimes \mathcal{H}_E$.
	The \emph{conditional von Neumann entropy} of $A$ given $E$ is
	\begin{equation}
		\vne(A \mid E)_\rho
		\coloneqq
		\vne(\rho_{AE})
		-
		\vne(\rho_E) \; ,
	\end{equation}
	where
	\begin{equation}
		\rho_E
		=
		\tr_A(\rho_{AE}) \; .
	\end{equation}
\end{definition}

\begin{definition}[Von Neumann entropy of a measurement]
	Let~$M=\{M_0,M_1\}$ be a binary measurement on a register~$A$, and let~$\rho_{AE}$ be a density operator.
	Measuring~$A$ with~$M$ produces a classical outcome~$X\in\{0,1\}$ and the classical--quantum state
	\begin{equation}
		\rho_{XE}
		\coloneqq
		\sum_{x\in\{0,1\}}
		\dyad{x}
		\otimes
		\tr_A\!\left[
		(M_x\otimes I_E)\rho_{AE}
		\right]
		\; .
	\end{equation}
	We define the von Neumann entropy of the measurement~$M$, and its conditional von Neumann entropy given~$E$, as the corresponding entropies of its classical outcome:
	\begin{equation}
		\vne(M)_\rho
		\coloneqq
		\vne(X)_{\rho}
		\qquad\text{and}\qquad
		\vne(M\mid E)_\rho
		\coloneqq
		\vne(X\mid E)_{\rho}
		\; .
	\end{equation}
\end{definition}

\begin{definition}[Squared overlap]
    \label{def:squared_overlap}
    Let~$\{\ket{u_i}\}_i$ and~$\{\ket{v_j}\}_j$ be two orthonormal bases of a finite-dimensional Hilbert space.
    Their \emph{squared overlap} is defined as
    \begin{equation}
        c
        \coloneqq
        \max_{i,j}
        \left|
            \langle u_i \mid v_j \rangle
        \right|^2
        \; .
    \end{equation}
    In this work, we only consider binary qubit observables.
    In particular, for
    \begin{equation}
        B_0 = \sigmaz
        \qquad\text{and}\qquad
        B_1 =
        \cos(\theta)\sigmaz
        +
        \sin(\theta)\sigmax
        \; ,
    \end{equation}
    where~$\theta\in[0,\pi/2]$ is the angle between their measurement axes, the squared overlap of their eigenbases is
    \begin{equation}
        c
        =
        \cos^2\left(\frac{\theta}{2}\right)
        =
        \frac{1}{2}
        +
        \frac{1}{2}\cos(\theta)
        \; .
    \end{equation}
\end{definition}

\begin{remark}
    The qubit representation above is without loss of generality for a pair of nondegenerate binary qubit observables.
    Up to a choice of basis and a relabeling of measurement outcomes, any such pair can be written in the stated form for some~$\theta\in[0,\pi/2]$.
\end{remark}

\begin{remark}[Binary measurements and Naimark dilation]
    \label{rem:binary_naimark}
    An arbitrary binary POVM can be implemented as a projective binary measurement on a Naimark dilation of the prover's register.
    We may therefore assume, without loss of generality, that the prover's binary Phase-B measurements are projective and are represented by self-adjoint unitaries with eigenvalues~$\pm1$.
    The dilation register is included in the prover's private register, while the adversary's reference register is left unchanged.
\end{remark}

\begin{lemma}[Entropic uncertainty principle~{\cite{eur} (Supplementary---Corollary 2)}]
	\label{lemma:uncertainty_relations}
	Let~$\Pi$ and~$M$ be two binary observables acting on a qubit over~$\mathcal{H}_\proverReg$, and let~$c$ denote the squared overlap of their eigenbases.
	For any density operator
	\begin{equation}
		\rho_{\proverReg\eveReg}
		\in
		\mathcal D\!\left(
			\mathcal H_{\proverReg}
			\otimes
			\mathcal H_{\eveReg}
		\right)
		\; ,
	\end{equation}
	it holds that
	\begin{equation}
		\vne(\Pi \mid \eveReg)
		\geq
		\log_2(1/c)
		-
		\vne(M)
		\; .
	\end{equation}
\end{lemma}

\begin{lemma}[Anticommutator and overlap]
	\label{lemma:anticommutator_overlap}
	Let~$\Pi$ and~$M$ be two binary observables acting on a qubit.
	Let~$c$ denote the squared overlap of their eigenbases.
	Then
	\begin{equation}
		\left\{
			\Pi,M
		\right\}^2
		=
		4(2c-1)^2\id
		\; .
	\end{equation}
	Equivalently,
	\begin{equation}
		c
		=
		\frac{1}{2}
		+
		\frac{1}{4}
		\sqrt{
			\frac{1}{2}
			\tr\!\left[
				\left\{\Pi,M\right\}^2
			\right]
		}
		=
		\frac{1}{2}
		+
		\frac{1}{2}|\cos(\theta)|
		\; ,
	\end{equation}
	where~$\theta\in[0,\pi]$ is the angle between the two observables in the Bloch representation.
\end{lemma}

\subsection{Commitments}
\label{subsec:commitments}

\begin{definition}[Bit Commitment Scheme]
    \label{def:bit_commitment}

    A classical bit commitment scheme~$\bc=(\commit,\open)$ consists of a classical interactive commit protocol~$\commit$ between a \ppt{} sender~$\sender$ and a \ppt{} receiver~$\receiver$, together with a deterministic polynomial-time classical verification algorithm~$\open$.
    All communication between the sender and receiver is classical.
    The commitment scheme consists of two stages, a commit phase and an open phase, defined as follows.

    \paragraph{Commit Phase.}
    On input a bit~$\bit\in\{0,1\}$ and the unary representation~$1^\security$ of the security parameter~$\security$, the sender and receiver interact according to~$\commit$.
    At the end of the interaction, a classical commitment transcript~$\trans$ is generated, and the sender retains a classical decommitment string~$\decomm$.

    \paragraph{Open Phase.}
    The sender sends~$(\bit,\decomm)$ to the receiver.
    The receiver evaluates
        \begin{equation}
            \open(\trans,\bit,\decomm)\in\{\top,\bot\}
            \; .
        \end{equation}

	We say that~$\bc$ is correct if, for all~$\security\in\N$ and all~$\bit\in\{0,1\}$,
	\begin{equation}
		\Pr
		\left[
			\open(\trans,\bit,\decomm)=\top
			:
			(\trans,\decomm)\leftarrow\commit(1^\security,\bit)
		\right]
		=1 \; .
	\end{equation}
\end{definition}

\begin{remark}[Canonical opening representation]
    \label{remark:canonical_opening_representation}
    Without loss of generality, we assume that the decommitment string~$\decomm\in\{0,1\}^{\ell}$ is the complete random tape of the honest sender, for some~$\ell=\poly(\security)$.
    Given~$(\trans,\bit,\decomm)$, the opening algorithm replays the honest sender's commit-phase algorithm using~$\decomm$ and the receiver messages recorded in~$\trans$, and accepts only if the resulting transcript is consistent with~$\trans$ and the original opening-verification algorithm accepts.
    This canonicalization preserves statistical hiding, since the commit phase is unchanged, and preserves computational binding, since any pair of valid openings in the canonicalized scheme induces a pair of valid openings in the original scheme.

    Consequently, for every commitment transcript~$\trans$, we may identify
    \begin{equation}
        \Decomm_{\bit,\trans}
        \coloneqq
        \left\{
            \decomm\in\{0,1\}^{\ell}
            :
            \decomm
            \text{ is an honest sender random tape consistent with }
            (\bit,\trans)
        \right\}
        \; .
    \end{equation}
\end{remark}

\begin{remark}
    In the security definitions below, the honest parties execute the classical protocols specified by~$\bc$, while the adversarial party may belong to a specified computational class.
    In particular, we consider both \ppt{} and \qpt{} senders and computationally unbounded receivers.
    When the adversarial sender is quantum, the interaction takes place in the \qccc{} model: the sender may perform quantum computations internally, but all communication with the honest receiver, including any openings, remains classical.
\end{remark}

\begin{definition}[Hiding experiment]
	Sample~$\bit\leftarrow\{0,1\}$.  The honest sender~$\sender(1^\security,\bit)$
	and the adversarial receiver~$\receiver^*(1^\security)$ execute the
	commit phase.  At the end of the interaction,~$\receiver^*$ outputs a
	bit~$\bit'$.
	The experiment outputs~$1$ iff~$\bit'=\bit$.
\end{definition}

\begin{definition}[$\epsilon$-Hiding]
	\label{def:hiding}
	A bit commitment scheme~$\bc=(\commit,\open)$ is \emph{$\epsilon$-hiding} if for every computationally unbounded adversarial receiver~$\receiver^*$,
	\begin{equation}
		\Pr\left[
			\hideExp_{\receiver^*,\bc}(1^\security)=1
		\right]
		\leq
		\frac12+\epsilon(\security) \; ,
	\end{equation}
	where the probability is over the randomness of the commitment scheme and the internal randomness of~$\receiver^*$.
	If~$\epsilon(\security)=\negl(\security)$, the scheme is \emph{statistically hiding}, and we denote it by~$\shbcscheme$.
\end{definition}

\begin{definition}[Binding Experiment]
	Let $\bc=(\commit,\open)$ be a bit commitment scheme and let~$\sender^*$ be an adversarial sender.
	Define the experiment $\bindExp_{\sender^*,\bc}(1^\security)$ as follows.
	\begin{enumerate}
		\item $\sender^*$ interacts with the honest receiver~$\receiver$ in the commit phase and produces a commitment transcript~$\trans$.
		\item $\sender^*$ outputs a pair of classical decommitment strings~$(\decomm_0,\decomm_1)$.
		\item The receiver~$\receiver$ evaluates both openings, and the experiment outputs~$1$ iff
		\begin{equation} \label{eq:double_open}
			\open(\trans,0,\decomm_0) = \top
			\quad\wedge\quad
			\open(\trans,1,\decomm_1) = \top \; .
		\end{equation}
	\end{enumerate}
\end{definition}

\begin{definition}[Computational Binding]
    \label{def:classical_binding}
    A classical bit commitment scheme~$\bc$ is computationally binding if, for every nonuniform \ppt{} adversarial sender~$\sender^*$ with polynomial-size classical advice independent of the honest receiver's fresh random coins, there exists a negligible function~$\nu_{\sender^*}$ such that
    \begin{equation}
        \Pr\left[
            \bindExp_{\sender^*,\bc}(1^\security)=1
        \right]
        \leq
        \nu_{\sender^*}(\security)
        \; .
    \end{equation}
\end{definition}

\begin{definition}[Post-Quantum Computational Binding\footnote{See Subsection~\ref{subsec:post_quantum_binding} for a discussion of this notion and its relation to stronger \pq{} binding notions.}]
    \label{def:post_quantum_binding}
    A classical bit commitment scheme~$\bc$ is \pq{} computationally binding if, for every nonuniform \qpt{} adversarial sender~$\sender^*$ with polynomial-size quantum advice independent of the honest receiver's fresh random coins, there exists a negligible function~$\nu_{\sender^*}$ such that
    \begin{equation}
        \Pr\left[
            \bindExp_{\sender^*,\bc}(1^\security)=1
        \right]
        \leq
        \nu_{\sender^*}(\security)
        \; .
    \end{equation}
\end{definition}





\subsection{\tcfshort{} Families}
\label{subsec:tcf}

\begin{definition}[\tcf{} Family]
	\label{def:tcf}
	A \tcf{} family is a collection
	\begin{equation}
		\Tcfunc
		=
		\left\{
			(\tcfunc_{\key,0},\tcfunc_{\key,1})
		\right\}_{\key\in\Key}
	\end{equation}
	together with \ppt{} algorithms
	\begin{equation}
		(\gen,\eval,\inv) \; ,
	\end{equation}
	satisfying the following properties.

	\paragraph{Key Generation.}
	The algorithm
	\begin{equation}
		(\key,\trapdoor)\leftarrow \gen(1^\security)
	\end{equation}
	outputs a public key~$\key\in\Key$ and trapdoor~$\trapdoor$.

	\paragraph{Function Pair.}
	For every key~$\key\in\Key$, the pair
	\begin{equation}
		(\tcfunc_{\key,0},\tcfunc_{\key,1})
	\end{equation}
	consists of efficiently computable functions
	\begin{equation}
		\tcfunc_{\key,0},\tcfunc_{\key,1}:\Preimage\to\Image \; .
	\end{equation}

	\paragraph{Injectivity.}
	For every~$\bit\in\{0,1\}$, the map~$\tcfunc_{\key,\bit}$ is injective.

	\paragraph{Common Image.}
	The images of~$\tcfunc_{\key,0}$ and~$\tcfunc_{\key,1}$ are identical.

	\paragraph{Trapdoor Inversion.}
	Given~$(\trapdoor,\bit,\image)$ such that
	\begin{equation}
		\image=\tcfunc_{\key,\bit}(\preimage) \; ,
	\end{equation}
	the inversion algorithm
	\begin{equation}
		\inv(\trapdoor,\bit,\image)
	\end{equation}
	outputs~$\preimage$.

	\paragraph{Efficient Evaluation.}
	For every~$\bit\in\{0,1\}$ and~$\preimage\in\Preimage$,
	\begin{equation}
		\eval(\key,\bit,\preimage)=\tcfunc_{\key,\bit}(\preimage) \; .
	\end{equation}

	\paragraph{Claw-Free Property.}
	For every \ppt{} adversary~$\adversary$,
	\begin{equation}
		\Pr\!\left[
			\begin{array}{c}
				\preimage_0,\preimage_1 \leftarrow \adversary(\key) : \\
				\tcfunc_{\key,0}(\preimage_0)=\tcfunc_{\key,1}(\preimage_1)
			\end{array}
		\right]
		\leq
		\negl(\security) \; ,
	\end{equation}
	where the probability is over
	\begin{equation}
		(\key,\trapdoor)\leftarrow\gen(1^\security)
	\end{equation}
	and the internal randomness of~$\adversary$.
\end{definition}

\begin{definition}[Post-quantum trapdoor claw-free function family]
	\label{def:pq_tcf}
	A \tcf{} family is called a \emph{post-quantum trapdoor claw-free family} if its claw-freeness property holds against all \qpt{} adversaries.
\end{definition}

\begin{remark}[Exact versus noisy trapdoor claw-free families]
	Definition~\ref{def:tcf} is an exact, noiseless abstraction: each branch is an injective function and the two branches have an identical image.
	The standard LWE-based primitive used in classical-verifier protocols is a noisy trapdoor claw-free family, whose branches are distributions rather than exact functions.
	Accordingly, the instantiation in Subsection~\ref{subsec:tcf_commitment} is conditional on an exact \pq{} \tcf{} family as defined here.
	Extending the commitment construction and the exact unique-opening check to noisy TCF families requires an additional approximation analysis and is not claimed in this work.
\end{remark}

\subsection{Computational Bell Inequalities and Hidden Virtual Inputs}

\begin{samepage}
	\begin{definition}[Canonical Form Protocol]
		\label{def:canonical_form_protocol}
		Let~$(\verifier,\prover)$ be the verifier-prover pair performing the following interactive protocol. 
	
		For every security parameter~$\security$, fix input set~$\Y$, output set~$\B$ and transcript set~$\Trans$.
		\begin{enumerate}[label=Phase \Alph*., leftmargin=60pt]
			\item ~\vspace{-7pt}
			\begin{enumerate}[label=\arabic*., leftmargin=15pt]
				\item The verifier and prover interact classically and produce an interaction transcript~$\trans \in \Trans$.
				The verifier then outputs~$\flag \in \{\accept, \abort, \continue\}$, and the prover holds a quantum state~$\transState$.
			\end{enumerate}
			\item (Conditioned on~$\flag = \continue$)
			\begin{enumerate}[label=\arabic*., leftmargin=15pt]
				\item The verifier samples a challenge~$\bellChallenge \leftarrow \BellChallenge$ uniformly at random and sends it to the prover.
				\item The prover responds with an output~$\bit \in \B$, which the verifier receives.
			\end{enumerate}
		\end{enumerate}
	\end{definition}
\end{samepage}

\begin{definition}[Virtual input, virtual output, and Bell mapping]
	\label{def:virtual_input_output}
	Let~$(\verifier,\prover)$ be a canonical-form protocol, and let~$\Trans$ denote the set of possible transcripts produced in Phase~A.
	Let
	\begin{equation}
		\imap:\Trans\to\X
		\qquad\text{and}\qquad
		\omap:\Trans\to\A
	\end{equation}
	be a pair of deterministic functions.
	For every transcript~$\trans\in\Trans$, we call~$\imap(\trans)$ the \emph{virtual input} and~$\omap(\trans)$ the \emph{virtual output} associated with~$\trans$.
	Together, the pair~$(\imap,\omap)$ defines a \emph{Bell mapping} from the correlations generated by the canonical-form protocol to correlations in the Bell scenario~$(|\A|,|\B|,|\X|,|\Y|)\in\N^{4}$.
\end{definition}

\begin{remark}
	For an execution of the protocol conditioned on~$\flag=\continue$, the Phase~A transcript~$\trans$ determines the virtual input-output pair
	\begin{equation}
		\bigl(\imap(\trans),\omap(\trans)\bigr)
		\in
		\X\times\A
		\; ,
	\end{equation}
	while in Phase~B the verifier samples an integer input~$y\in\Y$ and the prover returns an integer output~$b\in\B$.
	The Bell mapping therefore associates each protocol execution with variables
	\begin{equation}
		\virtualInput=\imap(\trans)\in\VirtualInput,
		\qquad
		\virtualOutput=\omap(\trans)\in\VirtualOutput,
		\qquad
		\bellChallenge\in\BellChallenge,
		\qquad
		\bellOutcome\in\BellOutcome
		\; ,
	\end{equation}
	and thereby induces the joint correlation
	\begin{equation}
		\left\{
			p(\virtualOutput,\bellOutcome,\virtualInput,\bellChallenge)
			\;\middle|\;
			\virtualOutput\in\VirtualOutput,\,
			\bellOutcome\in\BellOutcome,\,
			\virtualInput\in\VirtualInput,\,
			\bellChallenge\in\BellChallenge
		\right\}
	\end{equation}
	over the random variables~$(\virtualOutputRV,\bellOutcomeRV,\virtualInputRV,\bellChallengeRV)$, corresponding to the Bell scenario~$(\A,\B,\X,\Y)\in\N^{4}$.
\end{remark}

\begin{definition}[Quantum hidden virtual input]
    \label{def:hidden_virtual_input_quantum}
    Let~$\security\in\N$ be a security parameter.
    Let~$(\verifier,\prover)$ be a canonical-form protocol.
    Let~$\imap:\Trans\to\X$ be a function.

    We say that the virtual input~$\virtualInputRV=\imap(\transRV)$ is \emph{quantum hidden} if the following holds.
    For every adversarial \qpt{} prover~$\prover^*$ (possibly with advice) and every \qpt{} predictor~$\adversary$, consider an execution of the protocol~$(\verifier,\prover^*)$ conditioned on~$\flag=\continue$.
    Let~$\transRV$ denote the resulting transcript random variable, and let~$\transStaterv$ denote the adversarial prover's residual quantum state after Phase~A, including the classical transcript.
    Let~$\widehat{\virtualInputRV}$ denote the guess produced by~$\adversary$ given~$\transStaterv$.
    Then,
    \begin{equation}
        \Pr\!\left[
            \widehat{\virtualInputRV}
            =
            \virtualInputRV
            \;\middle|\;
            \flag=\continue
        \right]
        \leq
        \frac{1}{|\X|}
        +
        \negl(\security)
        \; .
    \end{equation}
    The probability is over the randomness of~$\verifier$, the adversarial prover, and the predictor's internal measurements.
\end{definition}

\begin{definition}[Classical hidden virtual input]
    \label{def:hidden_virtual_input_classical}
    We say that the virtual input~$\virtualInputRV$ is \emph{classical hidden} if the conclusion of Definition~\ref{def:hidden_virtual_input_quantum} holds for every adversarial \ppt{} prover~$\prover^*$ and every \ppt{} predictor~$\adversary$, with classical residual information~$\transStaterv$.
\end{definition}

\begin{definition}[Weighted hidden binary virtual input]
    \label{def:weighted_hidden_virtual_input}
    Let~$\virtualInputRV=\imap(\transRV)\in\{0,1\}$ denote the virtual input of a canonical-form protocol, and let~$\transStaterv$ denote the prover's residual state after Phase~A, including the classical transcript.
    For an adversarial prover~$\prover^*$ and a predictor~$\adversary$, let~$\widehat{\virtualInputRV}$ denote the guess produced by~$\adversary$ given~$\transStaterv$.

    We say that~$\virtualInputRV$ is \emph{weighted hidden} from an adversary class~$\mathsf{Adv}$ if, for every prover~$\prover^*$ and predictor~$\adversary$ in~$\mathsf{Adv}$, there exists a negligible function~$\nu$ such that
    \begin{equation}
        \Pr\!\left[
            \widehat{\virtualInputRV}
            =
            \virtualInputRV
            \;\wedge\;
            \flag=\continue
        \right]
        \leq
        \frac{1}{2}
        \Pr[\flag=\continue]
        +
        \nu(\security)
        \; .
    \end{equation}
    In the classical version, both the prover's residual information and the predictor are classical, whereas in the quantum version they may be quantum.
\end{definition}

\begin{remark}
    If~$\Pr[\flag=\continue]\geq 1/\poly(\security)$, then Definition~\ref{def:weighted_hidden_virtual_input} implies
    \begin{equation}
        \Pr\!\left[
            \widehat{\virtualInputRV}
            =
            \virtualInputRV
            \;\middle|\;
            \flag=\continue
        \right]
        \leq
        \frac{1}{2}
        +
        \negl(\security)
        \; .
    \end{equation}
    We use the weighted formulation because it remains meaningful without assuming a lower bound on the continuation probability.
\end{remark}

\begin{proposition}[Weighted computational CHSH bound]
    \label{prop:weighted_computational_chsh}
    Consider a binary canonical-form protocol with Bell-mapped variables~$(\virtualOutputRV,\bellOutcomeRV,\virtualInputRV,\bellChallengeRV)$.
    Suppose that~$\bellChallengeRV$ is sampled uniformly and independently after Phase~A and that~$\virtualInputRV$ is weighted hidden from every classical polynomial-time prover and predictor.
    Then, for every classical polynomial-time prover, there exists a negligible function~$\nu$ such that
    \begin{equation}
        \Pr\!
        \left[
            \flag=\continue
            \;\wedge\;
            \virtualOutputRV\oplus\bellOutcomeRV
            =
            \virtualInputRV\bellChallengeRV
        \right]
        \leq
        \frac{3}{4}\Pr[\flag=\continue]
        +
        \nu(\security)
        \; .
    \end{equation}
\end{proposition}

The hiding of the virtual input allows us to apply the computational Bell-inequality framework of~\cite{cbi}.
In particular, Proposition~\ref{prop:weighted_computational_chsh} bounds the CHSH winning probability of any efficient classical prover.
For entropy certification, we additionally require the virtual input to remain hidden from quantum provers, as established below.

The BGKM protocol~\cite{brakerski2023simple} is a prominent example of a canonical-form protocol that certifies nonclassical behavior through computationally hidden measurement choices.
Beyond illustrating the general framework, its hidden-qubit preparation procedure serves as a building block in our construction, similarly to the role played by the \cite{kcvy} protocol in the original Morimae--Yamakawa construction~\cite{MYowf23}.
We recall the \cite{brakerski2023simple} protocol and its honest quantum implementation in Appendix~\ref{app:bgkm_protocol}, emphasizing the features relevant to our construction.

\subsection{Smooth Min-Entropy and Entropy Accumulation}
\label{subsec:smooth_min_entropy_eat}

The single-round analysis in this work lower-bounds a conditional von Neumann entropy.
To obtain a composable statement for a sequential protocol, we use the smooth conditional min-entropy and the \eat{}~\cite{Arnon-Friedman2018,dupuis_2019,metger2022generalised}.
We follow the computational formulation used in~\cite{merkulov2023entropy}.
All logarithms in this subsection are base two.

\begin{definition}[Conditional min-entropy]
    Let~$\rho_{AE}$ be a subnormalized state.
    Its conditional min-entropy is
    \begin{equation}
        H_{\min}(A\mid E)_{\rho}
        \coloneqq
        \sup_{\sigma_E}
        \sup\left\{
            h\in\mathbb{R}:
            \rho_{AE}
            \leq
            2^{-h}\id_A\otimes\sigma_E
        \right\}
        \; ,
    \end{equation}
    where the first supremum is over normalized states~$\sigma_E$.
\end{definition}

\begin{definition}[Smooth conditional min-entropy]
    Let~$\epsilon\in(0,1)$.
    The~$\epsilon$-smooth conditional min-entropy is
    \begin{equation}
        H_{\min}^{\epsilon}(A\mid E)_{\rho}
        \coloneqq
        \sup_{\widetilde{\rho}_{AE}:\,
        P(\widetilde{\rho}_{AE},\rho_{AE})\leq\epsilon}
        H_{\min}(A\mid E)_{\widetilde{\rho}}
        \; ,
    \end{equation}
    where~$P(\cdot,\cdot)$ denotes the purified distance.
\end{definition}

\begin{definition}[\eat{} channels]
    \label{def:eat_channels}
    Let
    \begin{equation}
        \mathcal{M}_i:
        R_{i-1}
        \longrightarrow
        R_i O_i S_i Q_i
        \qquad
        (i\in[n])
    \end{equation}
    be quantum channels.
    They are called \emph{\eat{} channels} if the following conditions hold.
    \begin{enumerate}[label=\arabic*., leftmargin=30pt, rightmargin=10pt]
        \item
        Each~$O_i$ has finite dimension~$d_O$, and each~$Q_i$ is classical with a common finite alphabet~$\mathcal{Q}$.
        \item
        The value of~$Q_i$ can be computed from~$O_iS_i$ without changing the marginal state on~$O_iS_i$.
        \item
        For every initial state~$\rho_{R_0E}$, the final state satisfies the Markov conditions
        \begin{equation}
            O_1^{i-1}
            \longleftrightarrow
            S_1^{i-1}E
            \longleftrightarrow
            S_i
            \qquad
            (i\in[n])
            \; .
        \end{equation}
    \end{enumerate}
\end{definition}

\begin{definition}[Min-tradeoff function]
    \label{def:min_tradeoff_function}
    Let~$\{\mathcal{M}_i\}_{i\in[n]}$ be EAT channels with score alphabet~$\mathcal{Q}$.
    A differentiable convex function~$f_{\min}$ on the probability distributions over~$\mathcal{Q}$ is a \emph{min-tradeoff function} if, for every~$i$ and every distribution~$q$ over~$\mathcal{Q}$,
    \begin{equation}
        f_{\min}(q)
        \leq
        \inf_{\substack{
            \sigma_{R_{i-1}R'}:\\
            \mathcal{M}_i(\sigma)_{Q_i}=q
        }}
        \vne\!\left(
            O_i
            \mid
            S_iR'
        \right)_{\mathcal{M}_i(\sigma)}
        \; ,
    \end{equation}
    where~$R'$ purifies the input state.
\end{definition}

\begin{theorem}[Entropy Accumulation Theorem~(\eatshort{})]
    \label{thm:eat_rate_form}
    Let~$\{\mathcal{M}_i\}_{i\in[n]}$ be \eat{} channels, let~$f_{\min}$ be a min-tradeoff function, and let~$\Omega$ be an event determined by~$Q_1^n$.
    Let~$p_{\Omega}=\Pr[\Omega]$.
    Suppose that the set of empirical score distributions compatible with~$\Omega$ is convex and that every such distribution satisfies
    \begin{equation}
        f_{\min}(\operatorname{freq}(Q_1^n))
        \geq
        h_{\Omega}
        \; .
    \end{equation}
    Then, for every~$\epsilon_{\mathrm{s}}\in(0,1)$,
    \begin{equation}
        \label{eq:eat_rate_form}
        \vne_{\min}^{\epsilon_{\mathrm{s}}}\!\left(
            O_1^n
            \mid
            S_1^nE
        \right)_{\rho\mid\Omega}
        \geq
        n h_{\Omega}
        -
        \mu\sqrt{n}
        \; ,
    \end{equation}
    where one may take
    \begin{equation}
        \mu
        =
        2\left(
            \log(1+2d_O)
            +
            \left\lceil
                \lVert\nabla f_{\min}\rVert_{\infty}
            \right\rceil
        \right)
        \sqrt{
            1
            -
            2\log(\epsilon_{\mathrm{s}}p_{\Omega})
        }
        \; .
    \end{equation}
\end{theorem}

\begin{definition}[Efficient sequential prover with advice]
    \label{def:efficient_sequential_prover_advice}
    A sequential prover is computationally bounded if its private memory consists of polynomially many qubits and every round map is implemented by a quantum circuit of size polynomial in~$(\security,n)$.
    The initial prover--adversary state may be an arbitrary polynomial-size advice state, possibly entangled with~$\eveReg$, and need not be efficiently preparable.
    The advice state must be independent of the verifier's fresh random coins.
\end{definition}

The advice convention in Definition~\ref{def:efficient_sequential_prover_advice} is important.
The infimum in Definition~\ref{def:min_tradeoff_function} ranges over arbitrary input states of the round channel, rather than only states that have a known efficient preparation.
Accordingly, the single-round computational entropy statement used to construct~$f_{\min}$ must hold uniformly for polynomial-size advice states, as in~\cite{merkulov2023entropy}.

\section{Hidden Virtual Inputs and Randomness Certification}
\label{sec:randomness_certification}

In this section, we analyze a commitment-based canonical-form protocol derived from the \poq{} of Morimae and Yamakawa~\cite{MYowf23} and the hidden-qubit procedure of Brakerski \emph{et al.}~\cite{brakerski2023simple}.
We associate the protocol with a virtual Bell input and output and establish two corresponding hiding guarantees.
Classical computational binding hides the virtual input from efficient classical predictors, whereas \pq{} computational binding hides it even from quantum predictors with access to the residual post-Phase-A state.
The former enables a \poq{}, while the latter enables randomness certification.

\subsection{Protocol from \shbcshort{}}
\label{subsec:protocol_from_shbc}

In this subsection, we introduce our \pq{} \shbc{}-based single-round randomness-generation protocol.
We first describe the protocol as an interaction between the verifier and an arbitrary prover, specifying only the exchanged messages and the verifier's actions.
We then give the honest quantum implementation of the same protocol, which uses the commitment and hashing steps to isolate a pair of valid openings and the hidden-qubit procedure of Brakerski \emph{et al.}~\cite{brakerski2023simple} to prepare the state used in Phase~B.
This separation between the protocol interface and its honest implementation will allow us to analyze arbitrary prover strategies while retaining an explicit quantum strategy achieving the desired behavior.

\begin{algorithm}[H]
    \caption{Single-round randomness-generation protocol}
    \label{protocol:single_round}

    \vspace{0.25cm}

    Fix~$\varepsilon\in(0,1)$, $\hashNum\in\N$, and a \shbc{} scheme~$\shbcscheme$.

    \vspace{0.25cm}

    \textbf{Phase A1}

    \vspace{-0.15cm}

    \begin{enumerate}
        \item
        The verifier~$\verifier$ and prover~$\prover$ execute the commit phase of~$\shbcscheme$.
        Let~$\bctrans$ denote the resulting commitment transcript.

        \item
        The verifier uniformly samples a hash-scale index~$\hashIndex \leftarrow \{0,\ldots,\hashNum-1\}$.

        \item
        Let
        \begin{equation}
            \hashfam_{\hashIndex}
            =
            \left\{
                \hashfunc:
                \{0,1\}^{\ell}
                \rightarrow
                [\hashRangeSize_{\hashIndex}]
            \right\}
            \; 
        \end{equation}
        be a pairwise-independent family of hash functions, where
        \begin{equation}
            \hashRangeSize_{\hashIndex}
            \coloneqq
            \left\lceil
                (1+\hashEpsilon)^{\hashIndex}
            \right\rceil
            \; .
        \end{equation}
        The verifier uniformly samples~$\hashfunc_0,\hashfunc_1 \leftarrow \hashfam_{\hashIndex}$
        and sends~$(\hashIndex,\hashfunc_0,\hashfunc_1)$ to the prover.

        \item
    	The prover sends a hash value~$\hashval\in[\hashRangeSize_{\hashIndex}]$.

		\item
		The verifier checks whether, for both~$\bit\in\{0,1\}$,
		\begin{equation}
			\left|
				\Decomm_{\bit,\bctrans}
				\cap
				\hashfunc_{\bit}^{-1}(\hashval)
			\right|
			=
			1
			\; .
		\end{equation}
		If so, the verifier sets~$\flag=\continue$.
		Otherwise, the verifier sets~$\bellOutcome=0$ and aborts.

        \item
        Conditioned on~$\flag=\continue$, let~$\decomm_0,\decomm_1$ denote the unique elements satisfying
        \begin{equation}
            \decomm_{\bit}
            \in
            \Decomm_{\bit,\bctrans}
            \cap
            \hashfunc_{\bit}^{-1}(\hashval)
            \qquad
            \text{for }
            \bit\in\{0,1\}
            \; .
        \end{equation}
    \end{enumerate}

    \vspace{-0.15cm}

    \textbf{Phase A2}

    \vspace{-0.15cm}

    \begin{enumerate}[resume]
        \item
        The verifier uniformly samples~$\rand_0,\rand_1 \leftarrow \{0,1\}^{\ell}$ and sends~$(\rand_0,\rand_1)$ to the prover.

        \item
        The prover sends a string~$\hadx \in \{0,1\}^{\ell}$.
        The resulting Phase-A transcript is
        \begin{equation}
            \trans
            =
            (
                \bctrans,
                \hashIndex,
                \hashfunc_0,
                \hashfunc_1,
                \hashval,
                \hadx,
                \rand_0,
                \rand_1
            )
            \; .
        \end{equation}
    \end{enumerate}

    \vspace{-0.15cm}

    \textbf{Phase B}

    \vspace{-0.15cm}

    \begin{enumerate}[resume]
        \item
        The verifier uniformly samples~$\bellChallenge \leftarrow \{0,1\}$ and sends~$\bellChallenge$ to the prover.

        \item
        The prover sends a bit~$\bellOutcome \in \{0,1\}$.
    \end{enumerate}
\end{algorithm}

\iidProtocol{} specifies the interaction with an arbitrary prover and therefore makes no assumptions about the prover's internal implementation.
We next give a second description of the \emph{same protocol} that describes the behavior of an honest quantum prover.
This description makes the prover's internal state evolution explicit and will be useful for establishing completeness and identifying the hidden qubit produced at the end of Phase~A.

\addtocounter{algorithm}{-1}
\begin{algorithm}[H]
	\caption{Single-round randomness-generation protocol: \textbf{Honest implementation}}
	\label{protocol:single_round_honest}
	\vspace{0.25cm}
	
	Fix~$\varepsilon\in(0,1)$, $\hashNum\in\N$, and a \shbc{} scheme~$\shbcscheme$.

	\vspace{0.25cm}
	\textbf{Phase A1}
	\vspace{-0.15cm}
	\begin{enumerate}
		\item
		The verifier and prover execute the commit phase of~$\shbcscheme$, with the prover coherently implementing the honest classical sender's algorithm over a uniform superposition of its internal randomness tapes.
		Let~$\bctrans$ denote the resulting commitment transcript and~$\decomm$ the local information retained by the prover for opening the commitment.
		The prover then holds a superposition over all openings consistent with~$\bctrans$,
		\vspace{-0.2cm}
		\begin{equation}
			\sum_{\bit\in\{0,1\}}
			\sum_{\decomm\in\Decomm_{\bit,\bctrans}}
			\ket{\bit}\ket{\decomm}
			\; .
			\vspace{-0.4cm}
		\end{equation}

		\item 
		The verifier uniformly samples a hash-scale index~$\hashIndex\leftarrow\{0,\ldots,\hashNum-1\}$.
		
		\item 
		Let~$\hashfam_\hashIndex = \{\hashfunc:\{0,1\}^{\ell}\rightarrow[\hashRangeSize_\hashIndex]\}$ be a pairwise-independent family of hash functions, where~$\hashRangeSize_{\hashIndex}\coloneqq\left\lceil(1+\hashEpsilon)^{\hashIndex}\right\rceil$.
		The verifier uniformly samples~$\hashfunc_0,\hashfunc_1\leftarrow\hashfam_{\hashIndex}$ and sends~$(\hashIndex,\hashfunc_0,\hashfunc_1)$ to the prover.

		\item
		The prover coherently evaluates the corresponding hash function on each branch, obtaining
		\begin{equation}
			\sum_{\bit\in\{0,1\}}
			\sum_{\decomm\in\Decomm_{\bit,\bctrans}}
			\ket{\bit}
			\ket{\decomm}
			\ket{\hashfunc_{\bit}(\decomm)}
			\; .
		\end{equation}
		It measures the last register in the computational basis, obtaining a hash value~$\hashval$, and sends~$\hashval$ to the verifier.

		\item
		If, for all~$\bit\in\{0,1\}$,
		\begin{equation}
			\left|
				\Decomm_{\bit,\bctrans}
				\cap
				\hashfunc_\bit^{-1}(\hashval)
			\right|
			=
			1
			\; ,
		\end{equation}
		the verifier sets~$\flag=\continue$.
		Otherwise, the verifier sets~$\bellOutcome=0$, sets~$\flag=\abort$, and aborts.

		\item
		Conditioned on~$\flag=\continue$, let $\decomm_0,\decomm_1$ denote the unique openings.
		The honest prover holds
		\begin{equation}
			\ket{0}\ket{\decomm_0}
			+
			\ket{1}\ket{\decomm_1}
			\; .
		\end{equation}
	\end{enumerate}

	\vspace{-0.15cm}
	\textbf{Phase A2} --- The verifier and prover execute the hidden-qubit procedure in \hyperref[protocol:bgkm_protocol]{Protocol~\ref{protocol:bgkm_protocol}}.
	\begin{enumerate}[resume]
		\item
		The verifier uniformly samples~$\rand_0,\rand_1\leftarrow\{0,1\}^{\ell}$ and sends~$(\rand_0,\rand_1)$ to the prover.
		The prover coherently computes the corresponding inner product on each branch, yielding
		\vspace{-0.25cm}
		\begin{equation}
			\ket{0}\ket{\decomm_0}\ket{\decomm_0\cdot\rand_0}
			+
			\ket{1}\ket{\decomm_1}\ket{\decomm_1\cdot\rand_1}
			\; .
		\end{equation}
		\vspace{-0.85cm}
		\item
		\label{step:rc_phase_a2_measurement}
		The prover uncomputes the first register and measures the decommitment register in the Hadamard basis, obtaining a string~$\hadx\in\{0,1\}^{\ell}$, which it sends to the verifier.
		The resulting transcript is~$\trans =(\bctrans,\hashIndex,\hashfunc_0,\hashfunc_1,\hashval,\hadx,\rand_0,\rand_1)$
		and the prover holds
		\begin{equation}
			\label{eq:rc_post_phase_a_state}
			\ket{\transState}
			=
			\ket{\decomm_0\cdot\rand_0}
			+
			(-1)^{
				\hadx\cdot
				\left(
					\decomm_0\oplus\decomm_1
				\right)
			}
			\ket{\decomm_1\cdot\rand_1}
			\; .
		\end{equation}
	\end{enumerate}
	\textbf{Phase B}
	\vspace{-0.15cm}
	\begin{enumerate}[resume]
		\item
		The verifier samples~$\bellChallenge\leftarrow\{0,1\}$ uniformly.

		\item
		The prover measures
		\begin{equation}
			\bellMeasurement_{\bellChallenge}
			=
			\left(
				\sigmaz +
				(-1)^{\bellChallenge}\sigmax
			\right) / \sqrt{2}
		\end{equation}
		on~$\ket{\transState}$ and returns the outcome~$\bellOutcome$.
	\end{enumerate}
\end{algorithm}

A central event in the subsequent analysis is whether the hashing step successfully isolates a unique valid opening for each committed bit.
We denote this \emph{continuation event} by~$\contEvent$ and its probability, as a function of the security parameter~$\security$, by
\begin{equation}
	\label{eq:rc_continuation_probability}
	\pce(\security)
	\coloneqq
	\Pr[\contEvent]
	=
	\Pr[\flag=\continue]
	\; .
\end{equation}
Conditioned on~$\contEvent$, the prover holds a coherent pair of unique valid openings at the end of Phase~A1, and the protocol proceeds to the hidden-qubit procedure in Phase~A2 and subsequently to Phase~B.

\begin{lemma}[Completeness of Phase~A1]
    \label{lemma:rc_phase_a_success}
    Set~$\hashEpsilon=1/100$, and choose
    \begin{equation}
		\hashNum
		=
		\left\lceil
			\frac{\ell+1}{\log_2(1+\hashEpsilon)}
		\right\rceil
		\; .
    \end{equation}
    Then, for an honest \qpt{} prover in Protocol~\ref{protocol:single_round}, the continuation probability satisfies
    \begin{equation}
        \label{eq:MY_completeness}
        \Pr[\contEvent]
        \geq
        \frac{1}{10\hashNum}
        -
        \negl(\security)
        \; .
    \end{equation}
    In particular, since~$\ell=\poly(\security)$, we have~$\hashNum=\poly(\security)$, and hence the honest prover continues to Phase~A2 with inverse-polynomial probability.
\end{lemma}

\begin{proof}
    The argument follows the completeness analysis of Morimae and Yamakawa~\cite[Sec.~5.1]{MYowf23}.
    We recall the main ingredients to explain the role of the hashing procedure.

    By statistical hiding, except with negligible probability over the commitment transcript~$\bctrans$, the sets of consistent honest sender random tapes~$\Decomm_{0,\bctrans}$ and~$\Decomm_{1,\bctrans}$ have approximately equal cardinalities~\cite[Lemma~4.1]{MYowf23}.

    The verifier does not know the cardinalities of these sets.
    It therefore samples the hash-scale index~$\hashIndex$ uniformly from~$\{0,\ldots,\hashNum-1\}$.
    Our choice of~$\hashNum$ guarantees that
    \begin{equation}
        (1+\hashEpsilon)^{\hashNum}
        \geq
        2^{\ell+1}
        \; .
    \end{equation}
    Consequently, for every transcript satisfying the above cardinality condition, Lemma~4.2 of~\cite{MYowf23} guarantees the existence of a suitable hash scale at which the expected number of surviving openings in each branch is of constant order.
    Since~$\hashIndex$ is chosen uniformly, such a scale is selected with probability at least~$1/\hashNum$.

    Finally, conditioned on selecting a suitable scale, the pairwise-independent hashing analysis of~\cite[Lemma~5.3]{MYowf23} shows that the hash value produced by the honest prover isolates exactly one valid decommitment in each branch with probability at least~$1/10$.

    Combining these events gives
    \begin{align}
        \Pr[\contEvent]
        &\geq
        \left(
            1-\negl(\security)
        \right)
        \frac{1}{\hashNum}
        \frac{1}{10}
        \nonumber\\
        &\geq
        \frac{1}{10\hashNum}
        -
        \negl(\security)
        \; ,
    \end{align}
    as claimed.
\end{proof}

Following~\cite{cbi}, we associate each continuing transcript with a virtual Bell input and output through a pair of functions~$\imap$ and~$\omap$.
The choice of these functions is an important part of the Bell mapping: in particular, we seek a virtual-input function~$\imap$ for which~$\imap(\trans)$ remains computationally hidden from the prover even given its residual post-Phase-A state.
Recall from Equation~\eqref{eq:rc_post_phase_a_state} that the honest post-Phase-A state is one of the four BB84 states.
Motivated by this state, we choose~$\imap$ and~$\omap$ so that the honest execution admits the desired CHSH interpretation and, as we show below, the resulting virtual input satisfies the required quantum hiding property.

\begin{definition}[Virtual input and virtual output for {\iidProtocol{}}]
    \label{def:rc_virtual_input_output}
    Let~$\trans$ be a transcript for which~$\flag=\continue$.
    The virtual-input function~$\imap$ is
    \begin{equation}
        \imap(\trans)
        \coloneqq
        (\decomm_0\cdot\rand_0)
        \oplus
        (\decomm_1\cdot\rand_1)
        \; ,
    \end{equation}
    and the virtual-output function~$\omap$ is
    \begin{equation}
        \omap(\trans)
        \coloneqq
        \begin{cases}
            \decomm_0\cdot\rand_0
            \; ,
            &
            \text{if }\imap(\trans)=0
            \; ,
            \\[1.8mm]
            \hadx\cdot(\decomm_0\oplus\decomm_1)
            \; ,
            &
            \text{if }\imap(\trans)=1
            \; .
        \end{cases}
    \end{equation}
\end{definition}

Let~$\transRV$ denote the transcript random variable, and define the corresponding
virtual-input and virtual-output random variables by
\begin{equation}
    \virtualInputRV
    \coloneqq
    \imap(\transRV)
    \; ,
    \qquad
    \virtualOutputRV
    \coloneqq
    \omap(\transRV)
    \; .
\end{equation}
Let~$\bellChallengeRV$ and~$\bellOutcomeRV$ denote the Phase~B challenge and output
random variables, respectively.
\begin{definition}[Acceptance predicate and probability]
    \label{def:rc_acceptance_predicate}
    Conditioned on~$\contEvent$, the verifier accepts if and only if
    \begin{equation}
        \label{eq:rc_acceptance_predicate}
        \bellOutcomeRV
        \oplus
        \virtualOutputRV
        =
        \bellChallengeRV\cdot\virtualInputRV
        \; .
    \end{equation}
    This is precisely the CHSH winning condition under the Bell mapping chosen above.

    We denote the corresponding acceptance probability by
    \begin{equation}
        \label{eq:rc_acceptance_probability}
        \pacc
        \coloneqq
        \Pr\left[
            \bellOutcomeRV
            \oplus
            \virtualOutputRV
            =
            \bellChallengeRV\cdot\virtualInputRV
            \;\middle|\;
            \contEvent
        \right]
        \; .
    \end{equation}
    Here,~$\chshopt=\cos^2(\pi/8)$ denotes the optimal quantum CHSH value.
\end{definition}

The security of the protocol reduces to proving the computational hiding of~$\virtualInputRV=\imap(\transRV)$.

\newcommand{\virtualHidingRef}[1]{%
    \hyperref[thm:virtual_input_hiding]{%
        Theorem~\ref*{thm:intro_virtual_input_hiding}%
        \IfStrEqCase{#1}{%
            {}{}%
            {classical}{.\ref*{item:classical_hiding}}%
            {quantum}{.\ref*{item:quantum_hiding}}%
        }%
    }%
}

\begin{virtualinputrestate}[Classical and quantum hiding of the virtual input]
	\phantomsection
	\label{thm:virtual_input_hiding}
	Let~$\transStaterv$ denote the residual state held by the adversarial prover at the end of Phase~A2 of \iidProtocol{}, after Step~\ref{step:rc_phase_a2_measurement} (see Equation~\eqref{eq:rc_post_phase_a_state} for the corresponding honest-prover state).

	\begin{enumerate}[label=\alph*.,ref=\alph*,leftmargin=30pt,rightmargin=10pt]
		\item
		\label{item:classical_hiding}
		Suppose that the underlying \shbc{} is computationally binding according to Definition~\ref{def:classical_binding}.
		Then, for every adversarial \ppt{} prover~$\prover^*$ and every \ppt{} predictor~$\adversary$, there exists a negligible function~$\nu_{\prover^*,\adversary}$ such that
		\begin{equation}
			\Pr\left[
				\adversary(\transStaterv)
				=
				\virtualInputRV
				\;\wedge\;
				\contEvent
			\right]
			\leq
			\frac12\pce(\security)
			+
			\nu_{\prover^*,\adversary}(\security)
			\; .
		\end{equation}
		\item
		\label{item:quantum_hiding}
		Suppose that the underlying \shbc{} is \pq{} computationally binding according to Definition~\ref{def:post_quantum_binding}.
		Then, for every adversarial \qpt{} prover~$\prover^*$ and every \qpt{} predictor~$\qadversary$, there exists a negligible function~$\nu_{\prover^*,\qadversary}$ such that
		\begin{equation}
			\Pr\left[
				\qadversary(\transStaterv)
				=
				\virtualInputRV
				\;\wedge\;
				\contEvent
			\right]
			\leq
			\frac12\pce(\security)
			+
			\nu_{\prover^*,\qadversary}(\security)
			\; .
		\end{equation}
	\end{enumerate}

	Consequently, whenever~$\pce(\security)\geq1/\poly(\security)$, the corresponding conditional prediction probability is at most $1/2+\negl(\security)$.
\end{virtualinputrestate}

\subsection{Certification of Nonclassicality from \shbcshort{}}
\label{subsec:nonclassicality_from_shbc}

Classical hiding of the virtual input already turns the protocol into a proof of quantumness.
The reason is that a classical prover that wins the Bell-mapped CHSH test with probability greater than~$3/4$ must adapt its Phase-B response to the virtual input, whereas \virtualHidingRef{classical} rules out learning that input with non-negligible advantage.
This is precisely the computational replacement for the locality constraint in an ordinary Bell test.

To include executions that do not reach Phase~B, define
\begin{equation}
    \mathsf{Score}_{\mathrm{PoQ}}(\prover)
    \coloneqq
    \Pr_{\prover}
    \left[
        \contEvent
        \;\wedge\;
        \virtualOutputRV\oplus\bellOutcomeRV
        =
        \virtualInputRV\bellChallengeRV
    \right]
    +
    \frac{3}{4}\Pr_{\prover}[\neg\contEvent]
    \; .
\end{equation}

\begin{theorem}[Computational Bell violation]
    \label{thm:shbc_nonclassicality}
    Suppose that the commitment scheme is statistically hiding and computationally binding against polynomial-time classical senders.
    Then every polynomial-time classical prover~$\prover^*$ satisfies
    \begin{equation}
        \mathsf{Score}_{\mathrm{PoQ}}(\prover^*)
        \leq
        \frac{3}{4}
        +
        \negl(\security)
        \; .
    \end{equation}
    The honest quantum prover satisfies
    \begin{equation}
        \mathsf{Score}_{\mathrm{PoQ}}(\prover_{\mathsf{hon}})
        \geq
        \frac{3}{4}
        +
        \frac{\sqrt{2}-1}{40\hashNum}
        -
        \negl(\security)
        \; .
    \end{equation}
\end{theorem}

\begin{proof}
    \virtualHidingRef{classical} and Proposition~\ref{prop:weighted_computational_chsh} give
    \begin{equation}
        \Pr
        \left[
            \contEvent
            \;\wedge\;
            \mathsf{CHSH\ win}
        \right]
        \leq
        \frac{3}{4}\Pr[\contEvent]
        +
        \negl(\security)
        \; .
    \end{equation}
    Adding~$\frac{3}{4}\Pr[\neg\contEvent]$ proves the classical bound.
    For the honest quantum prover, the conditional CHSH winning probability is~$\cos^2(\pi/8)$ and Lemma~\ref{lemma:rc_phase_a_success} gives
    \begin{equation}
        \Pr[\contEvent]
        \geq
        \frac{1}{10\hashNum}
        -
        \negl(\security)
        \; .
    \end{equation}
    Substitution gives the claimed inverse-polynomial separation.
\end{proof}

Conditioned on~$\contEvent$, the honest Bell-mapped behavior achieves the CHSH winning probability~$\cos^2(\pi/8)$, giving a constant violation of~$(\sqrt{2}-1)/4$ over the classical bound.
However, the continuation event occurs with probability only~$\Omega(1/\poly(\security))$, so the completeness-soundness gap of the complete one-shot protocol is inverse polynomial.
This is also the case in the original construction of Morimae and Yamakawa~\cite{MYowf23}, who show that the gap can be amplified by sequential repetition using a stronger notion of soundness.
Such amplification is not immediate for inefficient-verifier protocols, since standard repetition arguments need not apply.

\subsection{Certification of Randomness from \pqshort{}-\shbcshort{}}
\label{subsec:randomness_from_pq_shbc}

Randomness certification requires strengthening the above classical hiding guarantee to quantum hiding.
A quantum prover may retain a residual quantum state from which the virtual input can be recovered efficiently, even when it remains hidden from every efficient classical prover.
Such a prover could condition its Phase~B measurement on the recovered virtual input and reproduce the desired Bell correlations without exhibiting the measurement incompatibility required for randomness certification.

The quantum hiding guarantee of \virtualHidingRef{quantum} rules out this possibility more generally by ensuring that no \qpt{} prover can predict the virtual input with non-negligible advantage from its residual post-Phase-A state.
Requiring the \shbc{} to additionally satisfy computational binding against \qpt{} senders ensures that the virtual input remains hidden even from a \qpt{} prover with access to its residual post-Phase-A quantum state.
Thus, strengthening the binding guarantee from classical to \pq{} security upgrades the protocol from a \poq{} to a randomness-certification protocol.

\begin{proof}[Proof of~\virtualHidingRef{quantum}]
	Fix an adversarial \qpt{} prover~$\prover^*$ and a \qpt{} predictor~$\qadversary$.
	Let~$\transStaterv$ denote the residual state held by~$\prover^*$ after Phase~A, including the classical transcript as classical side information.

	Assume towards a contradiction that the weighted prediction advantage
	\begin{equation}
		\label{eq:quantum_hiding_advantage}
		\delta(\security)
		\coloneqq
		\Pr\left[
			\qadversary(\transStaterv)
			=
			\virtualInputRV
			\;\wedge\;
			\contEvent
		\right]
		-
		\frac12
		\Pr[\contEvent]
	\end{equation}
	is non-negligible.
	Then there exist a polynomial~$q$ and an infinite set~$\mathcal I\subseteq\N$ such that
	\begin{equation}
		\label{eq:quantum_hiding_noticeable_advantage}
		\delta(\security)
		\geq
		\varepsilon_0(\security)
		\coloneqq
		\frac{1}{q(\security)}
		\end{equation}
	for every~$\security\in\mathcal I$.
	We restrict attention to these security parameters.

	We construct a \qpt{} sender~$\guesser$ that violates \pq{} computational binding.
	During the commit phase,~$\guesser$ internally executes~$\prover^*$ and forwards the messages between~$\prover^*$ and the honest receiver.
	After the commit phase is complete,~$\guesser$ executes the hashing step of Phase~A1 internally, sampling all verifier parameters exactly according to the protocol and sending the corresponding messages to~$\prover^*$.

	Let~$\Sigma$ denote the complete classical prefix immediately after~$\hashval$ is produced and immediately before~$\rand_0,\rand_1$ are sampled.
	Thus,~$\Sigma$ contains all public values fixed before~$\rand_0,\rand_1$ are sampled, including
	\begin{equation}
		\Sigma
		=
		\left(
			\bctrans,
			\hashIndex,
			\hashRangeSize,
			\hashfunc_0,
			\hashfunc_1,
			\hashval
		\right)
		\; .
	\end{equation}
	For every prefix~$\sigma$ with positive probability, let
	\begin{equation}
		\mu(\sigma)
		\coloneqq
		\Pr[\Sigma=\sigma]
		\; ,
	\end{equation}
	and let~$\rho_\sigma$ denote the residual quantum state held by~$\prover^*$ conditioned on~$\Sigma=\sigma$.
	Any auxiliary or advice registers used by~$\prover^*$ and~$\qadversary$ are included among the auxiliary registers throughout the reduction.

	Let~$\continuingPrefixes$ denote the set of prefixes for which~$\contEvent$ holds.
	Since~$\contEvent$ is determined by the classical prefix~$\Sigma$, we have
	\begin{equation}
		\label{eq:quantum_continuing_prefix_mass}
		\sum_{\sigma\in\continuingPrefixes}
		\mu(\sigma)
		=
		\Pr[\contEvent]
		=
		\pce
		\; .
	\end{equation}
	The reduction~$\guesser$ does not decide whether its sampled prefix belongs to~$\continuingPrefixes$.
	The set~$\continuingPrefixes$ is used only in the analysis.

	Fix a prefix~$\sigma\in\continuingPrefixes$.
	By the definition of the continuation event, there exist unique strings~$\decomm_{0,\sigma}$ and~$\decomm_{1,\sigma}$ such that
	\begin{equation}
		\Decomm_{\bit,\bctrans^\sigma}
		\cap
		\hashfunc_{\bit,\sigma}^{-1}(\hashval_\sigma)
		=
		\left\{
			\decomm_{\bit,\sigma}
		\right\}
		\qquad
		\text{for every }\bit\in\{0,1\}
		\; .
	\end{equation}
	Using the fixed polynomial-length encoding of decommitment strings, define
	\begin{equation}
		\concDecomm_\sigma
		\coloneqq
		\decomm_{0,\sigma}
		\mathbin\Vert
		\decomm_{1,\sigma}
		\in
		\{0,1\}^{2\ell}
		\; .
	\end{equation}
	For
	\begin{equation}
		\concRand
		\coloneqq
		\rand_0
		\mathbin\Vert
		\rand_1
		\in
		\{0,1\}^{2\ell}
		\; ,
	\end{equation}
	the virtual input satisfies
	\begin{align}
		\imap(\trans)
		&=
		(\rand_0\cdot\decomm_{0,\sigma})
		\oplus
		(\rand_1\cdot\decomm_{1,\sigma})
		\nonumber\\
		&=
		\concRand\cdot\concDecomm_\sigma
		\; .
		\label{eq:quantum_concatenated_inner_product}
	\end{align}

	For every~$\sigma\in\continuingPrefixes$, define
	\begin{equation}
		p_\sigma
		\coloneqq
		\Pr\left[
			\qadversary(\transStaterv)
			=
			\concRand\cdot\concDecomm_\sigma
			\;\middle|\;
			\Sigma=\sigma
		\right]
		\; ,
	\end{equation}
	where the probability is over uniform and independent
	$\concRand\leftarrow\{0,1\}^{2\ell}$ and over all measurements performed in the remainder of Phase~A and by~$\qadversary$.
	Define the corresponding bias by
	\begin{equation}
		\beta_\sigma
		\coloneqq
		p_\sigma-\frac12
		\; .
	\end{equation}

	By Equations~\eqref{eq:quantum_hiding_advantage} and~\eqref{eq:quantum_continuing_prefix_mass},
	\begin{align}
		\sum_{\sigma\in\continuingPrefixes}
		\mu(\sigma)\beta_\sigma
		&=
		\sum_{\sigma\in\continuingPrefixes}
		\mu(\sigma)
		\left(
			p_\sigma-\frac12
		\right)
		\nonumber\\
		&=
		\Pr\left[
			\qadversary(\transStaterv)
			=
			\virtualInputRV
			\;\wedge\;
			\contEvent
		\right]
		-
		\frac12\Pr[\contEvent]
		\nonumber\\
		&=
		\delta(\security)
		\nonumber\\
		&\geq
		\varepsilon_0(\security)
		\; .
		\label{eq:quantum_average_bias}
	\end{align}

	Define the set of positive-bias continuing prefixes by
	\begin{equation}
		\mathcal C_+
		\coloneqq
		\left\{
			\sigma\in\continuingPrefixes:
			\beta_\sigma>0
		\right\}
		\; .
	\end{equation}
	Since the prefixes outside~$\mathcal C_+$ contribute nonpositively to Equation~\eqref{eq:quantum_average_bias},
	\begin{equation}
		\label{eq:quantum_positive_bias_mass}
		\sum_{\sigma\in\mathcal C_+}
		\mu(\sigma)\beta_\sigma
		\geq
		\delta(\security)
		\geq
		\varepsilon_0(\security)
		\; .
	\end{equation}

	We now describe the quantum Goldreich--Levin procedure associated with a fixed prefix~$\sigma\in\mathcal C_+$.
	Let
	\begin{equation}
		\ket{\psi_\sigma}_{\proverReg\eveReg}
	\end{equation}
	be an arbitrary purification of~$\rho_\sigma$, where~$\proverReg$ contains all registers accessible to~$\guesser$ and~$\eveReg$ is an inaccessible reference register.
	The reference register is used only for the analysis, and all operations below act trivially on~$\eveReg$.

	By applying deferred measurement and retaining coherent copies of all measurement outcomes, the remainder of Phase~A together with the execution of~$\qadversary$ can be represented by a polynomial-size unitary~$U_\sigma$.
	On input~$\ket{\concRand}$, this unitary performs the following operations coherently:
	\begin{enumerate}
		\item it copies~$\concRand$ into separate message registers while leaving the original query register unchanged;
		\item it executes the remainder of~$\prover^*$'s Phase-A strategy and records all classical messages and measurement outcomes in auxiliary registers;
		\item it executes~$\qadversary$ on the resulting residual state and classical transcript; and
		\item it writes the predictor's output into a designated output qubit.
	\end{enumerate}
	Because the original query register is used only as a read-only control, the first~$2\ell$ qubits of every computational-basis query~$\ket{\concRand}$ remain in the state~$\ket{\concRand}$.
	Furthermore, measuring the designated output qubit after applying~$U_\sigma$ gives
	\begin{equation}
		\Pr\left[
			\predictorGuess
			=
			\concRand\cdot\concDecomm_\sigma
			\;\middle|\;
			\Sigma=\sigma
		\right]
		=
		\frac12+\beta_\sigma
		\; .
	\end{equation}

	Thus, for every~$\sigma\in\mathcal C_+$, the unitary~$U_\sigma$, its inverse~$U_\sigma^\dagger$, and the auxiliary state~$\rho_\sigma$ define a quantum inner-product query with hidden string~$\concDecomm_\sigma$ and bias~$\beta_\sigma$.
	Although~$\rho_\sigma$ may be mixed, the quantum Goldreich--Levin theorem applies by purifying~$\rho_\sigma$ and extending~$U_\sigma$ by the identity on the purifying register.
	The extractor itself acts only on the registers accessible to~$\guesser$ and does not require access to~$\eveReg$.

	By the quantum Goldreich--Levin theorem~\cite[Theorem~2.2]{brakerski2023simple}, based on~\cite{qgl}, there is a fixed quantum polynomial-time extraction procedure that uses~$U_\sigma$ and~$U_\sigma^\dagger$ and outputs a string
	\begin{equation}
		\widehat\concDecomm
		\in
		\{0,1\}^{2\ell}
	\end{equation}
	satisfying
	\begin{equation}
		\label{eq:quantum_qgl_extraction}
		\Pr\left[
			\widehat\concDecomm
			=
			\concDecomm_\sigma
			\;\middle|\;
			\Sigma=\sigma
		\right]
		\geq
		4\beta_\sigma^2
		\; .
	\end{equation}
	This is precisely the concatenated-string application of quantum Goldreich--Levin used in~\cite[Section~5.3]{brakerski2023simple}.

	The binding adversary~$\guesser$ applies this same extraction procedure to its single sampled prefix~$\Sigma=\sigma$, without deciding whether~$\sigma$ belongs to~$\continuingPrefixes$ or~$\mathcal C_+$.
	It parses the resulting string as
	\begin{equation}
		\widehat\concDecomm
		=
		\widehat\decomm_0
		\mathbin\Vert
		\widehat\decomm_1
		\; .
	\end{equation}
	If
	\begin{equation}
		\open(\bctrans^\sigma,0,\widehat\decomm_0)
		=
		\top
		\qquad\text{and}\qquad
		\open(\bctrans^\sigma,1,\widehat\decomm_1)
		=
		\top
		\; ,
	\end{equation}
	then~$\guesser$ outputs~$(\widehat\decomm_0,\widehat\decomm_1)$ in the binding experiment.
	Otherwise, it outputs an arbitrary pair of strings.

	Whenever~$\sigma\in\mathcal C_+$ and $\widehat\concDecomm=\concDecomm_\sigma$, both opening checks accept.
	Therefore,
	\begin{equation}
		\label{eq:quantum_binding_success_from_biases}
		\Pr\left[
			\bindExp_{\guesser,\bc}(1^\security)
			=
			1
		\right]
		\geq
		4
		\sum_{\sigma\in\mathcal C_+}
		\mu(\sigma)\beta_\sigma^2
		\; .
	\end{equation}

	By the Cauchy--Schwarz inequality,
	\begin{align}
		\left(
			\sum_{\sigma\in\mathcal C_+}
			\mu(\sigma)\beta_\sigma
		\right)^2
		&\leq
		\left(
			\sum_{\sigma\in\mathcal C_+}
			\mu(\sigma)
		\right)
		\left(
			\sum_{\sigma\in\mathcal C_+}
			\mu(\sigma)\beta_\sigma^2
		\right)
		\nonumber\\
		&\leq
		\pce
		\sum_{\sigma\in\mathcal C_+}
		\mu(\sigma)\beta_\sigma^2
		\; .
		\label{eq:quantum_bias_cauchy_schwarz}
	\end{align}
	Equation~\eqref{eq:quantum_positive_bias_mass} implies that~$\pce>0$.
	Combining Equations~\eqref{eq:quantum_binding_success_from_biases} and~\eqref{eq:quantum_bias_cauchy_schwarz}, we obtain
	\begin{align}
		\Pr\left[
			\bindExp_{\guesser,\bc}(1^\security)
			=
			1
		\right]
		&\geq
		\frac{4}{\pce}
		\left(
			\sum_{\sigma\in\mathcal C_+}
			\mu(\sigma)\beta_\sigma
		\right)^2
		\nonumber\\
		&\geq
		\frac{4\delta(\security)^2}{\pce}
		\nonumber\\
		&\geq
		4\delta(\security)^2
		\nonumber\\
		&\geq
		4\varepsilon_0(\security)^2
		\; ,
	\end{align}
	where the penultimate inequality uses~$\pce\leq1$.
	The final quantity is non-negligible on the infinite set~$\mathcal I$.
	This contradicts \pq{} computational binding.

	We conclude that, for the fixed pair~$(\prover^*,\qadversary)$, there exists a negligible function~$\nu_{\prover^*,\qadversary}$ such that
	\begin{equation}
		\Pr\left[
			\qadversary(\transStaterv)
			=
			\virtualInputRV
			\;\wedge\;
			\contEvent
		\right]
		\leq
		\frac12\Pr[\contEvent]
		+
		\nu_{\prover^*,\qadversary}(\security)
		\; .
	\end{equation}
\end{proof}

This quantum hidden-input property enables us to certify entropy in the prover's Phase~B output.

\newcommand{\entthmref}{%
    \hyperref[theorem:entropy_certification]{%
        Theorem~\ref*{thm:main}%
    }%
}

\begin{mainrestate}
    \phantomsection
    \label{theorem:entropy_certification}
    Let~$\bc$ be a \shbc{} that is \pq{} computationally binding in the sense of Definition~\ref{def:post_quantum_binding}.
    Let~$\prover^*$ be a nonuniform \qpt{} prover with polynomial-size quantum advice that may be entangled with~$\eveReg$ but is independent of the verifier's fresh random coins.
    Let~$\pce(\security)=\Pr[\contEvent]>0$.
    Conditioned on~$\contEvent$, suppose that
    \begin{equation}
        \Pr[\mathsf{accept}\mid\contEvent]
        \geq
        \chshopt-\varepsilon
        \; ,
        \qquad
        \varepsilon\geq0.
    \end{equation}
    Then, there exist a universal constant~$\entropyConst>0$ and a negligible function~$\nu_{\mathrm{hid}}$ such that
    \begin{equation}
        \label{eq:entropy_certification_formal}
        \vne\!\left(
            \bellOutcomeRV
            \mid
            \transRV,\bellChallengeRV,\eveReg
        \right)_{\rho\mid\contEvent}
        \geq
        \entropyConst
        -
        \O\!\left(\sqrt{\varepsilon}\right)
        -
        \frac{\nu_{\mathrm{hid}}(\security)}{\pce(\security)}
        \; .
    \end{equation}
    In particular, if~$\pce(\security)\geq1/\poly(\security)$, then the final term is negligible.
\end{mainrestate}

\paragraph{Quantitative and modular proof routes.}
Appendix~\ref{app:entropy_proof} proves the full quantitative statement above for every positive continuation probability, including the explicit factor~$\nu_{\mathrm{hid}}/\pce$ and the fixed-challenge bound used in Section~\ref{sec:eat}.
When~$\pce(\security)\geq1/\poly(\security)$, \virtualHidingRef{quantum} also implies negligible conditional virtual-input leakage.
In that regime, the computational Bell-inequality framework of~\cite{cbi} gives an alternative modular derivation of the entropy-certification implication.

\begin{proof}[Proof of \entthmref{}]
    The quantitative statement is proved in Appendix~\ref{app:entropy_proof}.
    Specifically, the final random-challenge entropy inequality in Appendix~\ref{app:entropy_proof} gives Equation~\eqref{eq:entropy_certification_formal}, and Equation~\eqref{eq:rc_fixed_challenge_entropy_corrected} gives the stronger fixed-challenge estimate used for entropy accumulation.
    If~$\pce(\security)\geq1/\poly(\security)$, the alternative derivation from~\cite{cbi} follows from the negligible conditional leakage obtained by dividing the weighted hiding bound of \virtualHidingRef{quantum} by~$\pce(\security)$.
\end{proof}

\section{Entropy Accumulation}
\label{sec:eat}

We now use the single-round conditional entropy bound of \entthmref{} to obtain a lower bound on the smooth min-entropy produced by a sequential protocol.
Our construction follows the spot-checking and entropy-accumulation strategy of~\cite{merkulov2023entropy}.
The essential point is that each entropy-accumulation channel contains one \emph{complete} execution of the single-round protocol, including Phases~A1, A2, and~B.
Consequently, the fresh commitment and hashing randomness used in round~$i$ is sampled inside the channel~$\mathcal{M}_i$ and is independent of its input state.
This allows the computational hidden-input guarantee to be applied uniformly to every polynomial-size input state of the round channel.

Throughout this section, the adversary's register~$\eveReg$ is not acted upon during the protocol.
The public side information released after the protocol consists of the round-type variables, the Phase-B challenges, and the test data defined below.
The Phase-A transcript of each round is used internally by the verifier to compute the virtual variables and the score, and may then be erased.
It is not included in the public side information of the sequential protocol.
This convention is necessary for the standard EAT formulation used below.
If the complete Phase-A transcript is to be released, an additional leakage argument or a correspondingly enlarged EAT output register is required.
The single-round entropy bound is proved even after conditioning on the Phase-A transcript, so tracing out this private register does not weaken the per-round entropy estimate.

\subsection{Sequential Spot-Checking Protocol}
\label{subsec:sequential_spot_checking_protocol}

Fix a security parameter~$\security\in\N$, a number of rounds~$n=\poly(\security)$, a test probability~$\generationRoundBernoulli\in(0,1)$, a test-frequency width~$\delta_{\mathrm{test}}\in(0,\generationRoundBernoulli)$, a continuation threshold~$p_{\mathrm{cont}}^{\min}>0$, a CHSH threshold~$\omega_{\min}\in(0,\chshopt)$, and a smoothing parameter~$\smoothingParam\in(0,1)$.
We use the convention that~$\generationRoundRV_i=0$ denotes a test round and~$\generationRoundRV_i=1$ denotes a generation round, with
\begin{equation}
    \Pr[\generationRoundRV_i=0]
    =
    \generationRoundBernoulli
    \; .
\end{equation}

The round-type variable~$\generationRoundRV_i$ is not revealed to the prover.
On a test round, the verifier samples the Phase-B challenge uniformly.
On a generation round, the verifier fixes the Phase-B challenge to~$1$, which is the challenge for which Equation~\eqref{eq:rc_fixed_challenge_entropy_corrected} gives a direct entropy lower bound.
The prover receives only the challenge~$\bellChallengeRV_i$ and therefore applies the same response operation whenever~$\bellChallengeRV_i=1$, regardless of whether the round is a test round or a generation round.

\begin{algorithm}[H]
    \caption{Sequential randomness certification from \shbcshort{}}
    \label{protocol:sequential_rand_certification}

    \small

    \begin{enumerate}[leftmargin=24pt]
        \item
        For every~$i\in[n]$, execute one complete instance of the single-round protocol as follows.

        \begin{enumerate}[label=(\alph*), leftmargin=24pt]
            \item
            Execute Phases~A1 and~A2 of Protocol~\ref{protocol:single_round} using fresh verifier randomness.
            Record the Phase-A transcript~$\transRV_i$ and continuation indicator~$\contEvent_i$.

            \item
            Sample~$\generationRoundRV_i\leftarrow\Ber(1-\generationRoundBernoulli)$, where~$\generationRoundRV_i=0$ is a test round, and keep~$\generationRoundRV_i$ hidden from the prover.

            \item
            If~$\generationRoundRV_i=0$, sample~$\bellChallengeRV_i\leftarrow\{0,1\}$ uniformly.
            If~$\generationRoundRV_i=1$, set~$\bellChallengeRV_i=1$.

            \item
            If~$\contEvent_i=1$, send~$\bellChallengeRV_i$ to the prover, receive~$\bellOutcomeRV_i\in\{0,1\}$, and compute the CHSH winning indicator~$\winRV_i$.
            If~$\contEvent_i=0$, send no Phase-B message and set~$\bellOutcomeRV_i=\winRV_i=0$.

            \item
            Compute the tagged output~$O_i$, score~$Q_i$, generation symbol~$Z_i$, and test symbol~$D_i$ according to Equations~\eqref{eq:tagged_round_output}--\eqref{eq:test_data_symbol}.
        \end{enumerate}

        \item
        Compute~$N_{\mathrm{test}}$, $N_{\mathrm{cont}}$, and~$N_{\mathrm{win}}$ as in Equation~\eqref{eq:sequential_test_counts}.
        Abort unless Equations~\eqref{eq:test_frequency_condition}--\eqref{eq:chsh_condition} hold.
        Let~$\Omega$ denote the event that the verifier does not abort.

        \item
        Output the tagged string~$Z_1^n$.
        Treat~$D_1^n$, $\generationRoundRV_1^n$, and~$\bellChallengeRV_1^n$ as public side information.
    \end{enumerate}
\end{algorithm}

For every round, define
\begin{equation}
    \winRV_i
    \coloneqq
    \begin{cases}
        \indicator{
            \virtualOutputRV_i
            \oplus
            \bellOutcomeRV_i
            =
            \virtualInputRV_i\bellChallengeRV_i
        }
        &
        \text{if~$\contEvent_i=1$}
        \; ,
        \\
        0
        &
        \text{if~$\contEvent_i=0$}
        \; .
    \end{cases}
\end{equation}
Define the tagged round output~$O_i$ by
\begin{equation}
    O_i
    \coloneqq
    \begin{cases}
        \mathsf{g}_{\bot}
        &
        \text{if~$\generationRoundRV_i=1$ and~$\contEvent_i=0$}
        \; ,
        \\
        \mathsf{g}_{\bellOutcomeRV_i}
        &
        \text{if~$\generationRoundRV_i=1$ and~$\contEvent_i=1$}
        \; ,
        \\
        \mathsf{t}_{0}
        &
        \text{if~$\generationRoundRV_i=0$ and~$\contEvent_i=0$}
        \; ,
        \\
        \mathsf{t}_{1}
        &
        \text{if~$\generationRoundRV_i=0$,~$\contEvent_i=1$, and~$\winRV_i=0$}
        \; ,
        \\
        \mathsf{t}_{2}
        &
        \text{if~$\generationRoundRV_i=0$,~$\contEvent_i=1$, and~$\winRV_i=1$}
        \; .
    \end{cases}
    \label{eq:tagged_round_output}
\end{equation}
Define the score symbol~$Q_i\in\{\bot,0,1,2\}$ by
\begin{equation}
    Q_i
    \coloneqq
    \begin{cases}
        \bot
        &
        \text{if~$O_i\in\{\mathsf{g}_{\bot},\mathsf{g}_0,\mathsf{g}_1\}$}
        \; ,
        \\
        j
        &
        \text{if~$O_i=\mathsf{t}_j$ for~$j\in\{0,1,2\}$}
        \; .
    \end{cases}
    \label{eq:score_symbol}
\end{equation}
For parameter estimation, set
\begin{equation}
    N_{\mathrm{test}}
    \coloneqq
    \sum_{i=1}^{n}
    \indicator{Q_i\neq\bot}
    \; ,
    \qquad
    N_{\mathrm{cont}}
    \coloneqq
    \sum_{i=1}^{n}
    \indicator{Q_i\in\{1,2\}}
    \; ,
    \qquad
    N_{\mathrm{win}}
    \coloneqq
    \sum_{i=1}^{n}
    \indicator{Q_i=2}
    \; .
    \label{eq:sequential_test_counts}
\end{equation}
The passing conditions are
\begin{align}
    n
    \left(
        \generationRoundBernoulli
        -
        \delta_{\mathrm{test}}
    \right)
    &\leq
    N_{\mathrm{test}}
    \leq
    n
    \left(
        \generationRoundBernoulli
        +
        \delta_{\mathrm{test}}
    \right)
    \; ,
    \label{eq:test_frequency_condition}
    \\
    N_{\mathrm{cont}}
    &\geq
    p_{\mathrm{cont}}^{\min}
    N_{\mathrm{test}}
    \; ,
    \label{eq:continuation_condition}
    \\
    N_{\mathrm{win}}
    &\geq
    \omega_{\min}
    N_{\mathrm{cont}}
    \; .
    \label{eq:chsh_condition}
\end{align}
Finally, define
\begin{equation}
    Z_i
    \coloneqq
    \begin{cases}
        \bot
        &
        \text{if~$O_i\notin\{\mathsf{g}_0,\mathsf{g}_1\}$}
        \; ,
        \\
        b
        &
        \text{if~$O_i=\mathsf{g}_b$ for~$b\in\{0,1\}$}
        \; ,
    \end{cases}
    \label{eq:tagged_generation_output}
\end{equation}
and
\begin{equation}
    D_i
    \coloneqq
    \begin{cases}
        \bot
        &
        \text{if~$Q_i=\bot$}
        \; ,
        \\
        Q_i
        &
        \text{if~$Q_i\in\{0,1,2\}$}
        \; .
    \end{cases}
    \label{eq:test_data_symbol}
\end{equation}

The tagged alphabet~$\{\bot,0,1\}$ makes the output length independent of the number of successful generation rounds.
An injective binary encoding of~$Z_1^n$ preserves its conditional smooth min-entropy.
The test-data string~$D_1^n$ is included in the conditioning system of the final theorem, so the result remains valid even if the verifier reveals all test outcomes after the protocol.

\subsection{EAT Channels}
\label{subsec:sequential_eat_channels}

Let~$R_{i-1}$ contain the prover's private memory immediately before the beginning of round~$i$.
All verifier work registers are freshly initialized inside the round channel, and private transcript registers that are not needed in later rounds are traced out at the end of the round.
In particular, the complete Phase-A execution of round~$i$ is part of the channel~$\mathcal{M}_i$ and is not part of its input.
Let~$\eveReg$ denote arbitrary polynomial-size quantum side information that may be entangled with~$R_0$ and is independent of the verifier's fresh coins.

\begin{definition}[Sequential round channels]
    \label{def:sequential_eat_channels}
    For every~$i\in[n]$, let
    \begin{equation}
        \mathcal{M}_i:
        R_{i-1}
        \longrightarrow
        R_i O_i S_i Q_i
        \label{eq:complete_round_eat_channel}
    \end{equation}
    be the CPTP map describing the complete~$i$-th round of Protocol~\ref{protocol:sequential_rand_certification}, where
    \begin{equation}
        S_i
        \coloneqq
        \generationRoundRV_i\bellChallengeRV_i
        \; ,
    \end{equation}
    the output~$O_i$ is defined in Equation~\eqref{eq:tagged_round_output}, and the score~$Q_i$ is defined in Equation~\eqref{eq:score_symbol}.
\end{definition}

\begin{lemma}[EAT-channel conditions]
    \label{lemma:complete_round_eat_channels}
    The channels~$\{\mathcal{M}_i\}_{i\in[n]}$ from Definition~\ref{def:sequential_eat_channels} are EAT channels in the sense of Definition~\ref{def:eat_channels}.
    Moreover, the output dimension is~$d_O=6$.
\end{lemma}

\begin{proof}
    The tagged output~$O_i$ is classical with alphabet
    \begin{equation}
        \mathcal{O}
        =
        \left\{
            \mathsf{g}_{\bot},
            \mathsf{g}_0,
            \mathsf{g}_1,
            \mathsf{t}_0,
            \mathsf{t}_1,
            \mathsf{t}_2
        \right\}
        \; ,
    \end{equation}
    and hence~$d_O=6$.
    The score~$Q_i$ is a deterministic function of~$O_i$, so it can be computed from~$O_iS_i$ without changing the marginal state on~$O_iS_i$.

    It remains to verify the Markov condition.
    The pair~$S_i=(\generationRoundRV_i,\bellChallengeRV_i)$ is sampled using fresh verifier randomness inside round~$i$.
    Its distribution is independent of~$R_{i-1}$, of the previous outputs, and of~$\eveReg$.
    Therefore, for every initial state~$\rho_{R_0\eveReg}$ and every~$i\in[n]$, the final state satisfies
    \begin{equation}
        O_1^{i-1}
        \longleftrightarrow
        S_1^{i-1}\eveReg
        \longleftrightarrow
        S_i
        \; .
    \end{equation}
    Thus all the conditions in Definition~\ref{def:eat_channels} are satisfied.
\end{proof}

The placement of Phase~A inside~$\mathcal{M}_i$ is essential.
If the EAT channel began only after Phase~A, its input could be an arbitrary state already containing the virtual input, and the computational hidden-input theorem would not provide a lower bound for the infimum in Definition~\ref{def:min_tradeoff_function}.
The complete-round construction avoids this problem because all verifier coins on which the virtual input depends are sampled after the channel input has been fixed.

\subsection{Single-Round Entropy and an Affine Min-Tradeoff Function}
\label{subsec:sequential_min_tradeoff}

We first rewrite the fixed-challenge entropy estimate in a form convenient for entropy accumulation.
Equation~\eqref{eq:rc_fixed_challenge_entropy_corrected} implies that there exist universal constants~$h_{\mathrm{gen}}>0$ and~$c_{\mathrm{gen}}>0$ and a negligible function~$\nu_{\mathrm{gen}}$ such that
\begin{equation}
    h_{\mathrm{gen}}
    =
    \frac{\goodBlockEntropy}{2(1+10\sqrt{2})}
    \; ,
    \label{eq:generation_entropy_constant}
\end{equation}
and the following holds.
For every polynomial-size auxiliary quantum system~$R'$, every input state independent of the verifier's fresh round coins, and every efficient round strategy, let
\begin{equation}
    p
    \coloneqq
    \Pr[\contEvent]
    \; ,
    \qquad
    \omega
    \coloneqq
    \Pr[\mathsf{CHSH\ win}\mid\contEvent]
    \; .
\end{equation}
Then
\begin{equation}
    \vne\!\left(
        \bellOutcomeRVAt{1}
        \mid
        \transRV,R'
    \right)_{\rho\mid\contEvent}
    \geq
    h_{\mathrm{gen}}
    -
    c_{\mathrm{gen}}
    \sqrt{
        \left(
            \chshopt-\omega
        \right)_{+}
    }
    -
    \frac{\nu_{\mathrm{gen}}(\security)}{p}
    \; ,
    \label{eq:fixed_challenge_entropy_eat_form}
\end{equation}
where~$(x)_{+}\coloneqq\max\{x,0\}$.
When~$p=0$, all expressions below are understood by continuity and the corresponding generation output is deterministic.

Define
\begin{equation}
    \Delta_{*}
    \coloneqq
    \left(
        \frac{h_{\mathrm{gen}}}{c_{\mathrm{gen}}}
    \right)^2
    \; ,
    \qquad
    \omega_{0}
    \coloneqq
    \chshopt-\Delta_{*}
    \; ,
    \label{eq:tradeoff_root}
\end{equation}
and
\begin{equation}
    s_{*}
    \coloneqq
    \frac{h_{\mathrm{gen}}}{1-\omega_0}
    =
    \frac{h_{\mathrm{gen}}}{1-\chshopt+\Delta_*}
    \; .
    \label{eq:tradeoff_slope}
\end{equation}
By increasing~$c_{\mathrm{gen}}$ if necessary, we may assume that
\begin{equation}
    0<\Delta_{*}\leq\chshopt
    \qquad\text{and}\qquad
    \frac{h_{\mathrm{gen}}}{c_{\mathrm{gen}}}
    \leq
    \sqrt{1-\chshopt}
    \; .
    \label{eq:tradeoff_constant_condition}
\end{equation}

\begin{lemma}[Global affine lower bound]
    \label{lemma:global_affine_entropy_bound}
    For every~$\omega\in[0,1]$,
    \begin{equation}
        \max\!\left\{
            0,
            h_{\mathrm{gen}}
            -
            c_{\mathrm{gen}}
            \sqrt{
                \left(
                    \chshopt-\omega
                \right)_{+}
            }
        \right\}
        \geq
        s_{*}
        \left(
            \omega-\omega_0
        \right)
        \; .
        \label{eq:global_affine_entropy_bound}
    \end{equation}
\end{lemma}

\begin{proof}
    For~$\omega\leq\omega_0$, the right-hand side of Equation~\eqref{eq:global_affine_entropy_bound} is nonpositive, whereas the left-hand side is nonnegative.
    For~$\omega\in[\omega_0,\chshopt]$, define
    \begin{equation}
        g(\omega)
        \coloneqq
        h_{\mathrm{gen}}
        -
        c_{\mathrm{gen}}
        \sqrt{\chshopt-\omega}
        \; .
    \end{equation}
    We have~$g(\omega_0)=0$ and
    \begin{equation}
        g'(\omega)
        =
        \frac{c_{\mathrm{gen}}}{2\sqrt{\chshopt-\omega}}
        \geq
        \frac{c_{\mathrm{gen}}^2}{2h_{\mathrm{gen}}}
        \geq
        s_*
        \; ,
    \end{equation}
    where the final inequality follows from Equation~\eqref{eq:tradeoff_constant_condition}.
    Hence~$g(\omega)\geq s_*(\omega-\omega_0)$ throughout this interval.
    Finally, for~$\omega\in[\chshopt,1]$, the left-hand side equals~$h_{\mathrm{gen}}$, while
    \begin{equation}
        s_*(\omega-\omega_0)
        \leq
        s_*(1-\omega_0)
        =
        h_{\mathrm{gen}}
        \; .
    \end{equation}
\end{proof}

Let~$q$ be a probability distribution over~$\{\bot,0,1,2\}$.
Define the unconditional continuing-test and winning-test masses
\begin{equation}
    r(q)
    \coloneqq
    q(1)+q(2)
    \; ,
    \qquad
    w(q)
    \coloneqq
    q(2)
    \; .
    \label{eq:score_masses}
\end{equation}
For every distribution generated by a round channel,
\begin{equation}
    q(0)+q(1)+q(2)
    =
    \generationRoundBernoulli
    \; ,
\end{equation}
and, whenever~$r(q)>0$,
\begin{equation}
    p(q)
    =
    \frac{r(q)}{\generationRoundBernoulli}
    \; ,
    \qquad
    \omega(q)
    =
    \frac{w(q)}{r(q)}
    \; .
    \label{eq:score_to_protocol_parameters}
\end{equation}
The first identity follows because the continuation event is determined before the hidden round-type variable is sampled.
The second identity is the CHSH winning probability conditioned on a continuing test round.

\begin{lemma}[Affine min-tradeoff function]
    \label{lemma:affine_min_tradeoff}
    There exists a negligible function~$\nu_{\mathrm{EAT}}$ such that
    \begin{equation}
        f_{\min}(q)
        \coloneqq
        \frac{1-\generationRoundBernoulli}{\generationRoundBernoulli}
        s_*
        \left[
            w(q)
            -
            \omega_0 r(q)
        \right]
        -
        \left(
            1-\generationRoundBernoulli
        \right)
        \nu_{\mathrm{EAT}}(\security)
        \label{eq:affine_min_tradeoff}
    \end{equation}
    is a min-tradeoff function for the channels~$\{\mathcal{M}_i\}_{i\in[n]}$.
\end{lemma}

\begin{proof}
    Fix~$i\in[n]$, a distribution~$q$, and an input state~$\sigma_{R_{i-1}R'}$ satisfying
    \begin{equation}
        \mathcal{M}_i(\sigma)_{Q_i}
        =
        q
        \; .
    \end{equation}
    Since~$O_i$ and~$S_i$ are classical, conditioning on the round type gives
    \begin{align}
        \vne(O_i\mid S_iR')_{\mathcal{M}_i(\sigma)}
        &\geq
        \left(
            1-\generationRoundBernoulli
        \right)
        \vne\!\left(
            Z_i
            \mid
            R',
            \generationRoundRV_i=1,
            \bellChallengeRV_i=1
        \right)
        \; .
        \label{eq:discard_test_entropy}
    \end{align}
    The symbol~$Z_i$ contains the continuation flag through the value~$\bot$.
    Consequently,
    \begin{equation}
		\label{eq:weight_continuation_entropy}
        \vne\!\left(
            Z_i
            \mid
            R',
            \generationRoundRV_i=1,
            \bellChallengeRV_i=1
        \right)
        \geq
        p(q)
        \vne\!\left(
            \bellOutcomeRVAt{1}
            \mid
            \transRV_i,R',\contEvent_i=1
        \right)
        \; .
    \end{equation}
    Here we also used the fact that conditioning on the classical transcript~$\transRV_i$ can only decrease the conditional entropy.

    Since~$\bellOutcomeRVAt{1}$ is classical, its conditional von Neumann entropy is nonnegative.
    Equation~\eqref{eq:fixed_challenge_entropy_eat_form} therefore also implies
    \begin{equation}
        \vne\!\left(
            \bellOutcomeRVAt{1}
            \mid
            \transRV_i,R',\contEvent_i=1
        \right)
        \geq
        \max\!\left\{
            0,
            h_{\mathrm{gen}}
            -
            c_{\mathrm{gen}}
            \sqrt{
                \left(
                    \chshopt-\omega(q)
                \right)_{+}
            }
        \right\}
        -
        \frac{\nu_{\mathrm{gen}}(\security)}{p(q)}
        \; .
        \label{eq:nonnegative_fixed_challenge_bound}
    \end{equation}
    Applying Equation~\eqref{eq:nonnegative_fixed_challenge_bound} and Lemma~\ref{lemma:global_affine_entropy_bound} gives
    \begin{align}
        &\vne(O_i\mid S_iR')_{\mathcal{M}_i(\sigma)}
        \nonumber\\
        &\qquad\geq
        \left(
            1-\generationRoundBernoulli
        \right)
        \left[
            p(q)s_*
            \left(
                \omega(q)-\omega_0
            \right)
            -
            \nu_{\mathrm{EAT}}(\security)
        \right]
        \nonumber\\
        &\qquad=
        f_{\min}(q)
        \; ,
    \end{align}
    where the final equality follows from Equation~\eqref{eq:score_to_protocol_parameters}.
    The function~$f_{\min}$ is affine and therefore differentiable and convex.
\end{proof}

The gradient of~$f_{\min}$ is constant.
Writing the coordinates of~$q$ in the order~$(\bot,0,1,2)$, we may take
\begin{equation}
    L_*
    \coloneqq
    \left\lVert
        \nabla f_{\min}
    \right\rVert_{\infty}
    =
    \frac{1-\generationRoundBernoulli}{\generationRoundBernoulli}
    s_*
    \max\left\{
        \omega_0,
        1-\omega_0
    \right\}
    \; .
    \label{eq:min_tradeoff_gradient}
\end{equation}

\subsection{Smooth Min-Entropy of the Sequential Output}
\label{subsec:sequential_smooth_min_entropy}

Let~$\widehat{\Omega}$ denote the set of empirical score distributions satisfying Equations~\eqref{eq:test_frequency_condition}--\eqref{eq:chsh_condition}.
This set is convex because all three conditions are linear in the empirical frequencies of~$Q_1^n$.
For every~$q\in\widehat{\Omega}$, we have
\begin{equation}
    r(q)
    \geq
    p_{\mathrm{cont}}^{\min}
    \left(
        q(0)+q(1)+q(2)
    \right)
    \geq
    p_{\mathrm{cont}}^{\min}
    \left(
        \generationRoundBernoulli
        -
        \delta_{\mathrm{test}}
    \right)
    \; ,
    \label{eq:pass_continuation_mass}
\end{equation}
and
\begin{equation}
    w(q)
    \geq
    \omega_{\min}r(q)
    \; .
    \label{eq:pass_winning_mass}
\end{equation}
Assume from now on that
\begin{equation}
    \omega_{\min}
    >
    \omega_0
    =
    \chshopt
    -
    \left(
        \frac{h_{\mathrm{gen}}}{c_{\mathrm{gen}}}
    \right)^2
    \; .
    \label{eq:positive_score_threshold}
\end{equation}
Combining Equations~\eqref{eq:affine_min_tradeoff}, \eqref{eq:pass_continuation_mass}, and~\eqref{eq:pass_winning_mass}, every~$q\in\widehat{\Omega}$ satisfies
\begin{equation}
    f_{\min}(q)
    \geq
    h_{\Omega}
    -
    \left(
        1-\generationRoundBernoulli
    \right)
    \nu_{\mathrm{EAT}}(\security)
    \; ,
    \label{eq:passing_tradeoff_rate}
\end{equation}
where
\begin{equation}
    h_{\Omega}
    \coloneqq
    \frac{1-\generationRoundBernoulli}{\generationRoundBernoulli}
    s_*
    \left(
        \omega_{\min}-\omega_0
    \right)
    p_{\mathrm{cont}}^{\min}
    \left(
        \generationRoundBernoulli
        -
        \delta_{\mathrm{test}}
    \right)
    \; .
    \label{eq:asymptotic_accumulation_rate}
\end{equation}

Let~$p_{\Omega}=\Pr[\Omega]>0$.
Applying Theorem~\ref{thm:eat_rate_form} with~$d_O=6$ and Equation~\eqref{eq:min_tradeoff_gradient} yields
\begin{equation}
    \vne_{\min}^{\smoothingParam}\!\left(
        O_1^n
        \mid
        \generationRoundRV_1^n,
        \bellChallengeRV_1^n,
        \eveReg
    \right)_{\rho\mid\Omega}
    \geq
    n h_{\Omega}
    -
    \mu_{\mathrm{EAT}}\sqrt{n}
    -
    n
    \left(
        1-\generationRoundBernoulli
    \right)
    \nu_{\mathrm{EAT}}(\security)
    \; ,
    \label{eq:full_tagged_output_entropy}
\end{equation}
where
\begin{equation}
    \mu_{\mathrm{EAT}}
    \coloneqq
    2
    \left(
        \log(13)
        +
        \left\lceil
            L_*
        \right\rceil
    \right)
    \sqrt{
        1
        -
        2\log
        \left(
            \smoothingParam p_{\Omega}
        \right)
    }
    \; .
    \label{eq:sequential_eat_second_order}
\end{equation}

The tagged symbol~$O_i$ is in one-to-one correspondence with the pair~$(Z_i,D_i)$ defined in Equations~\eqref{eq:tagged_generation_output} and~\eqref{eq:test_data_symbol}.
Conditioned on~$\generationRoundRV_1^n$, the test-data string~$D_1^n$ has at most~$3^{N_{\mathrm{test}}}$ possible values.
On~$\Omega$, Equation~\eqref{eq:test_frequency_condition} therefore gives
\begin{equation}
    \log
    \left|
        \operatorname{supp}
        \left(
            D_1^n
            \mid
            \generationRoundRV_1^n,
            \Omega
        \right)
    \right|
    \leq
    n
    \left(
        \generationRoundBernoulli
        +
        \delta_{\mathrm{test}}
    \right)
    \log 3
    \; .
    \label{eq:test_data_dimension_bound}
\end{equation}
The dimension chain rule for smooth min-entropy and Equations~\eqref{eq:full_tagged_output_entropy}--\eqref{eq:test_data_dimension_bound} now give the desired bound on the generation output.

\begin{theorem}[Sequential smooth min-entropy]
    \label{thm:sequential_smooth_min_entropy}
    Let~$n=\poly(\security)$, and let an efficient sequential prover with polynomial-size quantum advice execute Protocol~\ref{protocol:sequential_rand_certification}.
    Assume that the underlying commitment is statistically hiding and is computationally binding against nonuniform \qpt{} senders with polynomial-size quantum advice independent of the verifier's fresh coins.
    Assume further that the complete Phase-A transcript~$\transRV_1^n$ is retained by the verifier or erased after scoring and is not released as public side information.
    The conditioning system is therefore~$(D_1^n,\generationRoundRV_1^n,\bellChallengeRV_1^n,\eveReg)$, where the first three registers are public and~$\eveReg$ is the adversary's quantum side information.
    Let~$p_{\Omega}=\Pr[\Omega]>0$, and suppose that Equation~\eqref{eq:positive_score_threshold} holds.
    Then
    \begin{align}
        &\vne_{\min}^{\smoothingParam}\!\left(
            Z_1^n
            \mid
            D_1^n,
            \generationRoundRV_1^n,
            \bellChallengeRV_1^n,
            \eveReg
        \right)_{\rho\mid\Omega}
        \nonumber\\
        &\qquad\geq
        n R_{\mathrm{seq}}
        -
        \mu_{\mathrm{EAT}}\sqrt{n}
        -
        n
        \left(
            1-\generationRoundBernoulli
        \right)
        \nu_{\mathrm{EAT}}(\security)
        \; ,
        \label{eq:sequential_smooth_min_entropy}
    \end{align}
    where
    \begin{equation}
        R_{\mathrm{seq}}
        \coloneqq
        h_{\Omega}
        -
        \left(
            \generationRoundBernoulli
            +
            \delta_{\mathrm{test}}
        \right)
        \log 3
        \; ,
        \label{eq:sequential_entropy_rate}
    \end{equation}
    and~$h_{\Omega}$ and~$\mu_{\mathrm{EAT}}$ are defined in Equations~\eqref{eq:asymptotic_accumulation_rate} and~\eqref{eq:sequential_eat_second_order}, respectively.
    No claim is made here about the entropy conditioned additionally on the complete transcript~$\transRV_1^n$; public release of that register requires a separate leakage argument or an enlarged EAT output register.
\end{theorem}

\begin{proof}
    Equation~\eqref{eq:full_tagged_output_entropy} provides a lower bound on the smooth min-entropy of the joint tagged string~$(Z_1^n,D_1^n)$.
    Revealing a classical register with support size at most~$M$ decreases the conditional smooth min-entropy by at most~$\log M$.
    Applying this dimension bound with~$M=3^{N_{\mathrm{test}}}$ and then using Equation~\eqref{eq:test_data_dimension_bound} proves Equation~\eqref{eq:sequential_smooth_min_entropy}.
\end{proof}

Since~$n=\poly(\security)$, the term~$n\nu_{\mathrm{EAT}}(\security)$ remains negligible in the security parameter.
The leading-order rate is positive whenever
\begin{equation}
    \frac{1-\generationRoundBernoulli}{\generationRoundBernoulli}
    s_*
    \left(
        \omega_{\min}-\omega_0
    \right)
    p_{\mathrm{cont}}^{\min}
    \left(
        \generationRoundBernoulli
        -
        \delta_{\mathrm{test}}
    \right)
    >
    \left(
        \generationRoundBernoulli
        +
        \delta_{\mathrm{test}}
    \right)
    \log 3
    \; .
    \label{eq:positive_sequential_rate_condition}
\end{equation}
As in~\cite{merkulov2023entropy}, decreasing the test probability improves the leading-order rate but increases the gradient in Equation~\eqref{eq:min_tradeoff_gradient} and therefore worsens the finite-size term.
The test probability must consequently be optimized as a function of~$n$, the continuation threshold, and the desired completeness and soundness errors.

\subsection{Completeness}
\label{subsec:sequential_completeness}

For completeness, consider the honest sequential strategy that executes the honest single-round strategy independently in every round.
Let
\begin{equation}
    p_{\mathrm{hon}}
    \coloneqq
    \Pr[\contEvent]
    \geq
    \frac{1}{10\hashNum}
    -
    \nu_{\mathrm{hon}}(\security)
    \; ,
    \label{eq:honest_sequential_continuation}
\end{equation}
where the inequality follows from Lemma~\ref{lemma:rc_phase_a_success}.
Conditioned on continuation, the honest strategy wins the CHSH test with probability~$\chshopt$.

\begin{proposition}[Honest completeness]
    \label{prop:sequential_honest_completeness}
    Suppose that
    \begin{equation}
        p_{\mathrm{cont}}^{\min}
        <
        p_{\mathrm{hon}}
        \qquad\text{and}\qquad
        \omega_{\min}
        <
        \chshopt
        \; .
    \end{equation}
    Then the honest strategy satisfies
    \begin{align}
        \Pr[\Omega]
        \geq
        1
        &-
        2\exp\!\left(
            -2n\delta_{\mathrm{test}}^2
        \right)
        \nonumber\\
        &-
        \exp\!\left(
            -2n
            \left(
                \generationRoundBernoulli
                -
                \delta_{\mathrm{test}}
            \right)
            \left(
                p_{\mathrm{hon}}
                -
                p_{\mathrm{cont}}^{\min}
            \right)^2
        \right)
        \nonumber\\
        &-
        \exp\!\left(
            -2n
            \left(
                \generationRoundBernoulli
                -
                \delta_{\mathrm{test}}
            \right)
            p_{\mathrm{cont}}^{\min}
            \left(
                \chshopt
                -
                \omega_{\min}
            \right)^2
        \right)
        \; .
        \label{eq:honest_completeness_bound}
    \end{align}
\end{proposition}

\begin{proof}
    The number of test rounds is binomial with parameters~$(n,\generationRoundBernoulli)$, so the first failure probability follows from Hoeffding's inequality.
    Conditioned on~$N_{\mathrm{test}}$, the honest continuation indicators in the test rounds are independent Bernoulli variables with mean~$p_{\mathrm{hon}}$.
    On the event~$N_{\mathrm{test}}\geq n(\generationRoundBernoulli-\delta_{\mathrm{test}})$, Hoeffding's inequality gives the second term in Equation~\eqref{eq:honest_completeness_bound}.
    Finally, conditioned on~$N_{\mathrm{cont}}$, the honest CHSH winning indicators are independent Bernoulli variables with mean~$\chshopt$.
    On the event~$N_{\mathrm{cont}}\geq p_{\mathrm{cont}}^{\min}N_{\mathrm{test}}$, a final application of Hoeffding's inequality gives the third term.
\end{proof}

Theorem~\ref{thm:sequential_smooth_min_entropy} gives the desired lower bound on the conditional smooth min-entropy of the complete tagged output string.
A quantum-proof strong seeded extractor may subsequently be applied to~$Z_1^n$ to obtain nearly uniform random bits, with output length determined by Equation~\eqref{eq:sequential_smooth_min_entropy} and the chosen extractor error.

\section{Instantiations}
\label{sec:instantiations}

\entthmref{} assumes the existence of a \pq{} \shbc{} scheme.
A natural question is therefore whether such commitments can be realized concretely and what verification complexity the resulting randomness-certification protocols achieve.
We give three instantiations.

The first uses the NOVY commitment~\cite{novy} from a \pq{} \owp{}.
Recent works~\cite{chen2026minicrypt,dupuis2026postquantum} establish quantum binding notions stronger than the ordinary \pq{} computational binding required by our theorem, while the commitment itself is perfectly hiding.
Its unique-opening structure makes the additional hash-isolation step unnecessary, but the verifier must invert the permutation to recover the two openings, so verification is inefficient.

The second uses a \pq{} trapdoor claw-free function family.
The common-image and injectivity properties yield a statistically hiding commitment, quantum claw-freeness provides \pq{} computational binding, and the trapdoor allows the verifier to efficiently recover both openings.
This is the only one of the three instantiations that yields an efficient verifier within our framework.

The third uses the standard hash-based commitment~$c=\oracle(\bit\Vert\rand)$ in the quantum random-oracle model.
With a sufficiently large randomness space, it is statistically hiding, while a double opening gives a random-oracle collision.
The random oracle supplies no trapdoor for the exact opening-recovery task, so verification is again inefficient.

These examples separate the cryptographic requirements for security from the algorithmic structure required for efficient verification.
The commitment constructions themselves are standard; the contribution here is their use through quantum virtual-input hiding to obtain randomness certification.

\subsection{\owpshort{}-based Construction}
\label{subsec:owp_commitment}

Let
\begin{equation}
    \Pi
    =
    \left\{
        \pi_\security:
        \{0,1\}^{n(\security)}
        \longrightarrow
        \{0,1\}^{n(\security)}
    \right\}_{\security\in\N}
\end{equation}
be a \pq{} \owp{} family.
We instantiate the commitment primitive with the NOVY commitment~\cite{novy}.
To commit to a bit~$b$, the sender samples~$x$ uniformly, computes~$y=\pi_\security(x)$, and runs the NOVY interactive-hashing protocol on~$y$.
The resulting transcript determines two strings~$y_0$ and~$y_1$, ordered lexicographically such that~$y_0<_{\mathrm{lex}}y_1$.
The sender identifies~$\beta\in\{0,1\}$ such that~$y=y_\beta$.
The sender completes the commitment by sending
\begin{equation}
    d=b\oplus\beta.
\end{equation}
To open to~$b$, the sender reveals~$x$, and the receiver checks
\begin{equation}
    \pi_\security(x)=y_{b\oplus d}.
\end{equation}

The NOVY commitment is perfectly hiding.
Moreover, recent concurrent and independent works establish \pq{} security of the underlying interactive-hashing/NOVY construction from any \pq{} \owp{}.
Chen, Wang, and Xie~\cite{chen2026minicrypt} prove quantum sum-binding.
Dupuis \emph{et al.}~\cite{dupuis2026postquantum} prove open-once binding and, using the unique-opening property, also obtain sum-binding and collapse-binding.
Chevalier, Holin, and Vu~\cite{chevalier2026interactive} prove \pq{} security of multi-output interactive hashing for arbitrary relations and recover perfectly hiding commitments from post-quantum \owp{}s as a direct corollary.
In particular, NOVY satisfies the ordinary \pq{} computational binding required by \entthmref{}.

This instantiation has a particularly simple opening structure.
For a fixed commitment transcript~$t$ containing~$(y_0,y_1,d)$ and for each~$b\in\{0,1\}$,
\begin{equation}
    \Decomm_{b,t}
    =
    \left\{
        \pi_\security^{-1}(y_{b\oplus d})
    \right\}.
\end{equation}
Thus there is exactly one valid opening for each bit on every valid transcript.
In Protocol~\ref{protocol:single_round}, one may therefore set~$\hashNum=1$, take~$\hashIndex=0$, and use the unique constant map into~$[1]$ for both hash functions.
The hashing step then leaves the two singleton opening sets unchanged, and the honest continuation probability is~$1$.

The verifier is nevertheless inefficient.
To compute the virtual input and output, it must recover
\begin{equation}
    x_0
    =
    \pi_\security^{-1}(y_d)
    \qquad\text{and}\qquad
    x_1
    =
    \pi_\security^{-1}(y_{1\oplus d}).
\end{equation}
A \pq{} \owp{} provides no trapdoor for these inversions.
Our verifier can recover the two openings only by exhaustive search over~$\{0,1\}^{n(\security)}$.
Consequently, this instantiation gives randomness certification in the standard model from \pq{} \owp{}s, with a computationally inefficient classical verifier.

\begin{corollary}[Randomness certification from \pq{} \owps{}]
\label{corollary:owp_randomness_certification}
Assuming the existence of \pq{} \owp{}s, Protocol~\ref{protocol:single_round} satisfies the conclusion of Theorem~\ref{thm:main} with an inefficient classical verifier.
\end{corollary}

\begin{proof}
	The NOVY commitment is perfectly hiding and satisfies \pq{} computational binding by~\cite{chen2026minicrypt,dupuis2026postquantum,chevalier2026interactive}.
	Theorem~\ref{thm:main} therefore applies.
	The verifier can implement the unique-opening checks and recover the openings by exhaustive inversion of~$\pi_\security$, as described above.
\end{proof}


\subsection{\tcfshort{}-based Construction}
\label{subsec:tcf_commitment}

We next instantiate the commitment primitive using an exact \pq{} \tcf{} family.
Recall from Subsection~\ref{subsec:tcf} that a \tcf{} family consists of a pair of injective functions with a common image, together with a trapdoor that enables efficient inversion.
We assume throughout this subsection that the family is quantum-secure, as defined in Definition~\ref{def:pq_tcf}, and that valid public keys are efficiently recognizable as specified below.

\begin{assumption}[Efficient validation of TCF public keys]
    \label{assumption:tcf_key_validation}
    There exists a deterministic polynomial-time algorithm~$\mathsf{KeyCheck}$ such that:
    \begin{enumerate}
        \item every key~$\key$ output by~$\gen(1^\security)$ satisfies~$\mathsf{KeyCheck}(\key)=1$; and
        \item for every string~$\key$ satisfying~$\mathsf{KeyCheck}(\key)=1$, the functions~$\tcfunc_{\key,0},\tcfunc_{\key,1}:\X\to\Y$ are efficiently evaluable and satisfy the injectivity and common-image properties of Definition~\ref{def:tcf}.
    \end{enumerate}
\end{assumption}

The commitment scheme is the natural \pq{} analogue of the standard bit commitment from a claw-free pair of functions~\cite{tcf_commitment,Chase2013}.
The public key check is needed because statistical hiding is defined against an arbitrary, computationally unbounded receiver, which may otherwise send a malformed key outside the key space.

\begin{definition}[TCF-based bit commitment with key validation]
    \label{def:tcf_commitment}
    Let~$\F$ be an exact \tcf{} pair family satisfying Assumption~\ref{assumption:tcf_key_validation}.
    We define the \tcf{}-based bit commitment scheme~$\tcfbc$ as follows.

    \paragraph{Setup.}
    The honest receiver~$\receiver$ samples
    \begin{equation}
        (\key,\trapdoor)\leftarrow\gen(1^\security)
        \; ,
    \end{equation}
    and sends the public key~$\key$ to the sender~$\sender$.\footnote{Although the trapdoor~$\trapdoor$ is not part of the commitment transcript, it is exactly the ingredient that enables efficient recovery of the two openings.}

	\paragraph{Commit Phase.}
	To commit to~$\bit\in\{0,1\}$, the sender first checks whether~$\key$ is valid using~$\mathsf{KeyCheck}$.
	If~$\mathsf{KeyCheck}(\key)=0$, it sends the distinguished symbol~$c\coloneqq\mathsf{invalid}$, stores~$d\coloneqq\bot$, and terminates the commit phase.
    Otherwise, it samples~$\preimage\leftarrow\X$ uniformly, computes
    \begin{equation}
        c
        \coloneqq
        \tcfunc_{\key,\bit}(\preimage),
    \end{equation}
    sends~$c$, and stores~$d\coloneqq\preimage$.

    \paragraph{Open Phase.}
    To open, the sender sends~$(\bit,\preimage)$.
    The receiver accepts iff
    \begin{equation}
        \mathsf{KeyCheck}(\key)=1,
        \qquad
        c\neq\mathsf{invalid},
        \qquad\text{and}\qquad
        \tcfunc_{\key,\bit}(\preimage)=c.
    \end{equation}
\end{definition}

\begin{lemma}[Post-quantum computational binding]
    \label{lemma:tcfs_are_binding}
    Let~$\F$ be an exact \pq{} \tcf{} family satisfying Assumption~\ref{assumption:tcf_key_validation}.
    Then the commitment scheme~$\tcfbc$ is \pq{} computationally binding in the sense of Definition~\ref{def:post_quantum_binding}.
\end{lemma}

\begin{proof}
    The honest receiver samples~$(\key,\trapdoor)\leftarrow\gen(1^\security)$, so~$\mathsf{KeyCheck}(\key)=1$.
    Suppose that a \qpt{} sender produces a commitment~$c$ and two accepted openings~$\preimage_0,\preimage_1$ to the bits~$0$ and~$1$.
    Acceptance implies
    \begin{equation}
        \tcfunc_{\key,0}(\preimage_0)
        =
        c
        =
        \tcfunc_{\key,1}(\preimage_1).
    \end{equation}
    Thus~$(\preimage_0,\preimage_1)$ is a claw for the sampled public key~$\key$.
    A sender that wins the binding experiment with non-negligible probability therefore yields a \qpt{} claw-finding algorithm with the same success probability, contradicting Definition~\ref{def:pq_tcf}.
\end{proof}

\renewcommand{\Im}{\text{Image}}
\begin{lemma}[Statistical hiding]
    \label{lemma:tcfs_are_hiding}
    Let~$\F$ be an exact \tcf{} family satisfying Assumption~\ref{assumption:tcf_key_validation}.
    Then the commitment scheme~$\tcfbc$ is perfectly hiding, and hence statistically hiding in the sense of Definition~\ref{def:hiding}.
    In particular, for every computationally unbounded receiver~$\receiver^*$,
    \begin{equation}
        \Pr\left[
            \hideExp_{\receiver^*,\tcfbc}(1^\security)=1
        \right]
        =
        \frac12.
    \end{equation}
\end{lemma}

\begin{proof}
    Fix an arbitrary receiver strategy~$\receiver^*$ and condition on its complete internal state and on the public-key string~$\key$ that it sends.
    If~$\mathsf{KeyCheck}(\key)=0$, the honest sender returns the distinguished symbol~$\mathsf{invalid}$ independently of the committed bit, so the two receiver views are identical.

    Suppose instead that~$\mathsf{KeyCheck}(\key)=1$.
    By Assumption~\ref{assumption:tcf_key_validation}, both functions are injective and have the same image
    \begin{equation}
        I_\key
        \coloneqq
        \Im(\tcfunc_{\key,0})
        =
        \Im(\tcfunc_{\key,1}).
    \end{equation}
    For either~$\bit\in\{0,1\}$, sampling~$\preimage\leftarrow\X$ uniformly and outputting~$\tcfunc_{\key,\bit}(\preimage)$ produces the uniform distribution on~$I_\key$.
    Thus the complete receiver view is identical for commitments to~$0$ and~$1$, for every key string chosen by~$\receiver^*$.
    Averaging over the receiver's internal state and key choice proves the claimed equality.
\end{proof}

\begin{corollary}[Randomness certification from exact \pq{} \tcfs{}]
    Assume the existence of an exact post-quantum \tcf{} family satisfying Assumption~\ref{assumption:tcf_key_validation}.
    Then Protocol~\ref{protocol:single_round}, instantiated with~$\tcfbc{}$, satisfies the conclusion of \entthmref{} with an efficient classical verifier.
\end{corollary}

\begin{proof}
    Lemmas~\ref{lemma:tcfs_are_binding} and~\ref{lemma:tcfs_are_hiding} verify the two cryptographic hypotheses of \entthmref{}.
    For an honestly generated key, the two branches have a common image and are injective.
    Given a commitment value~$c$, the verifier applies trapdoor inversion on both branches and verifies the returned preimages by reevaluation; it rejects if~$c$ is outside the common image.
    If~$c$ lies in the common image, injectivity guarantees exactly one valid opening in each branch.
    The verifier can therefore evaluate the continuation predicate and the virtual Bell variables efficiently.
\end{proof}

\begin{remark}
    This corollary concerns exact, noiseless TCF abstraction of Definition~\ref{def:tcf}.
    Standard LWE-based constructions use noisy trapdoor claw-free functions~\cite{mahadev2018classical,brakerski2021cryptographic}.
    The LWE error hides the secret shift, while additional output noise ensures that the distributions of the two branches are statistically close.
    Extending our commitment-based construction to such families requires an additional approximation and robust-opening analysis, which is not considered here.
\end{remark}

\subsection{\qromshort{}-based Construction}
\label{sec:qrom_commitment}

Our third instantiation is in the quantum random oracle model (\qrom{}).
In the \qrom{}, a uniformly sampled random function is made available as an oracle, and quantum parties are allowed to query this oracle coherently on superpositions of inputs.
This distinguishes the \qrom{} from the classical random oracle model, in which oracle queries are restricted to classical inputs; see~\cite{qrom_boneh,qrom_vs_rom,zhandry_hashes} for further discussion of the model and the distinction between classical and quantum oracle access.

We instantiate the commitment primitive using the standard hash-based commitment scheme~$c=\oracle(\bit\Vert\rand)$.\footnote{Random-oracle-based commitments of this form, also referred to as hash-based commitments, are well known; see, for example,~\cite{qrom_commitment,wonderful_qrom}.}
We choose the parameters so that the resulting commitment is statistically hiding, and show below that it additionally satisfies \pq{} computational binding in the \qrom{}.

\begin{definition}[Random-oracle bit commitment]
    \label{def:qrom_commitment}
    Let
    \begin{equation}
        \oracle:\{0,1\}^* \to \{0,1\}^{n}
    \end{equation}
    be a random oracle, where~$n=n(\security)=\Theta(\security)$.
    Let~$\rho=\rho(\security)$ be a randomness length satisfying
    \begin{equation}
        \rho-n=\omega(\log \security)
        \; .
    \end{equation}
    For concreteness, one may take~$n=\security$ and~$\rho=3\security$.

    We define the random-oracle-based bit commitment~$\robc{}=(\commit_{\rotag{}},\open_{\rotag{}})$ as follows.

    \paragraph{Commit Phase.}
    To commit to a bit~$\bit\in\{0,1\}$, the sender samples
    \begin{equation}
        \rand\leftarrow\{0,1\}^{\rho}
    \end{equation}
    uniformly at random and computes
    \begin{equation}
        c
        \coloneqq
        \oracle(\bit\Vert\rand)
        \in\{0,1\}^{n}
        \; .
    \end{equation}
    The commitment transcript is~$\trans=c$, and the decommitment is~$\decomm=\rand$.

    \paragraph{Open Phase.}
    To open~$c$ as the bit~$\bit$, the sender sends~$(\bit,\rand)$.
    The receiver accepts iff
    \begin{equation}
        \oracle(\bit\Vert\rand)=c
        \; .
    \end{equation}
\end{definition}

\begin{lemma}[Statistical hiding]
	\label{lemma:qrom_commitment_hiding}
	For a fixed oracle~$\oracle$ and bit~$\bit$, let~$D_\bit^\oracle$
	denote the distribution of the commitment transcript
	\begin{equation}
		c=\oracle(\bit\,\|\,R),
		\qquad
		R\leftarrow\{0,1\}^{\rho}.
	\end{equation}
	Then
	\begin{equation}
		\underset{\oracle}{\E}\,
		\Delta\!\left(
			D_0^\oracle,
			D_1^\oracle
		\right)
		\leq
		2^{(n-\rho)/2}.
	\end{equation}
	Consequently,~$\robc{}$ is statistically hiding whenever
	\begin{equation}
		\rho-n=\omega(\log\security).
	\end{equation}
\end{lemma}

\begin{proof}
	\newcommand{\uniform}{U}
	For a fixed oracle~$\oracle$ and bit~$\bit$, let~$D_\bit^\oracle$ denote the distribution of
	\begin{equation}
		\oracle(\bit\,\|\,R)\;,
		\qquad
		R\leftarrow\{0,1\}^{\rho} \;.
	\end{equation}
	Equivalently, for each~$c\in\{0,1\}^{n}$,
	\begin{equation}
		D_\bit^\oracle(c)
		=
		\frac{1}{2^\rho}
		\left|
			\left\{
				\rand\in\{0,1\}^{\rho}
				:
				\oracle(\bit\,\|\,\rand)=c
			\right\}
		\right| \; .
	\end{equation}
	Let~$\uniform_n$ denote the uniform distribution over~$\{0,1\}^{n}$.
	For each fixed~$\bit$, the values
	\begin{equation}
		\left\{
			\oracle(\bit\,\|\,\rand)
		\right\}_{r\in\{0,1\}^{\rho}}
	\end{equation}
	are independent uniform samples from~$\{0,1\}^{n}$.
	Applying the standard bound on the expected total variation distance between an empirical distribution and its underlying distribution~\cite[Lemma~5]{binomialbound}, with sample size~$2^\rho$ and support size~$2^n$, gives
	\begin{equation}
		\underset{\oracle}{\E}\,
		\Delta(D_\bit^\oracle,\uniform_n)
		\leq
		\frac12 \cdot 2^{(n-\rho)/2} \; .
	\end{equation}
	Therefore, by the triangle inequality,
	\begin{align}
		\underset{\oracle}{\E}\,
		\Delta(D_0^\oracle,D_1^\oracle)
		&\leq
		\underset{\oracle}{\E}\,
		\Delta(D_0^\oracle,\uniform_n)
		+
		\underset{\oracle}{\E}\,
		\Delta(D_1^\oracle,\uniform_n) \\
		&\leq
		2^{(n-\rho)/2} \; .
	\end{align}
	Since~$\rho-n=\omega(\log\security)$, this quantity is negligible.
	This bounds the distinguishing advantage even of an unbounded receiver given the entire random oracle, and therefore also of any \qpt{} receiver with quantum oracle access.
\end{proof}

\begin{lemma}[\qrom{} computational binding]
	\label{lemma:qrom_commitment_double_opening_binding}
	The commitment scheme~$\robc{}$ satisfies
	\pq{} computational binding against \qpt{} senders in the \qrom{}.
	More precisely, there exists a universal constant~$C>0$ such that every
	quantum sender making at most~$q=q(\security)$ oracle queries satisfies
	\begin{equation}
		\Pr\left[
			\bindExp_{\sender^*,\robc{}}(1^\security)
			=
			1
		\right]
		\leq
		\frac{
			C\left(q(\security)+1\right)^3
		}{
			2^n
		}.
	\end{equation}
\end{lemma}

\begin{proof}
	\newcommand{\collisionFinder}{\mathscr{C}}
	Let~$\sender^*$ be a quantum adversarial sender making at most
	$q=q(\security)$ queries to~$\oracle$.

	We construct a quantum collision-finding algorithm~$\collisionFinder$
	that runs~$\sender^*$ in the binding experiment.

	Suppose that~$\sender^*$ produces a commitment transcript~$c$ and
	two candidate openings~$\rand_0$ and~$\rand_1$.
	The algorithm~$\collisionFinder$ outputs
	\begin{equation}
		x_0
		\coloneqq
		0\,\|\,\rand_0
		\qquad\text{and}\qquad
		x_1
		\coloneqq
		1\,\|\,\rand_1.
	\end{equation}

	If the binding experiment accepts, then
	\begin{equation}
		\oracle(0\,\|\,\rand_0)
		=
		c
		=
		\oracle(1\,\|\,\rand_1).
	\end{equation}
	Moreover,
	\begin{equation}
		0\,\|\,\rand_0
		\neq
		1\,\|\,\rand_1,
	\end{equation}
	because the first input bits are distinct.
	Therefore, every successful double opening yields a collision in~$\oracle$.

	\newcommand{\zhandryConst}{C_{_Z}}
	A quantum algorithm~$\collisionFinder$ making at most~$q$ queries to a random function
	with range size~$2^n$ outputs a collision with probability at most
	\begin{equation}
		\Pr[\collisionFinder\text{ outputs a collision}]
		\leq
		\frac{
			\zhandryConst(q+1)^3
		}{
			2^n
		}
	\end{equation}
	for a universal constant~$\zhandryConst$~\cite[Theorem~3.1]{zhandry2015collision}.
	Hence,
	\begin{equation}
		\Pr\left[
			\bindExp_{\sender^*,\robc{}}(1^\security)
			=
			1
		\right]
		\leq
		\frac{
			\zhandryConst\left(q(\security)+1\right)^3
		}{
			2^n
		}.
	\end{equation}
	Since~$q(\security)$ is polynomially bounded and
	$n=\Theta(\security)$, this probability is negligible.
	Thus~$\robc{}$ satisfies \pq{} computational binding.
\end{proof}

\begin{corollary}[Randomness certification in the \qrom{}]
	\label{corollary:qrom_randomness_certification}
	In the \qrom{}, \iidProtocol{} instantiated with a \qrom{}-based commitment scheme satisfies the conclusion of Theorem~\ref{thm:main}.
\end{corollary}

\begin{proof}
	By Lemma~\ref{lemma:qrom_commitment_hiding}, the commitment scheme is
	statistically hiding.
	By Lemma~\ref{lemma:qrom_commitment_double_opening_binding}, it satisfies
	\pq{} computational binding as in
	Definition~\ref{def:post_quantum_binding}.
	Therefore, the hypotheses of Theorem~\ref{thm:main} are satisfied.
\end{proof}

\subsection{Efficient vs. Inefficient Verifiers}

The protocol presented in this work separates the cryptographic generation of the hidden qubit in Phase~A from the Bell-type test performed in Phase~B.
This separation makes it possible to identify precisely which parts of the construction determine whether the verifier can be implemented efficiently.

The \tcf{} instantiation of Subsection~\ref{subsec:tcf_commitment} yields an efficient verifier.
This does not, by itself, provide a new method for efficiently verifying quantum advantage or certifying randomness from \tcfs{}.
Several previous works already use \tcfs{} to achieve such goals, including~\cite{brakerski2021cryptographic}.
The significance of the construction here is instead conceptual.

In \iidProtocol{}, the verifier must determine whether the hash buckets selected in Phase~A contain unique valid openings for both commitment values and recover these openings.
For a generic \shbc{}, this task may be computationally infeasible.
The three instantiations illustrate different aspects of this requirement.
In the \owp{}-based construction, each commitment value already has a unique valid opening, so the additional hash-isolation step is unnecessary.
Nevertheless, recovering these openings requires inverting the \owp{}, which is computationally infeasible for an efficient verifier.
In the \qrom{} construction, the verifier must both certify uniqueness after hashing and recover the corresponding openings, which requires exhaustive search in our implementation.
In contrast, the \tcf{} construction provides a trapdoor that enables efficient recovery of the relevant preimages and verification of their uniqueness.

Thus, the \tcf{} instantiation demonstrates that the protocol developed in this work can be implemented efficiently when the underlying commitment primitive provides sufficient structure for recovering and certifying the relevant openings.
Its role is not to establish a new \tcf{}-based randomness-certification result, but to illustrate that the abstract protocol is compatible with efficient verification under this additional algorithmic structure.

\subsection{Verifier Efficiency and Cryptographic Barriers in the \qromshort{}}
\label{sec:qrom_verifier_barriers}

The \tcf{} and \qrom{} instantiations separate the cryptographic properties required for security from the algorithmic structure required for efficient verification.
Statistical hiding provides the balancing property needed for the honest Phase~A procedure to continue with non-negligible probability.
Computational binding against \qpt{} senders establishes quantum hiding of the virtual input.
Neither property, by itself, gives the verifier an efficient method for decoding the transcript.

In \iidProtocol{}, the verifier must decide whether the selected hash bucket contains exactly one valid opening for each commitment value.
For a commitment transcript~$\bctrans$, define
\begin{equation}
    \Decomm_{\bit,\bctrans}
    \coloneqq
    \left\{
        \decomm :
        \open(\bctrans,\bit,\decomm)
        =
        \top
    \right\} \; .
\end{equation}
The verifier continues only if
\begin{equation}
    \left|
        \Decomm_{\bit,\bctrans}
        \cap
        \hashfunc_\bit^{-1}(\hashval)
    \right|
    =
    1
\end{equation}
for both~$\bit\in\{0,1\}$.
It must then recover the two unique values in order to evaluate the virtual input.

For the \tcf{} instantiation, the trapdoor supplies precisely this additional algorithmic structure.
Injectivity ensures uniqueness, while trapdoor inversion allows the verifier to recover the relevant openings in polynomial time.
The resulting protocol therefore admits an efficient classical verifier.

For the random-oracle commitment
\begin{equation}
    c
    =
    \oracle(\bit\Vert \rand) \; ,
\end{equation}
the corresponding set of valid openings is
\begin{equation}
    \Decomm_{\bit,c}
    \coloneqq
    \left\{
        \rand\in\{0,1\}^{\rho} :
        \oracle(\bit\Vert \rand)
        =
        c
    \right\} \; .
\end{equation}
The random oracle provides no trapdoor or algebraic structure for determining the exact size of
\begin{equation}
    \Decomm_{\bit,c}
    \cap
    \hashfunc_\bit^{-1}(\hashval) \; .
\end{equation}
The verifier in our construction therefore performs an exhaustive search over the opening space.
This is the sole source of verifier inefficiency in the \qrom{} instantiation.

From an operational perspective, an inefficient verifier substantially weakens the usual motivation for classical verification.
A computationally unbounded classical verifier can simulate the prescribed honest quantum strategy in exponential time and sample from its output distribution without interacting with a quantum device.
Our \qrom{} instantiation is therefore primarily a complexity-theoretic result.
It shows that the cryptographic security properties needed for quantum virtual-input hiding exist in the \qrom{}, while leaving open the complexity of recognizing the corresponding transcript relation.

Recent work by Bartusek and Shalit~\cite{itay_shalit_2026} gives additional evidence that this efficiency question is nontrivial.
They define a test of non-commutation using an efficient quantum prover and an efficient classical verifier, and show that such tests imply classical-communication key agreement in appropriate completeness--soundness parameter regimes.
Since a full qubit test implies a test of non-commutation, their result applies a fortiori to efficiently verifiable qubit tests satisfying their formal requirements.

Our \qrom{} construction does not directly instantiate this implication.
The verifier in our protocol is inefficient, whereas both the test of non-commutation and the resulting key-agreement protocol in~\cite{itay_shalit_2026} require polynomial-time honest parties.
If one nevertheless substitutes our verifier into the first step of their construction, the verifier-side key-agreement party inherits the exponential computation of the uniqueness check.
The result is therefore, at most, a nonstandard agreement primitive with an inefficient honest party, rather than standard cryptographic key agreement.

Suppose instead that the verifier in \iidProtocol{} could be made polynomial-time while preserving the following properties:
\begin{enumerate}
    \item the protocol can be placed in the canonical form required by~\cite{itay_shalit_2026};
    \item its completeness and commuting-strategy soundness parameters lie in the regime covered by their key-agreement theorem;
    \item its soundness holds against \qpt{} algorithms with access to the same quantum random oracle and in the required advice model.
\end{enumerate}
Under these conditions, the resulting test would imply efficient classical-communication key agreement relative to the quantum random oracle.
In the basic bit-agreement construction, one party plays the efficient classical verifier and the other plays the efficient quantum prover.

This consequence is striking, but it does not presently yield a contradiction.
Austrin \emph{et al.}~\cite{austrin2022qromka} show that when one honest party is classical, the other is quantum, and the parties agree with probability exactly~$1$, a protocol making~$d$ random-oracle queries admits a classical attack using~$\O(d^2)$ oracle queries.
Their unconditional result therefore concerns perfect agreement and is stated in the query-complexity model.

By contrast, the standard key-agreement notion obtained in~\cite{itay_shalit_2026} guarantees agreement with probability
\begin{equation}
    1-\negl(\security) \; ,
\end{equation}
rather than agreement with probability exactly~$1$.
The perfect-agreement impossibility theorem therefore does not directly rule out the imperfect-completeness regime relevant to our protocol, although more recent work establishes impossibility in certain restricted imperfect-completeness settings~\cite{kitagawa2026impossibilityimperfectlycompletekey}.
Furthermore, translating a polynomial-query attack into a contradiction with a particular time-bounded security definition requires that the remaining computation of the attack satisfy the corresponding efficiency requirement.

If an efficient \qrom{} qubit test additionally yielded perfectly correct key agreement with one classical and one quantum party, then it would run directly into the query-complexity barrier of~\cite{austrin2022qromka}.
For the CHSH-type tests considered here, however, the honest winning probability is below~$1$, and the key-agreement amplification generally produces negligible rather than zero disagreement.
The existing impossibility theorem consequently leaves open precisely the regime relevant to our protocol.

The combined conclusions of~\cite{itay_shalit_2026,austrin2022qromka}, together with more recent impossibility results for restricted imperfect-completeness settings~\cite{kitagawa2026impossibilityimperfectlycompletekey}, should therefore be interpreted as evidence rather than as a general impossibility theorem.
An efficiently verifiable test of non-commutation in the plain \qrom{} would imply unexpectedly strong cryptographic structure.
Nevertheless, existing impossibility results do not rule out such tests in the general setting with negligible completeness error.
In particular, they do not establish that the inefficient verifier in our construction cannot be replaced by an efficient one, nor that qubit tests with efficient classical verifiers are impossible in the \qrom{}.

Our results isolate a more specific open problem.
Can the exact unique-opening check in \iidProtocol{} be replaced by an efficiently verifiable condition while preserving quantum hiding of the virtual input and non-negligible honest continuation probability?
A positive answer, together with the corresponding soundness-to-entropy guarantee, would yield an efficient randomness certification protocol in the plain \qrom{} and, subject to the formal conditions above, would also have consequences for classical-communication key agreement.
A negative answer would require a new lower bound that applies to imperfect-completeness protocols and to the relevant \qrom{} adversarial model.

\section{Discussion and Open Problems}
\label{sec:conclusion}

The main conceptual observation of this work is that hidden virtual inputs naturally admit two computational notions.
Within the computational Bell-mapping framework used in this work, classical hiding is sufficient for the computational Bell inequality and \poq{} guarantee, whereas quantum hiding against \qpt{} provers with access to their residual state is the additional condition that yields randomness certification against quantum side information.
This distinction makes the framework modular: the Bell inequality and entropy-certification machinery developed in~\cite{cbi} remain unchanged, while upgrading the present \poq{} construction to randomness certification amounts to establishing quantum virtual-input hiding.
Once this property has been established, the general entropy-certification framework applies.

Our results also suggest a more abstract perspective on the cryptographic primitive underlying randomness certification.
For security, the protocol can be understood as requiring a statistically hiding commitment satisfying \pq{} computational binding, rather than a particular structured primitive such as a \tcf{}.
The \owp{}, \tcf{}, and \qrom{} instantiations demonstrate that this abstraction admits realizations with different algorithmic properties.
In this sense, the commitment primitive provides the cryptographic mechanism responsible for hiding the virtual input, while the Bell test and entropy analysis remain independent of its particular instantiation.

Finally, our constructions illustrate a separation between the cryptographic properties required for security and the algorithmic requirements for verifier efficiency.
Statistical hiding and \pq{} computational binding suffice for the security proof.
Efficient verification, however, additionally requires the verifier to efficiently certify the uniqueness condition arising during the hidden-qubit generation procedure.
The \tcf{} instantiation provides this additional structure through its trapdoor, whereas the \owp{} and \qrom{} instantiations require inefficient opening recovery.

\subsection{Open Problems}

Our work raises several natural directions for future research.

\paragraph{Efficient verification in the plain \qrom{}.}
Our \qrom{} instantiation yields randomness certification without computational assumptions beyond the quantum random oracle, but requires an inefficient classical verifier.
Recent works have established efficient randomness certification in the \qrom{} without the Aaronson--Ambainis conjecture.
Khurana \emph{et al.}~\cite{randomness_no_aa} proved security against adversaries making~$o(\log\security)$ adaptive quantum-query layers, while Coladangelo \emph{et al.}~\cite{unconditional_randomness} subsequently established efficient, publicly verifiable randomness certification against computationally unbounded adversaries making a subexponential number of adaptive quantum queries.
These results demonstrate that inefficient verification is not inherent to randomness certification in the \qrom{}.

In our commitment-based construction, however, the verifier must determine whether the selected hash buckets contain unique valid openings.
This exact unique-opening check is the source of the verifier's inefficiency.
A natural open problem is whether it can be replaced by an efficiently verifiable condition while preserving quantum hiding of the virtual input, non-negligible honest continuation probability, and the corresponding entropy-certification guarantee.
Such an improvement would provide an efficient realization of our commitment-based approach and further clarify the relationship between quantum virtual-input hiding, efficient verification, and classical-communication key agreement in the \qrom{}.
As discussed in Subsection~\ref{sec:qrom_verifier_barriers}, existing impossibility results do not rule out such an improvement in the general setting.

\paragraph{Characterizing efficient commitment instantiations.}
Our instantiations suggest that statistical hiding and \pq{} computational binding capture the cryptographic requirements for security, but do not characterize when the resulting protocol can be verified efficiently.
The \tcf{} construction provides additional structure that allows the verifier to efficiently recover the relevant openings and certify their uniqueness, whereas the \qrom{} construction does not.
It is therefore natural to ask which additional properties of a commitment scheme suffice for efficient hidden-qubit generation and verification.
Such a characterization could provide a general criterion for identifying commitment primitives that yield efficient randomness-certification protocols within our framework.

\paragraph{Verifiable remote state preparation from commitments.}
Zhang~\cite{zhang2025formulations} studies verifiable remote state preparation of BB84 states and develops constructions and transformations for broader state families, including coherent two-branch states and equatorial states, with applications to classical verification of quantum computation.
This suggests a natural question: can our commitment-based hidden-qubit preparation be strengthened to satisfy the corresponding verifiable remote-state-preparation functionality?
A positive answer, under compatible security definitions and with an efficient verifier, could extend the implications of our commitment-based framework beyond randomness certification to verifiable remote state preparation and potentially classical verification of quantum computation.

\paragraph{Complexity of inefficiently verifiable qubit tests.}
Our \qrom{} construction also gives a qubit test with an inefficient classical verifier.
Although such a protocol is of limited practical use, it raises a natural complexity-theoretic question concerning the tradeoff between verifier complexity and cryptographic assumptions.
In particular, how does the computational complexity of the classical verifier affect the cryptographic assumptions sufficient for qubit testing?
The commitment-based perspective developed here suggests that relaxing verifier efficiency may substantially weaken the required assumptions.
Understanding this tradeoff, particularly in intermediate regimes of superpolynomial but bounded verifier complexity, remains an interesting direction.

\paragraph{\Pq{} commitments from \owfs{}.}
A particularly appealing direction is to construct the \pq{} commitment primitive required by our framework from arbitrary \pq{} \owfs{}.
The relationship between \pq{} \owfs{} and quantum bit commitments has been studied in~\cite{bc_properties}, which raises related questions concerning the assumptions needed for quantum commitments.

Recent works establish the \pq{} security of the NOVY commitment protocol~\cite{novy} from \pq{} \owp{}.
Chen, Wang, and Xie~\cite{chen2026minicrypt} prove quantum sum-binding security and extend their construction to efficiently verifiable one-way puzzles with nearly uniform puzzle distributions.
Dupuis et al.~\cite{dupuis2026postquantum} establish \pq{} open-once binding and derive sum-binding and collapse-binding using the unique-opening structure of NOVY.
Chevalier, Holin, and Vu~\cite{chevalier2026interactive} establish \pq{} security of multi-output interactive hashing for arbitrary relations and recover perfectly hiding commitments from \pq{} \owp{} as a corollary.
Thus, \pq{} \owp{} suffice to construct statistically hiding commitments satisfying binding notions stronger than those required by our framework.
The arbitrary-relation theorem of~\cite{chevalier2026interactive} may also be relevant to extending classical commitment constructions from arbitrary \owfs{} to the \pq{} setting, although such an extension does not follow directly from their result.

An intriguing possibility is that statistically hiding commitments could be constructed from arbitrary \pq{} \owfs{} by relaxing the required binding guarantee from quantum sum-binding to ordinary \pq{} computational binding.
Indeed, our security proof only requires that no \qpt{} sender can produce valid openings of the same commitment to both bits, without requiring the stronger guarantees provided by sum-binding or collapse-binding.
This raises a natural question: can the construction of statistically hiding commitments from \pq{} \owp{} be extended to arbitrary \pq{} \owfs{} when only computational binding against quantum senders is required?
A positive answer would yield randomness certification from arbitrary \pq{} \owfs{} without resolving the corresponding open problem for stronger notions of quantum binding.

\section*{Acknowledgments.} 
	The authors would like to thank Thomas A. Hahn, Omri Shmueli and Thomas Vidick for helpful discussions, and Itay Shalit for discussions concerning the connection to classical-communication key agreement.

\newpage
\appendix

\section{Classical Hiding of the Virtual Input}
\label{app:classical_hiding}

\begin{proof}[Proof of~\virtualHidingRef{classical}]
	Assume towards a contradiction that there exists a \ppt{} prover~$\adversary$ and a non-negligible function~$\delta$ such that
	\begin{equation}
		\label{eq:classical_hiding_contradiction}
		\Pr\left[
			\adversary(\transStaterv) =
			\virtualInputRV
			\;\wedge\;
			\contEvent
		\right]
		>
			\frac12 \Pr[\contEvent]
			+
			\delta(\security)
		\; .
	\end{equation}
	Since~$\delta$ is non-negligible, there exist a polynomial~$q$ and an infinite set~$\mathcal I\subseteq\N$ such that
	\begin{equation}
		\delta(\security)
		\geq
		\varepsilon_0(\security)
		\coloneqq
		\frac{1}{q(\security)}
		\; ,
	\end{equation}
	for every~$\security\in\mathcal I$.
	We restrict attention to these security parameters.

	We construct a \ppt{} sender~$\guesser$ that violates classical computational binding.
	The sender~$\guesser$ internally executes~$\adversary$ through the commit phase and the hashing step of Phase~A1.
	Let~$\Sigma$ denote the complete classical configuration immediately after~$\hashval$ is produced and immediately before~$\rand_0,\rand_1$ are sampled.
	For every possible value~$\sigma$ of~$\Sigma$, let
	\begin{equation}
		\mu(\sigma)
		\coloneqq
		\Pr[\Sigma=\sigma] \; .
	\end{equation}
	The configuration~$\sigma$ contains the commitment transcript, the hash functions, the hash value, the internal state of~$\adversary$, and its random tape up to that point.
	Because~$\adversary$ is classical,~$\guesser$ can save and restore~$\sigma$ polynomially many times.

	Let~$\continuingPrefixes$ be the set of configurations for which~$\contEvent$ holds.
	Then
	\begin{equation}
		\sum_{\sigma\in\continuingPrefixes}
		\mu(\sigma) =
		\Pr[\contEvent] =
		\pce \; .
	\end{equation}
	For every~$\sigma\in\continuingPrefixes$, there exist unique strings~$\decomm_{0,\sigma}$ and~$\decomm_{1,\sigma}$ such that
	\begin{equation}
		\Decomm_{\bit,\bctrans^{\sigma}}
		\cap
		\hashfunc_{\bit,\sigma}^{-1}(\hashval_\sigma)
		=
		\left\{
			\decomm_{\bit,\sigma}
		\right\}
		\qquad
		\text{for every }\bit\in\{0,1\}
		\; .
	\end{equation}
	Define
	\begin{equation}
		\concDecomm_\sigma
		\coloneqq
		\decomm_{0,\sigma}
		\mathbin\Vert
		\decomm_{1,\sigma}
		\in
		\{0,1\}^{2\ell} \; .
	\end{equation}

	For every query
	\begin{equation}
		\concRand
		\coloneqq
		\rand_0
		\mathbin\Vert
		\rand_1
		\in
		\{0,1\}^{2\ell} \; ,
	\end{equation}
	define a randomized oracle~$\O_\sigma$ as follows.
	The oracle~$\O_\sigma$ rewinds~$\adversary$ to the configuration~$\sigma$, sends~$(\rand_0,\rand_1)$ to~$\adversary$, completes Phase~A2, and returns the predictor's final output.
	On every~$\sigma\in\continuingPrefixes$,
	\begin{align}
		\imap(\trans)
		&=
		(\rand_0\cdot\decomm_{0,\sigma})
		\oplus
		(\rand_1\cdot\decomm_{1,\sigma})
		\nonumber\\
		&=
		\concRand\cdot\concDecomm_\sigma
		\; .
		\label{eq:classical_concatenated_inner_product}
	\end{align}
	Thus,~$\mathcal O_\sigma$ is a noisy oracle for the Hadamard encoding of the entire pair~$(\decomm_{0,\sigma},\decomm_{1,\sigma})$ with noise level determined by the success rate of~$\adversary$.

	For every~$\sigma\in\continuingPrefixes$, define
	\begin{equation}
		p_\sigma
		\coloneqq
		\Pr_{\concRand\leftarrow\{0,1\}^{2\ell}}\left[
			\mathcal O_\sigma(\concRand)
			=
			\concRand\cdot\concDecomm_\sigma
		\right]
		\; ,
	\end{equation}
	and
	\begin{equation}
		\beta_\sigma
		\coloneqq
		p_\sigma-\frac12
		\; .
	\end{equation}
	Equation~\eqref{eq:classical_hiding_contradiction} implies
	\begin{align}
		\label{eq:classical_average_bias}
		\sum_{\sigma\in\continuingPrefixes}
		\mu(\sigma)\beta_\sigma
		& =
		\sum_{\sigma\in\continuingPrefixes}
        \mu(\sigma)\left(p_\sigma-\frac12\right)
		\\
		& =
		\sum_{\sigma\in\continuingPrefixes}
		\mu(\sigma)
		\Pr\left[
			\adversary(\transStaterv) =
			\virtualInputRV
			\mid
			\Sigma=\sigma
		\right]
		-\frac12 \pce 
		\\
		& =
		\pce
		\Pr\left[
			\adversary(\transStaterv) =
			\virtualInputRV
			\mid
			\contEvent
		\right]
		-\frac12 \pce
		\\
		& = 
		\Pr\left[
			\adversary(\transStaterv) =
			\virtualInputRV
			\;\wedge\;
			\contEvent
		\right]
		-\frac12 \pce
		\\
		& >
		\delta(\security)
		\; .
	\end{align}

	Define
	\begin{equation}
		\goodPrefixes
		\coloneqq
		\left\{
			\sigma\in\continuingPrefixes:
			\beta_\sigma
			\geq
			\frac{\varepsilon_0}{2}
		\right\}
		\; .
	\end{equation}
	Since~$\beta_\sigma\leq1/2$, we have
	\begin{align}
		\sum_{\sigma\in\continuingPrefixes}
		\mu(\sigma)\beta_\sigma
		&\leq
		\frac12
		\sum_{\sigma\in\goodPrefixes}
		\mu(\sigma)
		+
		\frac{\varepsilon_0}{2}
		\sum_{\sigma\in\continuingPrefixes\setminus\goodPrefixes}
		\mu(\sigma)
		\nonumber\\
		&\leq
		\frac12
		\sum_{\sigma\in\goodPrefixes}
		\mu(\sigma)
		+
		\frac{\varepsilon_0}{2}
		\; .
	\end{align}
	Combining this with Equation~\eqref{eq:classical_average_bias} gives
	\begin{equation}
		\label{eq:classical_good_prefix_mass}
		\sum_{\sigma\in\goodPrefixes}
		\mu(\sigma)
		>
		\varepsilon_0
		\; .
	\end{equation}

	Fix~$\sigma\in\goodPrefixes$.
	Then
	\begin{equation}
		\Pr_{\concRand\leftarrow\{0,1\}^{2\ell}}
		\left[
			\mathcal O_\sigma(\concRand)
			=
			\concRand\cdot\concDecomm_\sigma
		\right]
		\geq
		\frac12
		+
		\frac{\varepsilon_0}{2}
		\; .
	\end{equation}
	By the Goldreich--Levin theorem~\cite{goldreich1989hardcore}, a probabilistic polynomial-time list decoder with oracle access to~$\mathcal O_\sigma$ outputs a list~$\mathcal L_\sigma\subseteq\{0,1\}^{2\ell}$ of size polynomial in~$\security$ and~$1/\varepsilon_0$ such that
	\begin{equation}
		\Pr\left[
			\concDecomm_\sigma
			\in
			\mathcal L_\sigma
		\right]
		\geq
		\frac23
		\; .
	\end{equation}
	The constant success probability follows by standard amplification.

	For each~$\widehat\concDecomm\in\mathcal L_\sigma$, it parses
	\begin{equation}
		\widehat\concDecomm
		=
		\widehat\decomm_0
		\mathbin\Vert
		\widehat\decomm_1
		\; ,
	\end{equation}
	and checks whether
	\begin{equation}
		\open(\bctrans^{\sigma},0,\widehat\decomm_0)
		=
		\top
		\qquad\text{and}\qquad
		\open(\bctrans^{\sigma},1,\widehat\decomm_1)
		=
		\top
		\; .
	\end{equation}
	Whenever both checks accept,~$\guesser$ outputs the pair~$(\widehat\decomm_0,\widehat\decomm_1)$ in the binding experiment.
	For every~$\sigma\in\goodPrefixes$, the true string~$\concDecomm_\sigma$ passes both checks.
	Therefore,
	\begin{align}
		\Pr\left[
			\bindExp_{\guesser,\bc}(1^\security)
			=
			1
		\right]
		&\geq
		\frac23
		\sum_{\sigma\in\goodPrefixes}
		\mu(\sigma)
		\nonumber\\
		&>
		\frac{2}{3}\varepsilon_0
		\; .
	\end{align}
	This probability is non-negligible on the infinite set~$\mathcal I$, contradicting classical computational binding.

	We conclude that
	\begin{equation}
		\Pr\left[
			\adversary(\transRV)
			=
			\imap(\transRV)
			\;\wedge\;
			\contEvent
		\right]
		\leq
		\frac12\Pr[\contEvent]
		+
		\negl(\security)
		\; .
	\end{equation}
	Whenever~$\Pr[\contEvent]\geq1/\poly(\security)$, division by~$\Pr[\contEvent]$ gives the corresponding conditional hiding statement.
\end{proof}

\section{Analytical Entropy Bound}
\label{app:entropy_proof}

This appendix gives the formal quantitative proof of \entthmref{} for every continuation probability~$\pce(\security)>0$.
When~$\pce(\security)\geq1/\poly(\security)$, \virtualHidingRef{quantum} also gives negligible conditional virtual-input leakage, and the computational Bell-inequality framework of~\cite{cbi} supplies an alternative modular derivation.
The purpose of the proof below is to track the dependence on the continuation probability, obtain closed-form constants, and prove the fixed-challenge bound used in Section~\ref{sec:eat}.

\subsection{Setup and Proof Strategy}

The proof consists of three steps.
First, near-optimal acceptance implies that the two Phase~B observables approximately anti-commute.
Second, Jordan's lemma reduces the analysis to subspaces of dimension at most two, and quantum virtual-input hiding implies that a constant fraction of the state lies on blocks where the observables are nearly complementary and one observable has constant bias.
Third, an entropic uncertainty relation converts these two properties into a positive lower bound on the conditional entropy of the Phase~B output.

Throughout Appendix~\ref{app:entropy_proof}, unless explicitly stated otherwise, all states, probabilities, and expectations are conditioned on the continuation event~$\contEvent$.
Thus, for example,~$\Pr[\transRV=\trans]$ abbreviates~$\Pr[\transRV=\trans\mid\contEvent]$ and~$\pacc$ denotes~$\Pr[\mathsf{accept}\mid\contEvent]$.
The notation~$\rho\mid\contEvent$ is restored explicitly in the final entropy statements.

For every transcript~$\trans$, let~$\postPhaseState{\trans}$ denote the normalized joint post-Phase-A state of~$\proverReg\eveReg$, conditioned on~$\transRV=\trans$.
Here,~$\eveReg$ denotes arbitrary quantum side information held by an adversary and may be arbitrarily correlated with the prover's post-Phase-A state.
Define
\begin{equation}
	\postPhaseStateP{\trans}
	\coloneqq
	\tr_{\eveReg}\!\left[
		\postPhaseState{\trans}
	\right]
	\; .
\end{equation}

By Remark~\ref{rem:binary_naimark}, we may pass to a Naimark dilation and represent the prover's two binary Phase-B measurements by self-adjoint unitary observables~$\bellMeasurement_0$ and~$\bellMeasurement_1$, with eigenvalues~$\pm1$ corresponding to outcomes~$0$ and~$1$, respectively.
We include a classical copy of~$\transRV$ in the prover's register, so transcript-dependent strategies may be represented by fixed observables that are block diagonal in~$\transRV$.
For~$\bellChallenge\in\{0,1\}$, let~$\bellOutcomeRVAt{\bellChallenge}$ denote the classical outcome obtained by measuring~$\bellMeasurement_\bellChallenge$.
The actual Phase~B output is
\begin{equation}
	\bellOutcomeRV
	=
	\bellOutcomeRVAt{\bellChallengeRV}
	\; .
\end{equation}
Moreover,
\begin{equation}
	\tr\!\left[
		\bellMeasurement_\bellChallenge
		\postPhaseStateP{\trans}
	\right]
	=
	\Pr\left[
		\bellOutcomeRVAt{\bellChallenge}=0
		\;\middle|\;
		\transRV=\trans
	\right]
	-
	\Pr\left[
		\bellOutcomeRVAt{\bellChallenge}=1
		\;\middle|\;
		\transRV=\trans
	\right]
	\; .
\end{equation}

Quantum hiding implies that the virtual input is balanced.
Indeed, applying \virtualHidingRef{quantum} to the two constant predictors gives, for all~$\virtualInput\in\{0,1\}$,
\begin{equation}
	\label{eq:rc_virtual_input_balance}
	\left|
		\Pr[\virtualInputRV=\virtualInput] - \frac12
	\right|
	\leq
	\frac{1}{\pce}
	\negl(\security) \; .
\end{equation}

\subsection{Step 1: Approximate Anticommutation}

\begin{lemma}[Approximate anticommutation]
	\label{lemma:rc_approx_anticommutation}
	Suppose that
	\begin{equation}
		\pacc
		\geq
		\chshopt-\varepsilon
		\; ,
		\qquad
		\varepsilon\geq0.
	\end{equation}
	Then,
	\begin{equation}
		\label{eq:rc_approx_anticommutation}
		\anticommSquare
		\coloneqq
		\underset{\trans}{\E}~
		\tr\!\left[
			\left\{
				\bellMeasurement_0,
				\bellMeasurement_1
			\right\}^2
			\postPhaseStateP{\trans}
		\right]
		=
		\O
		\left(
			\varepsilon + \frac{1}{\pce}\negl(\security)
		\right)
		\; .
	\end{equation}
	Here, for each transcript~$\trans$, the quantity inside the expectation measures the squared anticommutator of the two Phase~B prover observables on the corresponding post-Phase-A state~$\postPhaseStateP{\trans}$.
    Thus, near-optimal acceptance implies that the two observables approximately anti-commute on average over the post-Phase-A states.
\end{lemma}

\begin{proof}
    Let~$s$ denote the advantage over a uniform guess in predicting the parity of the two correct answers in the notation of~\cite[Theorem~4.7]{brakerski2023simple}.
    By \virtualHidingRef{quantum},
    \begin{equation}
        0
        \leq
        s
        \leq
        \min\left\{
            \frac12,
            \frac{\nu_{\mathrm{hid}}(\security)}{\pce(\security)}
        \right\}
    \end{equation}
    for a negligible function~$\nu_{\mathrm{hid}}$.

    The proof of~\cite[Theorem~4.7]{brakerski2023simple} gives an~$\O(\varepsilon+\sqrt{s})$ bound as stated there.
    The only square-root loss in~$s$ enters through the estimate following Equation~(18) of that proof.
    For~$0\leq s\leq1/2$, concavity of the square-root function gives
    \begin{equation}
        \sqrt{\frac12+s}
        \leq
        \frac{1}{\sqrt2}
        +
        \frac{s}{\sqrt2}.
    \end{equation}
    Replacing the weaker estimate at that step and then following the remaining strict-concavity and Markov-inequality arguments of~\cite[Theorem~4.7]{brakerski2023simple} yields
    \begin{equation}
        \anticommSquare
        =
        \O(\varepsilon+s)
        =
        \O\!\left(
            \varepsilon
            +
            \frac{\nu_{\mathrm{hid}}(\security)}{\pce(\security)}
        \right),
    \end{equation}
    as claimed.\footnote{We thank Thomas Vidick for correspondence confirming this correction to the stated dependence.}
\end{proof}

\subsection{Step 2: Good Jordan Blocks}
\label{par:step_2}

The analysis follows the strategy used in previous works on entropic guarantees~\cite{pironio2009collective,brakerski2023simple,merkulov2023entropy}.
We apply Jordan's lemma to represent the prover's Phase~B measurements as a collection of qubit strategies and identify a set of \emph{good} transcript--block pairs.
On every good pair, the two observables are nearly complementary and the state has sufficient bias with respect to~$\bellMeasurement_0$.
These are precisely the properties used in \hyperref[par:step_3]{Step~3}.

Let~$\{\blockProjector{\blockIndex}\}_{\blockIndex}$ be the Jordan-block projections associated with~$\bellMeasurement_0$ and~$\bellMeasurement_1$.
Inside every two-dimensional block~$\blockIndex$, let~$\measurementAngle{\blockIndex}\in[0,\pi]$ denote the angle between the Bloch vectors corresponding to the binary observables~$\bellMeasurement_0^{(\blockIndex)}$ and~$\bellMeasurement_1^{(\blockIndex)}$.
Choose a basis such that
\begin{equation}
    \bellMeasurement_0^{(\blockIndex)} = \sigmaz \; ,
    \qquad
    \bellMeasurement_1^{(\blockIndex)} =
    \cos\!\left(\measurementAngle{\blockIndex}\right)\sigmaz +
    \sin\!\left(\measurementAngle{\blockIndex}\right)\sigmax
    \; .
\end{equation}

For every transcript~$\trans$ and block~$\blockIndex$, define the probability of measuring~$\blockIndex$ given a transcript~$\trans$ by
\begin{equation}
	\blockProbability{\trans}{\blockIndex}
	\coloneqq
	\tr\!\left[
		(\blockProjector{\blockIndex}\otimes\id_{\eveReg})
		\postPhaseState{\trans}
	\right]
	\; .
\end{equation}
Whenever~$\blockProbability{\trans}{\blockIndex}>0$, define the normalized post-Phase-A state conditioned on measuring~$\blockIndex$ by
\begin{equation}
	\blockPostPhaseState{\trans}{\blockIndex}
	\coloneqq
	\frac{
		(\blockProjector{\blockIndex}\otimes\id_{\eveReg})
		\postPhaseState{\trans}
		(\blockProjector{\blockIndex}\otimes\id_{\eveReg})
	}{
		\blockProbability{\trans}{\blockIndex}
	}
	\; ,
\end{equation}
and the corresponding state on the prover side
\begin{equation}
	\blockPostPhaseStateP{\trans}{\blockIndex}
	\coloneqq
	\tr_{\eveReg}\!\left[
		\blockPostPhaseState{\trans}{\blockIndex}
	\right]
	\; .
\end{equation}
Let~$\blockIndexRV$ denote the corresponding block-index random variable, so that
\begin{equation}
	\Pr\left[
		\blockIndexRV=\blockIndex
		\;\middle|\;
		\transRV=\trans
	\right]
	=
	\blockProbability{\trans}{\blockIndex}
	\; .
\end{equation}

For~$0<\angleThreshold<1$, define the set of Jordan blocks on which the two Phase~B observables are nearly complementary by
\begin{equation}
	\label{eq:rc_angle_block_set}
	\angleBlockSet{\angleThreshold}
	\coloneqq
	\left\{
		\blockIndex:
		\cos^2\!\left(
			\measurementAngle{\blockIndex}
		\right)
		\leq
		\angleThreshold
	\right\}
	\; .
\end{equation}
For~$0<\projThreshold<1$, define the set of transcript--block pairs on which the state has bias at least~$\projThreshold$ with respect to~$\bellMeasurement_0$ by
\begin{equation}
	\label{eq:rc_bias_pair_set}
	\biasPairSet{\projThreshold}
	\coloneqq
	\left\{
		(\trans,\blockIndex):
		\left|
			\tr\!\left[
				\bellMeasurement_0^{(\blockIndex)}
				\blockPostPhaseStateP{\trans}{\blockIndex}
			\right]
		\right|
		\geq
		\projThreshold
	\right\}
	\; .
\end{equation}

Set
\begin{equation}
	\label{eq:rc_concrete_thresholds}
	\projThreshold_*
	\coloneqq
	\frac{9(\sqrt2-1)}{10}
	\; ,
	\qquad
	\angleThreshold_*
	\coloneqq
	10^{-3}
	\; .
\end{equation}

\subsubsection{Weight on Nearly Complementary Blocks}

\begin{lemma}[Weight on nearly complementary blocks]
	\label{lemma:app_nearly_complementary_weight}
	For every~$0<\angleThreshold<1$,
	\begin{equation}
		\Pr\left[
			\blockIndexRV
			\in
			\angleBlockSet{\angleThreshold}
		\right]
		\geq
		1
		-
		\frac{\anticommSquare}{4\angleThreshold}
		\; .
	\end{equation}
\end{lemma}

\begin{proof}
	Define the nonnegative random variable
	\begin{equation}
		Z
		\coloneqq
		\cos^2\!\left(
			\measurementAngle{\blockIndexRV}
		\right)
		\; .
	\end{equation}
	By Lemma~\ref{lemma:anticommutator_overlap} and the definition of~$\blockIndexRV$, for every transcript~$\trans$,
	\begin{equation}
		\tr\!\left[
			\left\{
				\bellMeasurement_0,
				\bellMeasurement_1
			\right\}^2
			\postPhaseStateP{\trans}
		\right]
		=
		4
		\sum_{\blockIndex}
		\blockProbability{\trans}{\blockIndex}
		\cos^2\!\left(
			\measurementAngle{\blockIndex}
		\right)
		\; .
	\end{equation}
	Averaging over~$\trans$ therefore gives
	\begin{align}
		\anticommSquare
		&=
		4
		\underset{\trans}{\E}
		\sum_{\blockIndex}
		\blockProbability{\trans}{\blockIndex}
		\cos^2\!\left(
			\measurementAngle{\blockIndex}
		\right)
		\nonumber\\
		&=
		4\E[Z]
		\; .
	\end{align}
	By the definition of~$\angleBlockSet{\angleThreshold}$,
	\begin{align}
		\Pr\left[
			\blockIndexRV
			\notin
			\angleBlockSet{\angleThreshold}
		\right]
		&=
		\Pr\left[
			Z>\angleThreshold
		\right]
		\nonumber\\
		&\leq
		\Pr\left[
			Z\geq\angleThreshold
		\right]
		\; .
	\end{align}
	Applying Markov's inequality to~$Z$ yields
	\begin{equation}
		\Pr\left[
			Z\geq\angleThreshold
		\right]
		\leq
		\frac{\E[Z]}{\angleThreshold}
		=
		\frac{\anticommSquare}{4\angleThreshold}
		\; .
	\end{equation}
	Hence,
	\begin{equation}
		\Pr\left[
			\blockIndexRV
			\in
			\angleBlockSet{\angleThreshold}
		\right]
		\geq
		1
		-
		\frac{\anticommSquare}{4\angleThreshold}
		\; .
	\end{equation}
\end{proof}

\subsubsection{Measurement Bias from the CHSH Correlation}

We next extract a second consequence of near-optimal CHSH performance that will be needed for the entropy analysis.
In addition to forcing the two Phase~B observables to be nearly complementary, a near-optimal CHSH correlation implies that, conditioned on one value of the virtual input, the outcome of one of these observables cannot be close to uniform.
We will use this measurement bias, together with the near-complementarity established above, to identify the good Jordan blocks required for the entropic uncertainty argument.

Recall that, given the functions~$\imap,\omap:\Trans\rightarrow\{0,1\}$ that define the Bell mapping, the Phase~A transcript~$\trans$ determines the virtual input and output as~$\virtualInput=\imap(\trans)$ and~$\virtualOutput=\omap(\trans)$.
For every~$(\virtualOutput,\virtualInput)\in\{0,1\}^2$ with~$\Pr[\virtualOutputRV=\virtualOutput,\virtualInputRV=\virtualInput]>0$, define the normalized post-Phase-A states
\begin{equation}
	\rho_{\proverReg}^{\virtualOutput,\virtualInput}
	\coloneqq
	\frac{1}{
		\Pr[
			\virtualOutputRV=\virtualOutput,
			\virtualInputRV=\virtualInput
		]
	}
	\sum_{\substack{\trans:\\
		\omap(\trans)=\virtualOutput,\\
		\imap(\trans)=\virtualInput}}
	\Pr[\transRV=\trans]
	\postPhaseStateP{\trans}
	\; .
\end{equation}

\begin{definition}[CHSH contribution terms]
	\label{def:app_chsh_contributions}
	For~$\virtualInput,\bellChallenge\in\{0,1\}$, define the contribution to the CHSH correlation corresponding to the virtual input~$\virtualInput$ and challenge~$\bellChallenge$ by
	\begin{equation}
		S_{\virtualInput,\bellChallenge}
		\coloneqq
		\sum_{\virtualOutput\in\{0,1\}}
		\Pr[\bellChallengeRV=\bellChallenge]
		\Pr[
			\virtualOutputRV=\virtualOutput,
			\virtualInputRV=\virtualInput
		]
		\tr\!\left[
			(-1)^{
				\virtualOutput
				\oplus
				(\virtualInput\cdot\bellChallenge)
			}
			\bellMeasurement_\bellChallenge
			\rho_{\proverReg}^{\virtualOutput,\virtualInput}
		\right]
		\; .
	\end{equation}
	For each~$\virtualInput\in\{0,1\}$, define its total contribution by
	\begin{equation}
		S_{\virtualInput}
		\coloneqq
		\sum_{\bellChallenge\in\{0,1\}}
		S_{\virtualInput,\bellChallenge}
		\; .
	\end{equation}
	Since~$\bellChallengeRV$ is sampled independently of the Phase~A transcript, with this notation the CHSH correlation satisfies
	\begin{equation}
		\pacc
		=
		\frac{1}{2}
		+
		\frac
		{S_0+S_1}
		{2}
		\; .
	\end{equation}
\end{definition}

\begin{lemma}[Upper bound on the~$\virtualInput=1$ contribution]
	\label{lemma:app_xone_contribution}
	It holds that
	\begin{equation}
		S_1
		\leq
		\sqrt{
			\frac{
				\sqrt2
				+
				\sqrt{\anticommSquare}
			}{
				8\sqrt2
			}
		}
		+
		\frac{1}{\pce}
		\negl(\security)
		\; ,
	\end{equation}
	where~$\pce(\security)\coloneqq\Pr[\contEvent]$ denotes the probability that Phase~A1 does not abort.
\end{lemma}
\begin{proof}
	If~$\Pr[\virtualInputRV=1]=0$, then~$S_1=0$, and the claimed bound is immediate.
	Hence, we may assume that~$\Pr[\virtualInputRV=1]>0$.
	Terms corresponding to zero-probability pairs~$(\virtualOutput,\virtualInput)$ are understood to be omitted.

	For every~$(\virtualOutput,\virtualInput)$ and block~$\blockIndex$, define
	\begin{equation}
		p_{\virtualOutput,\virtualInput}^{(\blockIndex)}
		\coloneqq
		\tr\!\left[
			\blockProjector{\blockIndex}
			\rho_{\proverReg}^{\virtualOutput,\virtualInput}
		\right]
		\; .
	\end{equation}
	For~$\virtualInput\in\{0,1\}$ with~$\Pr[\virtualInputRV=\virtualInput]>0$, define
	\begin{equation}
		p_{\virtualInput}^{(\blockIndex)}
		\coloneqq
		\sum_{\virtualOutput\in\{0,1\}}
		\Pr\left[
			\virtualOutputRV=\virtualOutput
			\;\middle|\;
			\virtualInputRV=\virtualInput
		\right]
		p_{\virtualOutput,\virtualInput}^{(\blockIndex)}
		\; .
	\end{equation}
	For every such~$\virtualInput$, the coefficients
	$\{p_{\virtualInput}^{(\blockIndex)}\}_{\blockIndex}$
	form a probability distribution.

	Since~$\bellChallengeRV$ is uniform,
	\begin{align}
		S_1
		&=
		\frac12
		\sum_{\virtualOutput\in\{0,1\}}
		\Pr[
			\virtualOutputRV=\virtualOutput,
			\virtualInputRV=1
		]
		(-1)^\virtualOutput
		\tr\!\left[
			(\bellMeasurement_0-\bellMeasurement_1)
			\rho_{\proverReg}^{\virtualOutput,1}
		\right]
		\nonumber\\
		&\leq
		\frac12
		\sum_{\virtualOutput\in\{0,1\}}
		\Pr[
			\virtualOutputRV=\virtualOutput,
			\virtualInputRV=1
		]
		\sum_{\blockIndex}
		p_{\virtualOutput,1}^{(\blockIndex)}
		\left\|
			\bellMeasurement_0^{(\blockIndex)}
			-
			\bellMeasurement_1^{(\blockIndex)}
		\right\|_\infty
		\; .
	\end{align}
	Inside block~$\blockIndex$,
	\begin{align}
		\frac12
		\left\|
			\bellMeasurement_0^{(\blockIndex)}
			-
			\bellMeasurement_1^{(\blockIndex)}
		\right\|_\infty
		&=
		\left|
			\sin\!\left(
				\frac{\measurementAngle{\blockIndex}}{2}
			\right)
		\right|
		\nonumber\\
		&\leq
		\sqrt{
			\frac{
				1
				+
				\left|
					\cos\!\left(
						\measurementAngle{\blockIndex}
					\right)
				\right|
			}{2}
		}
		\; .
	\end{align}
	Consequently,
	\begin{equation}
		S_1
		\leq
		\Pr[\virtualInputRV=1]
		\sum_{\blockIndex}
		p_1^{(\blockIndex)}
		\sqrt{
			\frac{
				1
				+
				\left|
					\cos\!\left(
						\measurementAngle{\blockIndex}
					\right)
				\right|
			}{2}
		}
		\; .
	\end{equation}
	By Jensen's inequality,
	\begin{equation}
		S_1
		\leq
		\Pr[\virtualInputRV=1]
		\sqrt{
			\frac{
				1
				+
				\displaystyle
				\sum_{\blockIndex}
				p_1^{(\blockIndex)}
				\left|
					\cos\!\left(
						\measurementAngle{\blockIndex}
					\right)
				\right|
			}{2}
		}
		\; .
	\end{equation}
	By Cauchy--Schwarz,
	\begin{equation}
		\sum_{\blockIndex}
		p_1^{(\blockIndex)}
		\left|
			\cos\!\left(
				\measurementAngle{\blockIndex}
			\right)
		\right|
		\leq
		\sqrt{
			\sum_{\blockIndex}
			p_1^{(\blockIndex)}
			\cos^2\!\left(
				\measurementAngle{\blockIndex}
			\right)
		}
		\; .
	\end{equation}
	For brevity, set
	\begin{equation}
		q
		\coloneqq
		\Pr[\virtualInputRV=1]
		\qquad\text{and}\qquad
		T
		\coloneqq
		\sum_{\blockIndex}
		p_1^{(\blockIndex)}
		\cos^2\!\left(
			\measurementAngle{\blockIndex}
		\right)
		\; .
	\end{equation}
	The preceding estimates give
	\begin{equation}
		S_1
		\leq
		q
		\sqrt{
			\frac{
				1+\sqrt{T}
			}{2}
		}
		\; .
	\end{equation}

	Moreover,
	\begin{equation}
		\anticommSquare
		=
		4
		\sum_{\virtualInput\in\{0,1\}}
		\Pr[\virtualInputRV=\virtualInput]
		\sum_{\blockIndex}
		p_{\virtualInput}^{(\blockIndex)}
		\cos^2\!\left(
			\measurementAngle{\blockIndex}
		\right)
		\; .
	\end{equation}
	Since every term is nonnegative,
	\begin{equation}
		4qT
		\leq
		\anticommSquare
		\; ,
	\end{equation}
	and hence
	\begin{equation}
		T
		\leq
		\frac{\anticommSquare}{4q}
		\; .
	\end{equation}

	If~$S_1\leq0$, the claimed upper bound is immediate.
	Otherwise, squaring the preceding bound and using the estimate on~$T$ gives
	\begin{align}
		S_1^2
		&\leq
		\frac{q^2}{2}
		\left(
			1+\sqrt{T}
		\right)
		\nonumber\\
		&\leq
		\frac{q^2}{2}
		+
		\frac{q^2}{2}
		\sqrt{
			\frac{\anticommSquare}{4q}
		}
		\nonumber\\
		&=
		\frac{q^2}{2}
		+
		\frac{
			q^{3/2}
			\sqrt{\anticommSquare}
		}{4}
		\; .
	\end{align}

	By Equation~\eqref{eq:rc_virtual_input_balance}, there exists a negligible function~$\nu$ such that
	\begin{equation}
		\left|
			q-\frac12
		\right|
		\leq
		\frac{\nu(\security)}{\pce}
		\; .
	\end{equation}
	Set
	\begin{equation}
		\eta(\security)
		\coloneqq
		\min\left\{
			\frac12,
			\frac{\nu(\security)}{\pce}
		\right\}
		\; .
	\end{equation}
	Since~$\left|q-\frac12\right|\leq\frac12$, it follows that
	\begin{equation}
		\left|
			q-\frac12
		\right|
		\leq
		\eta(\security)
		\leq
		\frac12
		\qquad\text{and}\qquad
		\eta(\security)
		\leq
		\frac{\nu(\security)}{\pce}
		\; .
	\end{equation}
	In particular,
	\begin{equation}
		q
		\leq
		\frac12+\eta(\security)
		\; .
	\end{equation}
	Using~$\eta(\security)^2\leq\eta(\security)/2$, we obtain
	\begin{align}
		q^2
		&\leq
		\left(
			\frac12+\eta(\security)
		\right)^2
		\nonumber\\
		&=
		\frac14
		+
		\eta(\security)
		+
		\eta(\security)^2
		\nonumber\\
		&\leq
		\frac14
		+
		\frac32
		\eta(\security)
		\; .
	\end{align}
	The function
	\begin{equation}
		f(t)
		\coloneqq
		\left(
			\frac12+t
		\right)^{3/2}
	\end{equation}
	is convex on~$[0,1/2]$.
	Therefore, it lies below the secant line joining its values at~$0$ and~$1/2$, and hence
	\begin{align}
		q^{3/2}
		& \leq
		\left(
			2-1/\sqrt2
		\right)
		\eta(\security)
		+
		2^{-3/2}
		\nonumber \\
		& \leq
		\frac32
		\eta(\security)
		+
		2^{-3/2}
	\end{align}
	The decomposition of~$\anticommSquare$ also implies that
	\begin{equation}
		0
		\leq
		\anticommSquare
		\leq
		4
		\; .
	\end{equation}
	Combining the preceding estimates gives
	\begin{align}
		S_1^2
		&\leq
		\frac18
		+
		\frac34
		\eta(\security)
		+
		\frac{
			\sqrt{\anticommSquare}
		}{
			8\sqrt2
		}
		+
		\frac{
			3\sqrt{\anticommSquare}
		}{
			8
		}
		\eta(\security)
		\nonumber\\
		&\leq
		\frac{
			\sqrt2
			+
			\sqrt{\anticommSquare}
		}{
			8\sqrt2
		}
		+
		\frac32
		\eta(\security)
		\; .
	\end{align}
	Set
	\begin{equation}
		b
		\coloneqq
		\frac{
			\sqrt2
			+
			\sqrt{\anticommSquare}
		}{
			8\sqrt2
		}
		\; .
	\end{equation}
	Since~$b\geq1/8$, the inequality from Lemma~\ref{lemma:sqrt_linearization}
	for~$r\geq0$ yields
	\begin{align}
		S_1
		&\leq
		\sqrt{
			\frac{
				\sqrt2
				+
				\sqrt{\anticommSquare}
			}{
				8\sqrt2
			}
		}
		+
		\frac{3\sqrt2}{2}
		\eta(\security)
		\nonumber\\
		&\leq
		\sqrt{
			\frac{
				\sqrt2
				+
				\sqrt{\anticommSquare}
			}{
				8\sqrt2
			}
		}
		+
		\frac{3\sqrt2}{2}
		\frac{\nu(\security)}{\pce}
		\; .
	\end{align}
	Since a constant multiple of a negligible function is negligible, we conclude that
	\begin{equation}
		S_1
		\leq
		\sqrt{
			\frac{
				\sqrt2
				+
				\sqrt{\anticommSquare}
			}{
				8\sqrt2
			}
		}
		+
		\frac{1}{\pce}
		\negl(\security)
		\; .
	\end{equation}
\end{proof}

\begin{lemma}[Upper bound on the~$(\virtualInput,\bellChallenge)=(0,1)$ contribution]
	\label{lemma:app_zero_one_contribution}
	It holds that
	\begin{equation}
		S_{0,1}
		\leq
		\frac14
		+
		\frac{1}{\pce}
		\negl(\security)
		\; ,
	\end{equation}
	where~$\pce(\security)\coloneqq\Pr[\contEvent]$ denotes the probability that Phase~A1 does not abort.
\end{lemma}

\begin{proof}
	By Equation~\eqref{eq:rc_virtual_input_balance}, there exists a negligible function~$\nu$ such that
	\begin{equation}
		\Pr[\virtualInputRV=0]
		\leq
		\frac12
		+
		\frac{\nu(\security)}{\pce}
		\; .
	\end{equation}
	We therefore have
	\begin{align}
		S_{0,1}
		&=
		\frac12
		\sum_{\virtualOutput\in\{0,1\}}
		\Pr[
			\virtualOutputRV=\virtualOutput,
			\virtualInputRV=0
		]
		\tr\!\left[
			(-1)^\virtualOutput
			\bellMeasurement_1
			\rho_{\proverReg}^{\virtualOutput,0}
		\right]
		\nonumber\\
		&\leq
		\frac12
		\Pr[\virtualInputRV=0]
		\nonumber\\
		&\leq
		\frac14
		+
		\frac{\nu(\security)}{2\pce}
		\nonumber\\
		&=
		\frac14
		+
		\frac{1}{\pce}
		\negl(\security)
		\; ,
	\end{align}
	where the first inequality uses
	\begin{equation}
		\left|
			\tr\!\left[
				\bellMeasurement_1
				\rho_{\proverReg}^{\virtualOutput,0}
			\right]
		\right|
		\leq
		1
		\; ,
	\end{equation}
	and the final equality absorbs the constant factor~$1/2$ into the negligible function.
\end{proof}

\begin{corollary}[Conditional lower bound on the~$\bellMeasurement_0$ bias]
	\label{cor:app_bzero_bias}
	Suppose that
    \begin{equation}
        \pacc
        \geq
        \cos^2\!\left({\pi}/{8}\right)
        -
        \varepsilon
        \; ,
        \qquad
        \varepsilon\geq0.
    \end{equation}
	Then, whenever~$\Pr[\virtualInputRV=0]>0$,
	\begin{equation}
		\underset{\transRV\mid\virtualInputRV=0}{\E}
		\left|
			\tr\!\left[
				\bellMeasurement_0
				\postPhaseStateP{\transRV}
			\right]
		\right|
		\geq
		2\sqrt2
		-
		1
		-
		\sqrt{
			2
			+
			\sqrt{2\anticommSquare}
		}
		-
		8\varepsilon
		-
		\frac{1}{\pce}
		\negl(\security)
		\; ,
	\end{equation}
	where~$\pce(\security)\coloneqq\Pr[\contEvent]$ denotes the continuation probability.
\end{corollary}

\begin{proof}
	Set
	\begin{equation}
		B_0
		\coloneqq
		\underset{\transRV\mid\virtualInputRV=0}{\E}
		\left|
			\tr\!\left[
				\bellMeasurement_0
				\postPhaseStateP{\transRV}
			\right]
		\right|
		\; .
	\end{equation}
	By Definition~\ref{def:app_chsh_contributions},
	\begin{align}
		S_{0,0}
		&=
		\frac12
		\sum_{
			\trans:\imap(\trans)=0
			}
		\Pr[\transRV=\trans]
		\tr\!\left[
			(-1)^{\omap(\trans)}
			\bellMeasurement_0
			\postPhaseStateP{\trans}
		\right]
		\nonumber\\
		&\leq
		\frac12
		\Pr[\virtualInputRV=0]
		B_0
		\; .
	\end{align}
	By Equation~\eqref{eq:rc_virtual_input_balance}, there exists a negligible function~$\nu_0$ such that
	\begin{equation}
		\Pr[\virtualInputRV=0]
		\leq
		\frac12
		+
		\frac{\nu_0(\security)}{\pce}
		\; .
	\end{equation}
	Since~$0\leq B_0\leq1$, it follows that
	\begin{align}
		S_{0,0}
		&\leq
		\frac12
		\left(
			\frac12
			+
			\frac{\nu_0(\security)}{\pce}
		\right)
		B_0
		\nonumber\\
		&\leq
		\frac14
		B_0
		+
		\frac{\nu_0(\security)}{2\pce}
		\; .
	\end{align}
	Consequently,
	\begin{equation}
		B_0
		\geq
		4S_{0,0}
		-
		\frac{2\nu_0(\security)}{\pce}
		\; .
	\end{equation}

	On the other hand,
	\begin{equation}
		S_{0,0}
		=
		2\pacc-1
		-
		S_1
		-
		S_{0,1}
		\; .
	\end{equation}
	Substituting
	\begin{equation}
		2\pacc-1
		\geq
		\frac{1}{\sqrt2}
		-
		2\varepsilon
		\; ,
	\end{equation}
	and applying Lemmas~\ref{lemma:app_xone_contribution} and~\ref{lemma:app_zero_one_contribution}, we obtain
	\begin{align}
		S_{0,0}
		&\geq
		\frac{1}{\sqrt2}
		-
		2\varepsilon
		-
		\sqrt{
			\frac{
				\sqrt2
				+
				\sqrt{\anticommSquare}
			}{
				8\sqrt2
			}
		}
		-
		\frac14
		-
		\frac{1}{\pce}
		\negl(\security)
		\; .
	\end{align}
	Combining the preceding estimates and absorbing the sum of the negligible functions into a single negligible function gives
	\begin{align}
		B_0
		&\geq
		4\left(
			\frac{1}{\sqrt2}
			-
			2\varepsilon
			-
			\sqrt{
				\frac{
					\sqrt2
					+
					\sqrt{\anticommSquare}
				}{
					8\sqrt2
				}
			}
			-
			\frac14
		\right)
		-
		\frac{1}{\pce}
		\negl(\security)
		\nonumber\\
		&=
		2\sqrt2
		-
		1
		-
		\sqrt{
			2
			+
			\sqrt{2\anticommSquare}
		}
		-
		8\varepsilon
		-
		\frac{1}{\pce}
		\negl(\security)
		\; .
	\end{align}
\end{proof}

\begin{corollary}[Asymptotic conditional lower bound on the~$\bellMeasurement_0$ bias]
	\label{cor:app_bzero_bias_asymptotic}
	Under the hypotheses of Corollary~\ref{cor:app_bzero_bias},
	\begin{equation}
		\underset{\transRV\mid\virtualInputRV=0}{\E}
		\left|
			\tr\!\left[
				\bellMeasurement_0
				\postPhaseStateP{\transRV}
			\right]
		\right|
		\geq
		\sqrt2
		-
		1
		-
		\O(\sqrt{\varepsilon})
		-
		\frac{1}{\pce}
		\negl(\security)
		\; .
	\end{equation}
\end{corollary}

\begin{proof}
	By Lemma~\ref{lemma:rc_approx_anticommutation}, there exist a constant~$C>0$ and a negligible function~$\nu$ such that
	\begin{equation}
		\anticommSquare
		\leq
		C\varepsilon
		+
		\frac{\nu(\security)}{\pce}
		\; .
	\end{equation}
	Moreover,
	\begin{align}
		\sqrt{
			2
			+
			\sqrt{2\anticommSquare}
		}
		-
		\sqrt2
		&=
		\frac{
			\sqrt{2\anticommSquare}
		}{
			\sqrt{
				2
				+
				\sqrt{2\anticommSquare}
			}
			+
			\sqrt2
		}
		\nonumber\\
		&\leq
		\sqrt{\anticommSquare}
		\nonumber\\
		&\leq
		\sqrt{C\varepsilon}
		+
		\sqrt{
			\frac{\nu(\security)}{\pce}
		}
		\nonumber\\
		&\leq
		\O(\sqrt{\varepsilon})
		+
		\frac{1}{\pce}
		\negl(\security)
		\; .
	\end{align}
	Here, the final inequality uses~$\pce\leq1$ and the fact that the square root of a nonnegative negligible function is negligible.
	Therefore,
	\begin{equation}
		\sqrt{
			2
			+
			\sqrt{2\anticommSquare}
		}
		\leq
		\sqrt2
		+
		\O(\sqrt{\varepsilon})
		+
		\frac{1}{\pce}
		\negl(\security)
		\; .
	\end{equation}
	Substituting this estimate into Corollary~\ref{cor:app_bzero_bias} and absorbing~$\O(\varepsilon)$ into~$\O(\sqrt{\varepsilon})$ proves the result.
\end{proof}

\subsubsection{Good Transcript--Block Pairs}

\begin{lemma}[Good transcript--block pairs]
	\label{lemma:rc_good_pairs}
	Suppose that
    \begin{equation}
        \pacc
        \geq
        \chshopt-\varepsilon
        \; ,
        \qquad
        \varepsilon\geq0.
    \end{equation}
	Then, whenever~$\Pr[\virtualInputRV=0]>0$,
	\begin{equation}
		\label{eq:rc_good_pair_probability}
		\Pr\left[
			\blockIndexRV
			\in
			\angleBlockSet{\angleThreshold_*}
			\;\wedge\;
			(\transRV,\blockIndexRV)
			\in
			\biasPairSet{\projThreshold_*}
			\;\middle|\;
			\virtualInputRV=0
		\right]
		\geq
		\frac{1}{1+10\sqrt2}
		-
		\O(\sqrt{\varepsilon})
		-
		\frac{1}{\pce}
		\negl(\security)
		\; ,
	\end{equation}
	where~$\pce(\security)\coloneqq\Pr[\contEvent]$ denotes the continuation probability.
\end{lemma}

\begin{proof}[Proof of Lemma~\ref{lemma:rc_good_pairs}]
	Fix arbitrary parameters~$0<\angleThreshold\leq1$ and~$0\leq\projThreshold<1$.

	Define the event
	\begin{equation}
		\mathcal{E}_{\mathrm{angle}}
		\coloneqq
		\left\{
			\blockIndexRV
			\notin
			\angleBlockSet{\angleThreshold}
		\right\}
		\; .
	\end{equation}
	By Lemma~\ref{lemma:app_nearly_complementary_weight},
	\begin{equation}
		\Pr[
			\mathcal{E}_{\mathrm{angle}}
		]
		\leq
		\frac{\anticommSquare}{4\angleThreshold}
		\; .
	\end{equation}
	Set
	\begin{equation}
		q_0
		\coloneqq
		\Pr[\virtualInputRV=0]
		\qquad\text{and}\qquad
		r
		\coloneqq
		\Pr[
			\mathcal{E}_{\mathrm{angle}}
			\wedge
			\virtualInputRV=0
		]
		\; .
	\end{equation}
	Whenever the conditional probabilities below are defined, we have~$q_0>0$ and~$0\leq r\leq q_0$.
	By Equation~\eqref{eq:rc_virtual_input_balance}, there exists a negligible function~$\nu$ such that
	\begin{equation}
		\left|
			q_0-\frac12
		\right|
		\leq
		\frac{\nu(\security)}{\pce}
		\; .
	\end{equation}
	Therefore,
	\begin{align}
		\Pr\left[
			\blockIndexRV
			\notin
			\angleBlockSet{\angleThreshold}
			\;\middle|\;
			\virtualInputRV=0
		\right]
		&=
		\frac{r}{q_0}
		\nonumber\\
		&=
		2r
		+
		(1-2q_0)
		\frac{r}{q_0}
		\nonumber\\
		&\leq
		2r
		+
		\left|
			1-2q_0
		\right|
		\nonumber\\
		&\leq
		2
		\Pr[
			\mathcal{E}_{\mathrm{angle}}
		]
		+
		\frac{2\nu(\security)}{\pce}
		\nonumber\\
		&\leq
		\frac{\anticommSquare}{2\angleThreshold}
		+
		\frac{1}{\pce}
		\negl(\security)
		\; .
		\label{eq:app_conditional_bad_angle}
	\end{align}
	Consequently,
	\begin{equation}
		\label{eq:app_conditional_good_angle}
		\Pr\left[
			\blockIndexRV
			\in
			\angleBlockSet{\angleThreshold}
			\;\middle|\;
			\virtualInputRV=0
		\right]
		\geq
		1
		-
		\frac{\anticommSquare}{2\angleThreshold}
		-
		\frac{1}{\pce}
		\negl(\security)
		\; .
	\end{equation}

	Define
	\begin{equation}
		U_{\trans,\blockIndex}
		\coloneqq
		\left|
			\tr\!\left[
				\bellMeasurement_0^{(\blockIndex)}
				\blockPostPhaseStateP{\trans}{\blockIndex}
			\right]
		\right|
		\in
		[0,1]
		\; .
	\end{equation}
	For every transcript~$\trans$,
	\begin{align}
		\left|
			\tr\!\left[
				\bellMeasurement_0
				\postPhaseStateP{\trans}
			\right]
		\right|
		&=
		\left|
			\sum_{\blockIndex}
			\blockProbability{\trans}{\blockIndex}
			\tr\!\left[
				\bellMeasurement_0^{(\blockIndex)}
				\blockPostPhaseStateP{\trans}{\blockIndex}
			\right]
		\right|
		\nonumber\\
		&\leq
		\sum_{\blockIndex}
		\blockProbability{\trans}{\blockIndex}
		U_{\trans,\blockIndex}
		\; .
	\end{align}
	Corollary~\ref{cor:app_bzero_bias_asymptotic} therefore implies
	\begin{equation}
		\label{eq:app_average_block_bias}
		\underset{
			\transRV,\blockIndexRV
			\mid
			\virtualInputRV=0
		}{\E}
		U_{\transRV,\blockIndexRV}
		\geq
		\sqrt2
		-
		1
		-
		\O(\sqrt{\varepsilon})
		-
		\frac{1}{\pce}
		\negl(\security)
		\; .
	\end{equation}

	By the definition of~$\biasPairSet{\projThreshold}$,
	\begin{align}
		\underset{
			\transRV,\blockIndexRV
			\mid
			\virtualInputRV=0
		}{\E}
		U_{\transRV,\blockIndexRV}
		&\leq
		\Pr\left[
			(\transRV,\blockIndexRV)
			\in
			\biasPairSet{\projThreshold}
			\;\middle|\;
			\virtualInputRV=0
		\right]
		\nonumber\\
		&\qquad+
		\projThreshold
		\Pr\left[
			(\transRV,\blockIndexRV)
			\notin
			\biasPairSet{\projThreshold}
			\;\middle|\;
			\virtualInputRV=0
		\right]
		\nonumber\\
		&=
		\projThreshold
		+
		(1-\projThreshold)
		\Pr\left[
			(\transRV,\blockIndexRV)
			\in
			\biasPairSet{\projThreshold}
			\;\middle|\;
			\virtualInputRV=0
		\right]
		\; .
	\end{align}
	Combining this with Equation~\eqref{eq:app_average_block_bias} gives
	\begin{align}
		&\Pr\left[
			(\transRV,\blockIndexRV)
			\in
			\biasPairSet{\projThreshold}
			\;\middle|\;
			\virtualInputRV=0
		\right]
		\nonumber\\
		&\qquad\geq
		\frac{
			\sqrt2-1-\projThreshold
		}{
			1-\projThreshold
		}
		-
		\frac{
			\O(\sqrt{\varepsilon})
		}{
			1-\projThreshold
		}
		-
		\frac{1}{
			(1-\projThreshold)\pce
		}
		\negl(\security)
		\; .
		\label{eq:app_large_bias_probability}
	\end{align}

	Using the union bound together with Equations~\eqref{eq:app_conditional_good_angle} and~\eqref{eq:app_large_bias_probability}, we obtain
	\begin{align}
		&\Pr\left[
			\blockIndexRV
			\in
			\angleBlockSet{\angleThreshold}
			\;\wedge\;
			(\transRV,\blockIndexRV)
			\in
			\biasPairSet{\projThreshold}
			\;\middle|\;
			\virtualInputRV=0
		\right]
		\nonumber\\
		&\qquad\geq
		\frac{
			\sqrt2-1-\projThreshold
		}{
			1-\projThreshold
		}
		-
		\frac{
			\O(\sqrt{\varepsilon})
		}{
			1-\projThreshold
		}
		-
		\frac{\anticommSquare}{2\angleThreshold}
		-
		\frac{1}{
			(1-\projThreshold)\pce
		}
		\negl(\security)
		\; .
		\label{eq:app_general_good_pair_probability}
	\end{align}
	Here, the two continuation-dependent negligible terms have been absorbed into a single negligible function, using~$1/(1-\projThreshold)\geq1$.

	Set~$\projThreshold=\projThreshold_*$ and~$\angleThreshold=\angleThreshold_*$ as in Equation~\eqref{eq:rc_concrete_thresholds}.
	Then,
	\begin{equation}
		\frac{
			\sqrt2-1-\projThreshold_*
		}{
			1-\projThreshold_*
		}
		=
		\frac{1}{1+10\sqrt2}
		\; .
	\end{equation}
	Moreover, Lemma~\ref{lemma:rc_approx_anticommutation} gives
	\begin{equation}
		\frac{\anticommSquare}{2\angleThreshold_*}
		\leq
		\O(\varepsilon)
		+
		\frac{1}{\pce}
		\negl(\security)
		\; .
	\end{equation}
	Since~$\projThreshold_*$ and~$\angleThreshold_*$ are fixed constants, substituting these estimates into Equation~\eqref{eq:app_general_good_pair_probability}, absorbing~$\O(\varepsilon)$ into~$\O(\sqrt{\varepsilon})$, and absorbing constant factors into the negligible function gives
	\begin{equation}
		\Pr\left[
			\blockIndexRV
			\in
			\angleBlockSet{\angleThreshold_*}
			\;\wedge\;
			(\transRV,\blockIndexRV)
			\in
			\biasPairSet{\projThreshold_*}
			\;\middle|\;
			\virtualInputRV=0
		\right]
		\geq
		\frac{1}{1+10\sqrt2}
		-
		\O(\sqrt{\varepsilon})
		-
		\frac{1}{\pce}
		\negl(\security)
		\; ,
	\end{equation}
	which proves Lemma~\ref{lemma:rc_good_pairs}.
\end{proof}

The lemma has a simple geometric interpretation in terms of the qubit strategies arising from the Jordan decomposition.
The condition~$\blockIndex\in\angleBlockSet{\angleThreshold_*}$ selects those qubits for which the two measurement axes are nearly orthogonal, while the condition~$(\trans,\blockIndex)\in\biasPairSet{\projThreshold_*}$ selects those for which the state has a constant projection onto the~$\bellMeasurement_0^{(\blockIndex)}$ axis.
Lemma~\ref{lemma:rc_good_pairs} shows that, conditioned on~$\virtualInputRV=0$, these two properties hold simultaneously on a constant fraction of the probability mass, up to a correction~$\O(\sqrt{\varepsilon})+\tfrac{1}{\pce}\negl(\security)$.
We retain this good fraction and discard the remaining blocks from the entropy lower bound.
In~\hyperref[par:step_3]{Step~3}, we show that every retained block contributes a constant amount of entropy.

\subsection{Step 3: Entropy from Good Blocks}
\label{par:step_3}

We now show that each of the good blocks identified in Step~2 contributes a constant amount of entropy.
Recall that the constants~$\projThreshold_*$ and~$\angleThreshold_*$, defined in Equation~\eqref{eq:rc_concrete_thresholds}, quantify the two properties of a good block: the state has constant bias with respect to~$\bellMeasurement_0$, while the two Phase~B observables are nearly complementary.
The entropic uncertainty relation of Lemma~\ref{lemma:uncertainty_relations} then implies that the outcome of measuring~$\bellMeasurement_1$ has constant entropy even conditioned on the adversary's arbitrary quantum side information~$\eveReg$.
We first establish this entropy bound for an individual good block and then combine it with the probability mass guaranteed by Lemma~\ref{lemma:rc_good_pairs} to obtain the overall entropy bound.

\begin{lemma}[Entropy on a good block]
	\label{lemma:rc_entropy_good_block}
	Let~$(\trans,\blockIndex)\in\Trans\times\N$ be a transcript--block pair such that
	\begin{equation}
		\imap(\trans)
		=
		0
		\; ,
		\qquad
		\blockIndex
		\in
		\angleBlockSet{\angleThreshold_*}
		\; ,
		\qquad
		(\trans,\blockIndex)
		\in
		\biasPairSet{\projThreshold_*}
		\; .
	\end{equation}
	Then
	\begin{equation}
		\vne\!\left(
			\bellOutcomeRVAt{1}^{(\blockIndex)}
			\mid
			\eveReg
		\right)_{\blockPostPhaseState{\trans}{\blockIndex}}
		\geq
		\goodBlockEntropy
		\; ,
	\end{equation}
	where
	\begin{equation}
		\label{eq:rc_good_block_entropy_constant}
		\goodBlockEntropy
		\coloneqq
		1
		-
		\log\!\left(
			1+10^{-3/2}
		\right)
		-
		h\!\left(
			\frac12
			+
			\frac{9(\sqrt2-1)}{20}
		\right)
		>
		0.057
		\; .
	\end{equation}
	That is, on every good block, the von Neumann entropy of~$\bellOutcomeRVAt{1}^{(\blockIndex)}$, the outcome of measuring~$\bellMeasurement_1^{(\blockIndex)}$, remains bounded below by the positive constant~$\goodBlockEntropy$ even when conditioned on arbitrary quantum side information~$\eveReg$ held by the adversary.
\end{lemma}

\begin{proof}
	Let~$c_{\blockIndex}$ denote the maximum squared overlap between the respective eigenbases of~$\bellMeasurement_0^{(\blockIndex)}$ and~$\bellMeasurement_1^{(\blockIndex)}$.
	Then
	\begin{equation}
		c_{\blockIndex}
		=
		\frac{
			1
			+
			\left|
				\cos\!\left(
					\measurementAngle{\blockIndex}
				\right)
			\right|
		}{2}
		\; .
	\end{equation}
	Since~$\blockIndex\in\angleBlockSet{\angleThreshold_*}$,
	\begin{equation}
		\left|
			\cos\!\left(
				\measurementAngle{\blockIndex}
			\right)
		\right|
		\leq
		\sqrt{\angleThreshold_*}
		=
		10^{-3/2}
		\; .
	\end{equation}
	Therefore,
	\begin{equation}
		\label{eq:rc_good_block_overlap_bound}
		\log(1/c_{\blockIndex})
		\geq
		1
		-
		\log\!\left(
			1+10^{-3/2}
		\right)
		\; .
	\end{equation}

	Since~$(\trans,\blockIndex)\in\biasPairSet{\projThreshold_*}$,
	\begin{equation}
		\left|
			\tr\!\left[
				\bellMeasurement_0^{(\blockIndex)}
				\blockPostPhaseStateP{\trans}{\blockIndex}
			\right]
		\right|
		\geq
		\projThreshold_*
		\; .
	\end{equation}
	Hence
	\begin{equation}
		\label{eq:rc_good_block_first_entropy}
		\vne\!\left(
			\bellOutcomeRVAt{0}^{(\blockIndex)}
		\right)_{\blockPostPhaseStateP{\trans}{\blockIndex}}
		\leq
		h\!\left(
			\frac{1+\projThreshold_*}{2}
		\right)
		=
		h\!\left(
			\frac12
			+
			\frac{9(\sqrt2-1)}{20}
		\right)
		\; .
	\end{equation}

	Applying Lemma~\ref{lemma:uncertainty_relations} inside block~$\blockIndex$ gives
	\begin{align}
		\vne\!\left(
			\bellOutcomeRVAt{1}^{(\blockIndex)}
			\mid
			\eveReg
		\right)_{\blockPostPhaseState{\trans}{\blockIndex}}
		&\geq
		\log(1/c_{\blockIndex})
		-
		\vne\!\left(
			\bellOutcomeRVAt{0}^{(\blockIndex)}
		\right)_{\blockPostPhaseStateP{\trans}{\blockIndex}}
		\nonumber\\
		&\geq
		\goodBlockEntropy
		\; .
	\end{align}
\end{proof}

\subsection{Proof of the Entropy-Certification Theorem}

\begin{proof}[Proof of \entthmref{}]
	We first lower-bound the entropy of~$\bellOutcomeRVAt{1}$, the outcome corresponding to the fixed challenge~$\bellChallenge=1$.
	Since~$\virtualInputRV$ is a deterministic function of~$\transRV$,
	\begin{align}
		\vne\!\left(
			\bellOutcomeRVAt{1}
			\mid
			\transRV,\eveReg
		\right)
		&=
		\sum_{\virtualInput\in\{0,1\}}
		\Pr[\virtualInputRV=\virtualInput]
		\vne\!\left(
			\bellOutcomeRVAt{1}
			\mid
			\transRV,\eveReg,
			\virtualInputRV=\virtualInput
		\right)
		\nonumber\\
		&\geq~
		\Pr[\virtualInputRV=0]
		\vne\!\left(
			\bellOutcomeRVAt{1}
			\mid
			\transRV,\eveReg,
			\virtualInputRV=0
		\right)
		\; .
	\end{align}
	The inequality uses the nonnegativity of the conditional entropy of a classical register given quantum side information.
	By Equation~\eqref{eq:rc_virtual_input_balance},
	\begin{equation}
		\label{eq:rc_entropy_reduce_virtual_zero}
		\vne\!\left(
			\bellOutcomeRVAt{1}
			\mid
			\transRV,\eveReg
		\right)
		\geq
		\frac12
		\vne\!\left(
			\bellOutcomeRVAt{1}
			\mid
			\transRV,\eveReg,
			\virtualInputRV=0
		\right)
		-
		\frac1\pce \negl(\security) \; .
	\end{equation}

	Every projector~$\blockProjector{\blockIndex}$ commutes with~$\bellMeasurement_1$.
	We may therefore refine, without loss of generality, the measurement of~$\bellMeasurement_1$ by first measuring the block~$\blockIndexRV$.
	Giving~$\blockIndexRV$ to the side-information system can only decrease the conditional entropy, and hence
	\begin{equation}
		\label{eq:rc_condition_on_block}
		\vne\!\left(
			\bellOutcomeRVAt{1}
			\mid
			\transRV,\eveReg,
			\virtualInputRV=0
		\right)
		\geq
		\vne\!\left(
			\bellOutcomeRVAt{1}
			\mid
			\blockIndexRV,\transRV,\eveReg,
			\virtualInputRV=0
		\right)
		\; .
	\end{equation}
	Expanding over~$\transRV$ and~$\blockIndexRV$ gives
	\begin{equation}
		\label{eq:rc_entropy_average_blocks}
		\vne\!\left(
			\bellOutcomeRVAt{1}
			\mid
			\blockIndexRV,\transRV,\eveReg,
			\virtualInputRV=0
		\right) =
		\sum_{\substack{\trans,\blockIndex:\\
			\imap(\trans)=0}}
		\Pr\left[
			\transRV=\trans,
			\blockIndexRV=\blockIndex
			\;\middle|\;
			\virtualInputRV=0
		\right]
		\vne\!\left(
			\bellOutcomeRVAt{1}^{(\blockIndex)}
			\mid
			\eveReg
		\right)_{\blockPostPhaseState{\trans}{\blockIndex}}
		\; .
	\end{equation}
	Define
	\begin{equation}
		\goodPairSet
		\coloneqq
		\left\{
			(\trans,\blockIndex):
			\imap(\trans)=0,
			\quad
			\blockIndex\in\angleBlockSet{\angleThreshold_*},
			\quad
			(\trans,\blockIndex)\in\biasPairSet{\projThreshold_*}
		\right\}
		\; .
	\end{equation}
	Restricting the sum in Equation~\eqref{eq:rc_entropy_average_blocks} to~$\goodPairSet$ and applying Lemma~\ref{lemma:rc_entropy_good_block}, we obtain
	\begin{align}
		&\vne\!\left(
			\bellOutcomeRVAt{1}
			\mid
			\blockIndexRV,\transRV,\eveReg,
			\virtualInputRV=0
		\right)
		\nonumber\\
		&\qquad\geq
		\goodBlockEntropy
		\Pr\left[
			(\transRV,\blockIndexRV)
			\in
			\goodPairSet
			\;\middle|\;
			\virtualInputRV=0
		\right]
		\nonumber\\
		&\qquad\geq
		\goodBlockEntropy
		\left(
			\frac{1}{1+10\sqrt2}
			-
			\O(\sqrt{\varepsilon})
		\right)
		-
		\frac1\pce
		\negl(\security)
		\; ,
		\label{eq:rc_entropy_from_good_pairs}
	\end{align}
	where the second inequality follows from Lemma~\ref{lemma:rc_good_pairs}.

	Combining Equations~\eqref{eq:rc_entropy_reduce_virtual_zero}, \eqref{eq:rc_condition_on_block}, and~\eqref{eq:rc_entropy_from_good_pairs} gives the fixed-challenge bound
	\begin{equation}
		\label{eq:rc_fixed_challenge_entropy_corrected}
		\vne\!\left(
			\bellOutcomeRVAt{1}
			\mid
			\transRV,\eveReg
		\right)_{\rho\mid\contEvent}
		\geq
		\frac{\goodBlockEntropy}{2(1+10\sqrt{2})}
		-
		\O(\sqrt{\varepsilon})
		-
		\frac1\pce
		\negl(\security)
		\; .
	\end{equation}
	The challenge~$\bellChallengeRV$ is uniform and independent of the post-Phase-A state.
	Consequently,
	\begin{align}
		&\vne\!\left(
			\bellOutcomeRV
			\mid
			\bellChallengeRV,\transRV,\eveReg
		\right)_{\rho\mid\contEvent}
		\nonumber\\
		&\qquad=
		\frac12
		\vne\!\left(
			\bellOutcomeRVAt{0}
			\mid
			\transRV,\eveReg
		\right)_{\rho\mid\contEvent}
		+
		\frac12
		\vne\!\left(
			\bellOutcomeRVAt{1}
			\mid
			\transRV,\eveReg
		\right)_{\rho\mid\contEvent}
		\nonumber\\
		&\qquad\geq
		\frac{\goodBlockEntropy}{4(1+10\sqrt{2})}
		-
		\O(\sqrt{\varepsilon})
		-
		\frac1\pce
		\negl(\security)
		\; .
	\end{align}
	Therefore, \entthmref{} holds with
	\begin{equation}
		\entropyConst
		=
		\frac{\goodBlockEntropy}{4(1+10\sqrt{2})}
		>
		0
		\; .
	\end{equation}
	Equation~\eqref{eq:rc_fixed_challenge_entropy_corrected} is the stronger fixed-generation-challenge statement that will be used in the entropy-accumulation protocol.
	This completes the proof of \entthmref{}.
\end{proof}

\section{BGKM Protocol}
\label{app:bgkm_protocol}

\subsection{Protocol Description}
\label{app:bgkm_protocol_description}

\renewcommand{\image}{z}
\begin{algorithm}[H]
	\caption{BGKM protocol~\cite{brakerski2023simple}}
	\label{protocol:bgkm_protocol}

	\vspace{0.25cm}

	Fix a security parameter~$\security\in\N$ and a \tcf{} family~$\Tcfunc$.

	\vspace{0.25cm}

	\textbf{Phase A}

	\vspace{-0.15cm}

	\begin{enumerate}
		\item
		The verifier samples a key and trapdoor
		\begin{equation}
			(\key,\trapdoor)
			\leftarrow
			\gen(1^\security)
			\; ,
		\end{equation}
		from a \tcf{} family~$\Tcfunc$, and sends~$\key$ to the prover.

		\item
		The prover sends an image~$\image$ to the verifier.

		\item
		The verifier uniformly samples
		\begin{equation}
			\rand_0,\rand_1
			\leftarrow
			\{0,1\}^{n}
			\; ,
		\end{equation}
		and sends~$(\rand_0,\rand_1)$ to the prover.

		\item
		The prover sends a string~$\hadx\in\{0,1\}^{n}$ to the verifier.

		\item
		The transcript is
		\begin{equation}
			\trans
			\coloneqq
			(\key,\image,\hadx,\rand_0,\rand_1)
			\; .
		\end{equation}
	\end{enumerate}

	\vspace{-0.15cm}

	\textbf{Phase B}

	\vspace{-0.15cm}

	\begin{enumerate}[resume]
		\item
		The verifier uniformly samples
		$\bellChallenge\leftarrow\{0,1\}$
		and sends~$\bellChallenge$ to the prover.

		\item
		The prover sends a bit~$\bellOutcome\in\{0,1\}$ to the verifier.
	\end{enumerate}
\end{algorithm}

\subsection{Honest Quantum Implementation}
\label{app:bgkm_honest_implementation}

The description above specifies only the interaction between the verifier and an arbitrary prover, without imposing any particular internal implementation on the prover.
For completeness, and because the resulting hidden qubit is used explicitly in the main construction, we next describe the corresponding honest quantum implementation.
In particular, this implementation makes explicit how the honest prover's state evolves during Phase~A and how the residual qubit on which the Phase-B measurement is performed is prepared.

\addtocounter{algorithm}{-1}
\begin{algorithm}[H]
	\caption{BGKM protocol: honest implementation~\cite{brakerski2023simple}}
	\label{protocol:bgkm_protocol_honest}
	\vspace{0.25cm}

	Fix a security parameter~$\security\in\N$ and a \tcf{} family~$\Tcfunc$.

	\vspace{0.25cm}
	\textbf{Phase A}
	\begin{enumerate}
		\item $\verifier$ samples a key and trapdoor~$(\key,\trapdoor)\leftarrow\gen(1^\security)$ from a \tcf{} family~$\Tcfunc$, and sends~$\key$ to~$\prover$.

		\item $\prover$ prepares the state
		\begin{equation}
			\sum_{\bit\in\{0,1\}}
			\sum_{\preimage\in\Preimage}
			\ket{\bit}
			\ket{\preimage}
			\ket{\tcfunc_{\key,\bit}(\preimage)}
			\; .
		\end{equation}
		It measures the last register in the computational basis, obtaining~$\image$, and sends~$\image$ to~$\verifier$.
		Denoting by~$\preimage_0,\preimage_1$ the two preimages associated with~$\image$, the remaining state is
		\begin{equation}
			\ket{0}\ket{\preimage_0} +
			\ket{1}\ket{\preimage_1} \; .
		\end{equation}

		\item $\verifier$ uniformly samples~$\rand_0,\rand_1\leftarrow\{0,1\}^n$, and sends~$(\rand_0,\rand_1)$ to~$\prover$.

		\item $\prover$ coherently computes the inner product, obtaining
		\begin{equation}
			\ket{0}
			\ket{\preimage_0}
			\ket{\preimage_0\cdot\rand_0}
			+
			\ket{1}
			\ket{\preimage_1}
			\ket{\preimage_1\cdot\rand_1} \; .
		\end{equation}

		\item $\prover$ uncomputes the first register and measures the preimage register in the Hadamard basis, obtaining a string~$\hadx\in\{0,1\}^n$
		which it sends to~$\verifier$.
		The remaining qubit is
		\begin{equation}
			\ket{\protocolState_\trans}
			\propto
			\ket{\preimage_0\cdot\rand_0}
			+(-1)^
			{
				\hadx\cdot
				(\preimage_0\oplus\preimage_1)
			}
			\ket{\preimage_1\cdot\rand_1}
			\; .
		\end{equation}

		\item The transcript becomes
		\begin{equation}
			\trans
			\coloneqq
			(
				\key,
				\image,
				\hadx,
				\rand_0,
				\rand_1
			)
			\; .
		\end{equation}
	\end{enumerate}

	\textbf{Phase B}
	\begin{enumerate}[resume]
		\item $\verifier$ uniformly samples
		$\bellChallenge\leftarrow\{0,1\}$
		and sends~$\bellChallenge$ to~$\prover$.

		\item $\prover$ performs the measurement
		\begin{equation}
			\bellMeasurement_{\bellChallenge}
			=
			\left(\sigmaz + (-1)^{\bellChallenge}\sigmax\right) / \sqrt{2}
			\; ,
		\end{equation}
		on its state~$\ket{\protocolState_\trans}$ and sends the result~$\bellOutcome$ to~$\verifier$.
	\end{enumerate}
\end{algorithm}

\begin{remark}[Uncomputing the branch register]
	By Definition~\ref{def:bit_commitment}, the opening verification algorithm~$\open$ is a deterministic polynomial-time classical algorithm.
	For any fixed commitment transcript~$\bctrans$, define
	\begin{equation}
		v_1(\bctrans,\decomm)
		\coloneqq
		\indicator{
			\left[
				\open(\bctrans,1,\decomm)=\top
			\right]
		}
		\; .
	\end{equation}
	The prover can implement~$v_1$ reversibly as
	\begin{equation}
		\ket{b}\ket{\decomm}
		\longmapsto
		\ket{b\oplus v_1(\bctrans,\decomm)}
		\ket{\decomm}
		\; .
	\end{equation}
	For the two valid openings~$\decomm_0$ and~$\decomm_1$, \pq{} computational binding implies that, except with negligible probability,~$\decomm_0$ is not also a valid opening to~$1$; hence
	\begin{equation}
		v_1(\bctrans,\decomm_0)=0
		\qquad\text{and}\qquad
		v_1(\bctrans,\decomm_1)=1.
	\end{equation}
	Therefore,
	\begin{equation}
		\ket{0}\ket{\decomm_0}
		+
		\ket{1}\ket{\decomm_1}
		\longmapsto
		\ket{0}
		\left(
			\ket{\decomm_0}
			+
			\ket{\decomm_1}
		\right),
	\end{equation}
	so the branch register factors from the decommitment register and can be discarded.
\end{remark}

\section{An Auxiliary Inequality}

\begin{lemma}[Square-root linearization]
	\label{lemma:sqrt_linearization}
	Let~$b>0$ and~$r\geq0$.
	Then,
	\begin{equation}
		\sqrt{b+r}
		\leq
		\sqrt{b}
		+
		{r}/{(2\sqrt{b})}
		\; .
	\end{equation}
\end{lemma}

\begin{proof}
	We have
	\begin{align}
		\sqrt{b+r}
		-
		\sqrt{b}
		&=
		\frac{
			(b+r)-b
		}{
			\sqrt{b+r}
			+
			\sqrt{b}
		}
		\nonumber\\
		&=
		\frac{
			r
		}{
			\sqrt{b+r}
			+
			\sqrt{b}
		}
		\; .
	\end{align}
	Since~$r\geq0$,
	\begin{equation}
		\sqrt{b+r}
		\geq
		\sqrt{b}
		\; ,
	\end{equation}
	and hence
	\begin{equation}
		\sqrt{b+r}
		+
		\sqrt{b}
		\geq
		2\sqrt{b}
		\; .
	\end{equation}
	Therefore,
	\begin{equation}
		\sqrt{b+r}
		-
		\sqrt{b}
		\leq
		{r}/{2\sqrt{b}}
		\; .
	\end{equation}
\end{proof}

\newpage
\bibliographystyle{alpha}
\bibliography{bib}

\end{document}